\documentclass[11pt,reqno]{amsart}
\usepackage{amssymb}
\usepackage{enumitem}
\usepackage{amsthm,amsfonts,amssymb,euscript,color,bbm}
\usepackage{comment}
\usepackage{mathtools}
\usepackage{stmaryrd}
\usepackage{mathrsfs}
\usepackage{pifont}
\usepackage{amsfonts,amsmath,amssymb,amsthm,amscd,cancel}
\usepackage{graphicx}
\usepackage[]{geometry}
\usepackage{verbatim}
\usepackage{mathrsfs}
\usepackage{fancyhdr}
\usepackage{footnpag,footmisc}
\usepackage[normalem]{ulem}

\usepackage{float}

\usepackage{cite}
\usepackage{hyperref}
\usepackage{engord}
\usepackage[displaymath]{lineno}

\theoremstyle{definition}
\newtheorem{theorem}{Theorem}[section]
\newtheorem{lemma}{Lemma}[section]
\newtheorem{proposition}{Proposition}[section]

\newtheorem{remark}{Remark}[section]

\newtheorem{boot}{Bootstrap Assumption}

\allowdisplaybreaks[4]

\DeclareMathAlphabet{\mathsfsl}{OT1}{cmss}{m}{sl}
\numberwithin{equation}{section}

\newcommand{\D}{\mathrm{d}}

\newcommand{\tr}{\mathrm{tr}}

\def\alphab{\underline{\alpha}}
\def\betab{\underline{\beta}}
\def\chib{\underline{\chi}}
\def\chibh{\widehat{\underline{\chi}}}
\def\chih{\widehat{\chi}}
\def\etab{\underline{\eta}}

\def\Lb{\underline{L}}
\def\mub{\underline{\mu}}

\def\tr{\mathrm{tr}}
\def\omegab{\underline{\omega}}

\def\tensor{\widehat{\otimes}}

\def\ub{\underline{u}}
\def\Cb{\underline{C}}

\def\Lh{\widehat{L}}
\def\Lbh{\widehat{\underline{L}}}

\def\Xib{\underline{\Xi}}
\def\Ib{\underline{I}}
\def\Lambdab{\underline{\Lambda}}

\def\Thetab{\underline{\Theta}}
\def\Kb{\underline{K}}
\def\Dbf{\mathbf{D}}

\newcommand{\Db}{\underline{D}}
\newcommand{\Dh}{\widehat{D}}
\newcommand{\Dbh}{\widehat{\underline{D}}}

\newcommand{\gammat}{\widetilde{\gamma}}

\def\nablas{\mbox{$\nabla \mkern -13mu /$ }}
\def\Deltas{\mbox{$\Delta \mkern -13mu /$ }}

\def\divs{\mbox{$\mathrm{div} \mkern -13mu /$ }}
\def\curls{\mbox{$\mathrm{curl} \mkern -13mu /$ }}
\def\omegas{\mbox{$\omega \mkern -13mu /$ }}
\def\omegabs{\mbox{$\omegab \mkern -13mu /$ }}
\def\ds{\mbox{$\nabla \mkern -13mu /$ }}
\def\gs{\mbox{$g \mkern -9mu /$}}
\def\epsilons{\mbox{$\epsilon \mkern -9mu /$}}

\def\Lie{\mbox{$\mathcal{L} \mkern -10mu/$}}

\def\Us{\mbox{$U \mkern -13mu /$ }}
\def\Js{\mbox{$J \mkern -11mu /$ }}

\def\Es{\mbox{$\mathcal{E} \mkern -11mu /$}}

\def\Ks{\mbox{$K \mkern -13mu / $}}
\def\is{\mbox{$i \mkern -8mu /$}}
\def\js{\mbox{$j \mkern -8mu /$}}
\def\Rics{\mbox{$\mathbf{Ric} \mkern -17mu /\ \ $}}
\def\Ricsef{\mbox{$\mathbf{Ric_4} \mkern -17mu /\ $}}
\def\Ricset{\mbox{$\mathbf{Ric_3} \mkern -17mu /\ $}}

\begin{document}

\title[Instability of naked singularities in Einstein--Euler]{Instability of naked singularities of perfect fluid under $C^{1,\alpha}$ gravitational perturbations}

\author[Junbin Li]{Junbin Li}
\address{Department of Mathematics, Sun Yat-sen University, Guangzhou, China}
\email{lijunbin@mail.sysu.edu.cn}
\author[Tingting Li]{Tingting Li}
\address{Department of Mathematics, Sun Yat-sen University, Guangzhou, China}
\email{litt76@mail2.sysu.edu.cn}
\author[Xi-Ping Zhu]{Xi-Ping Zhu}
\address{Department of Mathematics, Sun Yat-sen University, Guangzhou, China}
\email{stszxp@mail.sysu.edu.cn}
\thanks{All authors are supported by National Key R\&D Program of China (No. 2022YFA1005400). J. Li and T. Li are also supported by NSFC (12326602, 12141106). }

 \begin{abstract}
  We study the instability of the naked singularities arising in the spherically symmetric self-similar collapsing of the Einstein--Euler system under gravitational perturbations. We show that small $C^{1,\alpha}$ perturbations (without any symmetries) of the initial conformal metric lead to trapped surface formation. One key point is that the speed of sound is strictly slower the the speed of light,  so the gravitational radiation can become sufficiently concentrated to form a trapped surface before the fluid develops any singularity.

  \end{abstract}

\maketitle

\tableofcontents

\setcounter{tocdepth}{1}


\section{Introduction}

\subsection{Previous works and the main results}

The weak cosmic censorship conjecture in general relativity asserts that naked singularities will not appear in gravitational collapse generically. The original statement proposed by Penrose \cite{Pen69} asserts that there are no naked singularities arising in gravitational collapse, but after examples of naked singularities in spherically symmetric scalar field collapse (for example \cite{Cho93, Chr94}) were found, people then tried to show that naked singularities are unstable under small perturbations.

\begin{figure}
\includegraphics[width=4.5 in]{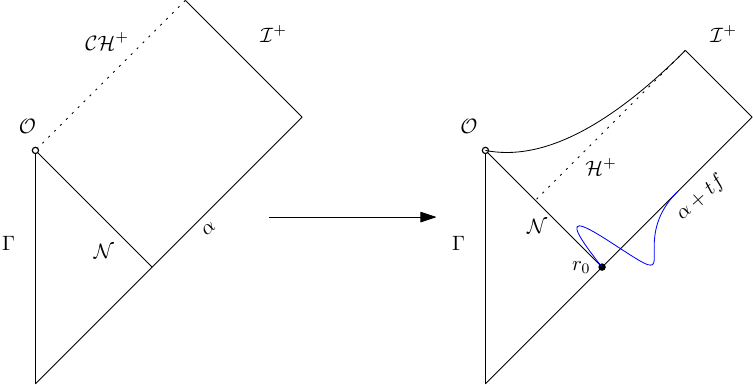}
\caption{\label{figure:nakedperturbation}}
\end{figure}

Soon after his construction of naked singularities, Christodoulou \cite{Chr99} was able to show that general naked singularities are unstable in the context of  spherically symmetric solutions of Einstein--scalar field equations. What was done in \cite{Chr99} is the following.  The left diagram in Figure \ref{figure:nakedperturbation} refers to a naked singularity solution. It is the maximal development of some initial data, say $\alpha$, given on an outgoing cone of a point at the central line $\Gamma$. In spherically symmetric Einstein--scalar case, $\alpha$ is the derivative $\partial_r$ along the null generator of the initial cone of the function $r\phi$ where $\phi$ is simply the scalar function of the scalar field. Christodoulou showed \cite{Chr93} that such initial value problem is well-posed for $\alpha\in BV[0,+\infty)$ and the maximal development lies in the class of solutions, so-called BV solutions. Let $\mathcal{O}$ be the first singularity on $\Gamma$. The nature of a naked singularity solution is that the future null cone $\mathcal{CH}^+$ (called Cauchy horizon) of the singularity is regular and extends to the future null infinity (see \cite{R-Sh23} for an accurate  definition for naked singularities). The spacetime can be extended beyond $\mathcal{CH}^+$ verifying the Einstein--scalar equations but cannot extend as a solution of the initial value problem of $\alpha$. In \cite{Chr99}, under the assumption that $\mu=\frac{2m}{r}$ does not tend to zero along the past null cone $\mathcal{N}$ of $\mathcal{O}$ as approaching $\mathcal{O}$ (which is a condition connected to the blow up criterion established in \cite{Chr93}), Christodoulou constructed a function $f\in BV$, depending on the geometry of $\mathcal{N}$ (in fact, the construction of $f$ does not depend on whether the singularity $\mathcal{O}$ is naked or not), so that the maximal developments of $\alpha+tf$ contain no naked singularities and a black hole eventually form (the right diagram in Figure \ref{figure:nakedperturbation}), for any $t\ne0$.  More accurately, these maximal developments can be proved to have a complete future null infinity (see \cite{Chr99cqg}). An--Tan \cite{A-T24} studied similar questions for the gravitational collapse of a charged scalar field. 

As the above perturbations are from scalar field, it is natural to ask what will happen under gravitational perturbations, which is however not spherically symmetric by Birkhoff theorem. The first author and Liu started this program  in \cite{Liu-Li18} by giving another constructive argument based on a priori estimates, instead of the original contradiction argument by Christodoulou \cite{Chr99}.  We further showed that  \cite{Li-Liu22}, also assuming $\mu\nrightarrow0$ along $\mathcal{N}$,  we were able to construct $C^1$ gravitational perturbations such that the maximal developments contain a closed trapped surface.  We remark that outside spherical symmetry, it is not known whether the existence of such a closed trapped surface would lead to a black hole (i.e., complete future null infinity). In view of this, Christodoulou proposed \cite{Chr99cqg} that one can first produce closed trapped surfaces by adding small perturbations as a first step towards the weak cosmic censorship conjecture, called the trapped surface conjecture, a local version of the weak cosmic censorship conjecture. Recently, An \cite{An24} showed that when the background naked singularity solutions are those constructed by Christodoulou in \cite{Chr94}, which are continuously self-similar (called $k$-self-similar naked singularities, as a consequence, $\mu=\text{const.}$ along $\mathcal{N}$), fully anisotropic apparent horizon (not only trapped surfaces) can form under small gravitational perturbations in scale--critical norm. Also see some construction and instability works on non-spherically symmetric background \cite{An2, An3}. 

It is important to note that in all instability results above, the perturbations we put in are supported in the exterior region, which is the region to the future of $\mathcal{N}$. Such perturbations will not change the past of the singularity, called the interior region. Working in spherical symmetry, for general naked singularity solutions, the first author \cite{Li25} was able to construct interior perturbations still leading to black hole formation, so that the supports of the perturbations intersect the interior region. This provides an essentially new insights into understanding the weak cosmic censorship in this model.  Moreover, if the background naked singularity is $k$-self-similar, for both exterior and interior perturbations, the regularity of the perturbations can be made arbitrarily close to the threshold regularity from below. The construction of interior perturbations do not depend of the geometry of $\mathcal{N}$, but instead depend on the geometry of the incoming null cones before and close to $\mathcal{N}$. Such an interior mechanism was generalized to non-spherically symmetric gravitational perturbations in \cite{Li-Li26},  where the same threshold regularity was recovered for genuinely anisotropic perturbations whose angular support shrinks to a point. 

In this paper we focus on the instability of naked singularities of another system, namely the Einstein--Euler system, which were first found by Ori--Piran (see \cite{O-P90} and the references therein) by numerical method, assuming spherical symmetry and continuous self-similarity of the spacetime. The rigorous proof of such construction was provided recently by Guo--Hazic--Jang \cite{G-H-J23}. The equation of state is the linear one $p=\kappa\varrho$ where $\kappa$ is a sufficiently small positive constant. Significantly different from the Einstein--scalar field case, the instability of naked singularities of such Einstein--Euler system is not known in even in spherical symmetry. On the other hand, the non-relativistic version of such singular solutions (called Larson-Penston collapse in Euler--Poisson system) are shown to be nonlinearly stable under radial perturbations \cite{GHJS25}, which suggests that the stability/instability of the Einstein--Euler naked singularities in spherical symmetry (or more precisely, under perturbations from fluid) is a delicate issue.  But as a couple system, we have additional  degrees of freedom from gravity side, which must however be non-spherically symmetric. We will also show that if we go beyond spherically symmetric, there are more rooms for us to create instability, which is the main novelty  of this paper. 

In this paper, we consider the following Einstein--Euler system 
\begin{equation}\label{EinsteinEuler}\mathbf{Ric}_{\alpha\beta}-\frac{1}{2}\mathbf{R}g_{\alpha\beta}=2\mathbf{T}_{\alpha\beta}\end{equation}
where the energy momentum tensor  of perfect fluid  is
\begin{equation}\label{fluidem}\mathbf{T}_{\alpha\beta}=(\varrho+p) U_\alpha U_\beta+pg_{\alpha\beta}=(1+\kappa)\varrho U_\alpha U_\beta+\kappa\varrho g_{\alpha\beta}.\end{equation}
Here the equation of state for the fluid is the isothermal one: $p=\kappa\varrho$ where $\kappa\in(0,1)$ is a constant, the square of sound speed. $\varrho,p, U$ are respectively the density, pressure and $4$-velocity of the fluid, and $U^\alpha U_\alpha=-1$. The Euler system reads
\begin{equation}\label{Euler}
\begin{split}\nabla_U\varrho+(1+\kappa)\varrho\nabla_\alpha U^\alpha=0,\\
 (1+\kappa)\varrho \nabla_U U+\kappa(\nabla\varrho+ U\nabla_U\varrho)=0,
 \end{split}
 \end{equation}
 which is equivalent to conservation law $\nabla^\alpha\mathbf{T}_{\alpha\beta}=0$.  A rough version of the main result of this paper is the following.
\begin{theorem}\label{roughmain}
The spherically symmetric and continuously self-similar naked singularity solutions constructed in \cite{O-P90, G-H-J23}, solutions of the Einstein--Euler system \eqref{EinsteinEuler}, are $C^{1,\alpha}$  unstable to trapped surface formation under gravitational perturbations.
\end{theorem}
\begin{remark}
This theorem in fact applies to much more general classes of naked singularity solutions of isothermal perfect fluid.  The background solution is only required to satisfy certain self-similar bounds on its past null cone $\mathcal{N}$ of the singularity. 
 \end{remark}
\begin{remark}
Here  $C^{1,\alpha}$ is the regularity of the departure of the perturbed initial conformal metric from the original one of the naked singularity solution along outgoing null direction. Different from the continuously self-similar naked singularity solutions of Einstein--scalar, the naked singularity solutions of Einstein--Euler are smooth across both the past sound cone and the past null cone of the singularity. So it makes sense that we construct smooth perturbations. Nevertheless, the size of the perturbations tends to zero only in $C^{1,\alpha}$. 
 \end{remark}
 \begin{remark}
The perturbations studied in this paper are exterior perturbations, that is, supported in the future of the past null cone $\mathcal{N}$ of the singularity. The perturbations of the initial conformal metric in fact converge to the original one in $C^\infty$ away from the intersection of $\mathcal{N}$ and the initial null cone. One can also study interior perturbations in the Einstein--Euler system as in \cite{Li25}.  One can also study interior anisotropic perturbations with shrinking angular supports as in \cite{Li-Li26}. But in all these settings, the instability we create relies on the geometry of $\mathcal{N}$ or null cones close to and before $\mathcal{N}$, small perturbations from fluid cannot create a closed trapped surface because the sound speed $\sqrt{\kappa}$ is strictly less than the light speed $1$. It would be interesting to find small perturbations from fluid creating a closed trapped surface. 
 \end{remark}
\begin{remark}
As the original naked singularity solution of Einstein--Euler is smooth, one may be interested in finding $C^\infty$ or at least $C^2$ perturbations leading to instability. This is still open even in the Einstein--scalar field system, in which case conjecturally there exists a smooth naked singularity solution (critical collapse found numerically by Choptuik \cite{Cho93}, which is discretely self-similar), and there exists smooth perturbations to it leading to black hole formation or even dispersion.  Go back to the low regularity instability established in Theorem \ref{roughmain}, although not being smooth,  it still makes sense because it is conjectured (see \cite{R-Sh23}, in vacuum) that initial data in this topology admit local well-posedness across the vertex. We remark that, which can be seen from the proof,  under the $C^{1,\alpha}$ perturbations from initial conformal metric, the departure of the resulting fluid density and velocity from the original one is at least $C^{2,\widetilde{\alpha}}$ along outgoing null direction for a (potentially) different $\widetilde{\alpha}$.
 \end{remark}
\begin{remark}
In a companion paper \cite{L-Z25}, we studied Einstein--Euler--scalar field system in spherical symmetry as a toy model of Theorem \ref{roughmain}. In that model, the scale function $\phi$ plays the role of the metric. 
 \end{remark}

\subsection{The set up of the problem, the precise statement}
Let us take $(\mathcal{M},g,\varrho,U)$ to be the spherically symmetric  continuously self-similar naked singularity solution of Einstein--Euler system. Let $\mathcal{N}$ be the past null cone of the singularity $\mathcal{O}$, and $\mathcal{C}$ be the part to the future of $\mathcal{N}$, of a outgoing null cone from a point at the central line of $\mathcal{M}$. Let $(\ub,u)$ be a double null coordinate around $\mathcal{N}$, so that $\ub=0$ on $\mathcal{N}$, $u=u_0$ on $\mathcal{C}$, and $u=-r$ on $\mathcal{N}$ where $r$ is the area redius. If we denote $C_{u}$ be the level set of $u$ and $\Cb_{\ub}$ be the level set of $\ub$, then $\mathcal{N}=\Cb_0$ and $\mathcal{C}=C_{u_0}$. We will consider the double null initial value problem on $C_{u_0}\cup\Cb_0$, where the data on $\Cb_0$ is induced from the naked singularity solution and the data on $C_{u_0}$ is the perturbation we put in the spacetime. 

The data on $\Cb_0$ is in particular spherically symmetric, therefore the only nontrivial components of connection coefficients and curvature components\footnote{See Appendix \ref{appendix} for notations and equations.} are $\tr\chi, \tr\chib, \omega, \omegab$ and $\rho$. By the choice $u=-r$, we have on $\Cb_0$, 
$$\Omega\tr\chib=-\frac{2}{r}.$$
Since $\Cb_0$ is the level set of the self-similar parameter in the construction of continuously self-similar solution, and $L'$ is the normalized pair of $\Lb$ ($g(\Lb, L')=-2$), then on $\Cb_0$,
$$r\tr\chi', r\omegab, r^2\rho=\text{const.}$$
and the value of $\omega$ is not needed. For the fluid density and velocity, we will have, along $\Cb_0$ (which can also be seen in \cite{G-H-J23}),
$$r^2\varrho=\text{const.}, \ U_{\Lb}=\text{const.},$$
where $U_X:=g(U,X)$ for a vectorfield $X$. , we have
$$U_{L'}=\text{const.}$$
from the condition $U^\alpha U_\alpha=-1$. Moreover, $\Lb\varrho$, $\Lb U_{\Lb}, \Lb U_{L'}$ and $L'\varrho, L'U_{\Lb}, L'U_{L'}$ and higher order derivatives are also constants. {\bf In this paper, we only require that the above dimensionless quantities have an upper and lower bounds away from zero, and their derivatives related to $\Lb$ or $L'$ are also bounded up to sufficiently high order.} We remark that these bounds is just a rough description which is sufficiently for our theorem.  In the course of the proof we will need more explicit bounds.

Because $\Db\log\Omega=\omegab$, we know that 
\begin{equation}\label{Omegaupperlower}
c_1|u|^{\beta_1}\le \Omega_0(u)\le c_2|u|^{\beta_2}
\end{equation}
for some $c_1,c_2,\beta_1,\beta_2>0$ and $\Omega_0(u)=\Omega(0,u)$, the restriction of the lapse on $\Cb_0$. Moreover, plugging in $\Omega\tr\chib=-\frac{2}{r}$ in \eqref{Dbtrchib}, we will see that
$$\omegab=-\frac{1}{4}r\mathbf{Ric}_{\Lb\Lb}$$
and hence $\omegab$ is negative (constant in exactly self-similar).

Readers may refer to \cite{L-Z25} for more discussions on the condition \eqref{Omegaupperlower}, which holds for quite general continuously self-similar spacetime. The condition \eqref{Omegaupperlower} is related to the strength of the singularity \cite{Tip77, Cl-Kr85}, in which the singularity is said to be gravitationally strong. Therefore, the instability results in this paper supports a conjecture formulated by Newman \cite{Newman86} that naked singularities appearing in gravitational collapse are in some sense gravitationally weak. 

The data on $C_{u_0}$ is the perturbations we put in. We assume that $\ub$ on $C_{u_0}$ is the affine parameters, then $\Omega$ is constant along $C_{u_0}$ and $\omega=0$. Then the initial data consists of the shear tensor $\chih$ and the fluid variables.  We simply assume that the fluid variables on $C_{u_0}$ have sufficiently many bounded derivatives (both tangent and transversal to $C_{u_0}$). Since the data on $C_{u_0}\cap\Cb_0$ is spherically symmetric, then angular derivatives of fluid variables are bounded at least by $a\ub$ for some constant $a>0$. Before perturbation, the naked singularity solution is spherically symmetric so the shear tensor vanishes. The gravitational perturbations we put in are essentially the new shear $\chih$.

The main theorem of this paper is the following.

\begin{theorem}\label{main}
Consider double null characteristic problem of the system \eqref{EinsteinEuler}. Suppose that the data on the incoming null cone $\Cb_0$ and the fluid data on $C_{u_0}$ are described above.  Then there is a one-parameter family of initial shear $\chih_t$ (for $t$ sufficiently small) on $C_{u_0}$, such that 
$$\chih_t\to0, t\to0$$
in $C^{\alpha}_{\ub} H^{s}_\vartheta$  for some $\alpha>0$ and sufficiently large integer $s$,  and the maximal future development of such initial data on $C_{u_0}$ has a closed trapped surface  for $t\ne0$.
\end{theorem}

\begin{figure}[H]
\includegraphics[width=2.5 in]{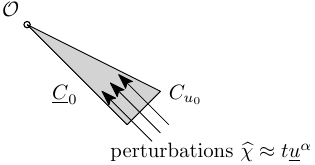}
\caption{\label{figure:perturbations}}
\end{figure}
\begin{remark}
As in \cite{L-Z25}, the perturbations are in fact generic in the following sense: the set of all data $\chih$ on $C_{u_0}$ such that there is a closed trapped surface preceding $\mathcal{O}$ in its maximal future development, contains an open set of which zero data (leading to naked singularity solution) is its limit point in $C^\alpha$ topology.
\end{remark}

\begin{remark}
It will be clear in the proof that $\chih$ has the form $t\ub^{\alpha}$ near $\ub=0$. By cutting off we will get smooth $\chih$. But if we allow $\chih$ to be non-smooth at $\ub=0$, then we can find a sequence of closed trapped surfaces approaching the singularity $\mathcal{O}$, which is similar to \cite{Chr99, Liu-Li18, Li-Liu22} and an apparent horizon emerging from the $\mathcal{O}$ may also be located as in \cite{An24}.
 \end{remark}

\subsection{Proof ingredients} As in our previous works \cite{Liu-Li18, Li-Liu22, L-Z25} and An's work \cite{An24}, the naked singularity  (exterior) instability problem is essentially a trapped surface formation problem, started by Chrisotodoulou \cite{Chr08}, and also studied, for example in \cite{K-R12, An-Luk17} in vacuum. The difference is that in the trapped surface formation problem, the incoming cone $\Cb_0$ is Minkowskian, and in instability problem, the incoming cone in singular, which was first studied in \cite{Li-Liu22} for Einstein--scalar field system.  The basic strategy is however the same: prove an a priori estimate to guarantee the existence of the Einstein equations near $\Cb_0$ under certain initial bounds, and then find closed trapped surface by using Raychaudhuri equation, under additional assumption on an energy lower bound. When the incoming cone is singular, it turns out that on one hand we only need a much milder lower bound for trapped surface formation, on the other hand we have to do much delicater a priori estimates to get the ``long-time'' existence. 

Let us first discus the estimates on the Einstein quantities. 

{\bf Blue-shift estimates for the Einstein quantities.} In self-similar regime (see \cite{R-Sh23}, in which closed trapped surfaces do not form, or see an exposition in \cite{JS24}),  all dimensionless quantities are expected be bounded. For example, in the region $\ub|u|^{-1}\ll1$,  one should have
$$|u|\Omega\chibh, |u|\Omega\tr\chib, |u|\chi', |u|\tr\chi'$$
is bounded, since they are the shears and expansions related to normalized pair $\Lb$ and $L'$, where $\Lb=-\partial_r$ on $\Cb_0$.  If Euler is included, then
$$|u|^2\varrho, U_{\Lb}, U_{L'}$$
are also expected to be bounded (above and below away from zero). Also, $|u|\Db$ and $|u|\Omega^{-2}D$ derivatives are expected to ``bounded''. We may call such bounds related to $\Lb$ and $L'$  the self-similar bounds.

 However, in instability regime, in order that closed trapped surfaces can form, the Einstein quantities (i.e., the connection coefficients and curvature components)  related to $\Lb$ can also be proved to be bounded, but instead of those related to $L'$ are not. Only those related to $L$, for example,
$$|u|\Omega\chih, |u|\Omega\tr\chi,$$
can be proved to be bounded. This can be viewed from the equation \eqref{Dbchih}, in which the right hand side contains no terms like $\omegab\Omega\chih$, which would cause a logarithm loss after integration along $\Db$. As a consequence, $\Dbh(|u|\Omega\chih)$ is expected to be integrable in $u\in [u_0,0)$. Integrating along $u$, we will have
$$|u|\Omega\chih-|u_0|\Omega\chih\big|_{C_{u_0}}\approx f(u),$$
where $f(u)$ has a limit as $u\to0$, and then $|u|\Omega\chih$ can be shown to be bounded. We may call the bounds related to $L$ the blue-shift bounds.  But if we pursuit self-similar bounds, since we have
$$|u|\chih'\approx \Omega^{-2}(|u_0|\Omega\chih\big|_{C_{u_0}}+f(u)),$$
and $\Omega\approx\Omega_0\to0$ as $u\to0$, we will need a careful choice of initial data on $C_{u_0}$ so that the right hand side is still bounded,  which is \textit{not generic}. 

\begin{remark}\label{equationwithOmega}
 For the other connection coefficients, we will show that
$$|u|\Omega\eta,|u|\Omega\etab, |u|\omega, |u|\omegab$$
are bounded, and for curvature components, we will show that
$$|u|^2\Omega^2\alpha,|u|^2\Omega^2\beta,|u|^2\Omega^2\rho,|u|^2\Omega^2\sigma,|u|^2\Omega^2\betab,|u|^2\Omega^2\alphab,$$
are (suitably) bounded. Different to the previous work \cite{Li-Liu22}, we have an additional $\Omega$ factor on each connection coefficient ($\omega,\omegab$ already have an additional $\Omega$ in their definitions) and $\Omega^2$ on each curvature component. The advantage of this formulation is that the null structure equations and Bianchi equations have no individual $\Omega$ factors, and hence the estimates can be treated more systematically. See more detailed discussions on this point in \cite{Li-Li26}.
\end{remark}

We next discuss the estimates on the Euler quantities.

{\bf Self-similar and blue-shift estimates for Euler.} In the presence of Euler, we should figure out whether the above estimates for the Einstein quantities still hold and what are the correct bounds for the fluid variables.  In view of our previous work \cite{L-Z25}, we can expect that the following fluid variables
$$|u|^2\varrho, U_{\Lb}, U_{L'}$$
are still bounded, but the only the derivatives $|u|D$ (instead of $\Omega^{-2}D$) and $|u|\Db$ are bounded. This is of course a consequence of different behaviors of the Einstein quantities. But it is very important that the zeroth order estimates of the fluid variables still obey self-similar estimates. Let us look at the the first equation of the Euler equations \eqref{Euler}, and write 
$$0=U\varrho+(1+\kappa)\varrho \nabla_\alpha U^\alpha=U\varrho-\frac{1}{2}\varrho(1+\kappa)\tr\chi' U_{\Lb}+\cdots.$$
Based on the above discussions, $|u|\tr\chi' U_{\Lb}$ can only be bounded by $\Omega^{-2}$, but fortunately, the length of the integral curves of $U$ in the region $\ub|u|^{-1}\ll1$ is approximately $\Omega^2\ub$, as shown in \cite{L-Z25}, which tells us that $|u|^2\varrho$ can be shown to be bounded. Now this Euler equation again tells us that $|u|^3U\varrho$ can only be bounded by $\Omega^{-2}$, as we will show that $|u|^3\Lb\varrho$ is bounded, then write $U$ as a linear combination of $\Lb,L'$ (and tangential components), we have
$$|u|L'\approx |u|U+|u|\Lb\lesssim\Omega^{-2}$$
which shows that $|u|L \lesssim1$, as claimed above.  Once of fluid variables are bounded correctly, we find that the estimates for Einstein quantities follow by energy estimate procedure in double foliation, which is somehow standard nowadays. 
\begin{remark}\label{DbintermsofDlossOmega}
One may also estimate $\Lb$ derivative in terms of $L$ derivative as
$$|u|\Lb\approx |u|U+|u|L'.$$
But we can only hope to prove $|u|L\lesssim 1$ (and also $|u|U\lesssim \Omega^{-2}$), so we can only have $|u|\Lb\lesssim\Omega^{-2}$ which is worse than expected.  In view of this, we prefer to first estimate $\Lb$ derivative (by commuting $\Lb$ to the Euler equations) and then express $L$ in terms of $\Lb$.
\end{remark}

To realize the estimates for the fluid variables described above, in spherically symmetric toy model \cite{L-Z25},   we rewrite the equations as ODEs along both characteristics, and then simply integrate along the characteristics, but this fails outside spherical symmetry. For this we will need a method of energy estimates for Euler equations, which is the energy currents introduced by Christodoulou \cite{Chr00, Chr07}, employed by \cite{Speck12, D-S19, Song}. In particular Song's paper \cite{Song}  employed the energy currents in double null foliation. The coerciveness of these energy currents, in turn,  depends strictly on the self-similar estimates of the zeroth order fluid variables, i.e.,  upper and lower bounds (which is away from zero) of $|u|^2\varrho, U_{\Lb}, U_{L'}$. Nevertheless, the bootstrap arguments allow to close up the proof.  Once all the estimates are done correctly, then the trapped surface formation follows as before, like \cite{L-Z25}.

It is interesting to note that, solving the Einstein--Euler system in the region $\ub|u|^{-1}\ll1$ is a semi-global problem for Einstein but just a local problem for Euler, which is related to the fact that the sound speed is strictly slower than light speed. This is why  $|u|^2\varrho, U_{\Lb}, U_{L'}$ themselves still obey self-similar estimates. We expect that our method apply to other matter fields that have a slower sound speed than light speed.

\section{The Existence theorem}

 The proof of Theorem \ref{main} is started by the existence theorem.

\begin{theorem}\label{existencetheorem}
Let $N$ be a sufficiently large integer (say $N=10$) and $A>0$. Suppose that the spherically symmetric initial data on $\Cb_{0}$ obeys the following estimates: 
$$|u|\tr\chi', |u|\omegab, |u|^2\rho,  |u|^{-2}\varrho^{-1},   |u|^{2+j}\Db^{j}\varrho, |u|^{j}\Db^{j}U_{\Lb}, |u|^j\Db^jU_{L'}\le A$$
for all $0\le j\le N+1$, 
and the initial data on $C_{u_0}$ (with $\omega=0$)
$$\sup_{i\le N+3, k\le 1}\sup_{\ub} \ub^{-p}\|\nablas^i(\ub D)^k\chih\|_{L^2(S_{\ub,u_0})}\le A,$$
$$\sup_{i+j+k\le N+2, k\le 1}\sup_{\ub}\|\nablas^i\Db^jD^k(\varrho, U_{\Lb}, U_{L'}, \Us)\|_{L^2(S_{\ub,u_0})}\le A,$$
where $\Us$ is the projection of $U$ to $S_{\ub,u}$ and $p>0$. Then there is an $\varepsilon>0$ sufficiently small depending on $A, \kappa, u_0, p$ such that the solution of Einstein--Euler equations exist in the region $\ub|u|^{-1}\le\varepsilon$, and the estimates established in Propositions \ref{fluidestimates}, \ref{curvature}, \ref{improvebootA}, \ref{improvebootB} hold in this region. 
\end{theorem}

Here the notation $\|\nablas^i\Db^jD^k(\varrho, U_{\Lb}, U_{L'}, \Us)\|_{L^2(S_{\ub,u_0})}\le A$  means
$$\|\nablas^i\Db^jD^k\varrho\|_{L^2(S_{\ub,u_0})}, \|\nablas^i\Db^jD^k  U_{\Lb} \|_{L^2(S_{\ub,u_0})}, \|\nablas^i\Db^jD^k U_{L'} \|_{L^2(S_{\ub,u_0})}, \|\nablas^i\Db^jD^k \Us\|_{L^2(S_{\ub,u_0})} \le A$$
and similar notations will be used throughout the paper.

In the following the notation $a\lesssim b$  denotes $a\le Cb$ for some  constant $C$ {\bf depending only on the $A$ and $\kappa, u_0, p$}. 

\begin{remark}
Note here that $\varepsilon$ is allowed to depend on the initial bound $A$. In previous works in vacuum and scalar field, $\varepsilon$ is  of order $A^{-1}$. But according  our previous work \cite{L-Z25} on Einstein--Euler--scalar field system, the dependence of $\varepsilon$ on $A$ is more complicated. So we do not attempt to track this dependence, which is harmless unless one is concerned with the sharp regularity.
\end{remark}

The  main part of the proof of the Theorem \ref{existencetheorem} is to establish the estimates in Propositions \ref{fluidestimates}, \ref{curvature}, \ref{improvebootA}, \ref{improvebootB}. Once the estimates are established, the construction of the spacetime can be done as in \cite{Chr08}, or more precisely, in Song's paper \cite{Song}. The estimates are derived under a bulk of bootstrap assumptions on the connection coefficients, which will be introduced in the course of the proof.  We derive estimates for the fluid variables via the Euler equations in Proposition \ref{fluidestimates}, for the curvature components in Proposition \ref{curvature}, and finally derive improved estimates for the lower order and top order derivatives of the connection coefficients in Propositions \ref{improvebootA} and \ref{improvebootB} respectively. The bootstrap argument can then be closed. 

Note that we close the estimate by using $N$ derivatives of the curvature components, and $N+1$ derivatives of the connection coefficients and fluid variables. 

\section{The estimates for fluid variables}

In this section we derive energy estimates for the Euler equations. The coefficients and commutators of the Euler equations are connection coefficients of the spacetime, so we begin by   the following bootstrap assumptions for connection coefficients. When deriving estimates for the curvature components and for the connection coefficients, we will introduce some more assumptions there.

 Let $\gamma\in(0,p)$ and $\gammat\in(0,\gamma)$ to be determined, where $p$ is the constant in the statement of Theorem \ref{existencetheorem}
 . Let also $\delta>0$ to be determined, and the following bootstrap assumptions are assumed to hold in the region $\ub|u|^{-1}\le\varepsilon$.

 \renewcommand{\theboot}{(A)} 
\begin{boot}\label{bootstrapA}

For $\Omega$: 
\begin{align*}
|u|^{-1}\| \Db^j D^k\log\Omega\|_{L^2(S_{\ub,u})}\le |u|^{-(j+k)}\varepsilon^{-\delta},&\  j\le N, k\le1,\\
|u|^{-1}\|\nablas^i \Db^j \log\Omega\|_{L^2(S_{\ub,u})}\le \ub^{1+\gamma}|u|^{-1-\gamma-(i+j)}\varepsilon^{-\delta}, &\   i\ge1, i+j\le N+1,\ \\
|u|^{-1}\|\nablas^i\Db^j D \log\Omega\|_{L^2(S_{\ub,u})}\le \ub^\gamma|u|^{-\gamma-1-(i+j)}\varepsilon^{-\delta}, &\ i\ge1, i+j\le N,\\
|u|^{-1}\|\Omega\nablas^{N+1}\log\Omega\|_{L^2(S_{\ub,u})}\le \ub^{1+\gamma}|u|^{-1-\gamma-(N+1)}\varepsilon^{-\delta}.
\end{align*}

Lower order derivatives: ($i+j+k\le N$):\\
For $\chih,\tr\chi,\omega$:
\begin{align*}
|u|^{-1}\|  D  (\Omega\chih)\|_{L^2(S_{\ub,u})}\le  \ub^{-1+\gamma}|u|^{-1-\gamma}\varepsilon^{-\delta},&\\
|u|^{-1}\|  \Db^j D^k  (\Omega\chih,\Omega\tr\chi,\omega)\|_{L^2(S_{\ub,u})}\le  |u|^{-1-(j+k)}\varepsilon^{-\delta},&\ k\le1,\ \text{except}\ D(\Omega\chih),\\
|u|^{-1}\|  \nablas^i\Db^j (\Omega\chih,\Omega\tr\chi,\omega)\|_{L^2(S_{\ub,u})}\le  \ub^{\gamma}|u|^{-1-\gamma-(i+j)}\varepsilon^{-\delta},&\ i\ge1.
\end{align*}
For $\eta,\etab,\chibh,\tr\chib$:
\begin{align*}
|u|^{-1}\|\Db^j  D^k (\Omega\eta,\Omega\etab,\Omega\chibh, \Omega\tr\chib)\|_{L^2(S_{\ub,u})}\le  |u|^{-1-(j+k)}\varepsilon^{-\delta},&\ k\le 1,\\
|u|^{-1}\|\nablas^i \Db^j D^k  (\Omega\eta,\Omega\etab,\Omega\chibh,  \Omega\tr\chib)\|_{L^2(S_{\ub,u})}\le  \ub^{\gamma}|u|^{-1-\gamma-(i+j+k)}\varepsilon^{-\delta},&\ k=0, i\ge1\ \text{or}\ i=k=1.\end{align*}

For $\Ks$:
\begin{align*}
|u|^{-1}\| \Db^j \Ks\|_{L^2(S_{\ub,u})}\le |u|^{-2-j}\varepsilon^{-\delta},&\ j\le N-1, \\
|u|^{-1}\|\nablas^i \Db^j  (\Omega\Ks)\|_{L^2(S_{\ub,u})}\le \ub^\gamma|u|^{-2-\gamma-(i+j)}\varepsilon^{-\delta}, &\  i\ge1, i+j\le N-1.
\end{align*}
\end{boot}

\begin{remark}\label{boostrapAdiscussion}
Unlike the scalar field case (for example \cite{Li-Liu22, An24}), in which only angular derivatives of the connection coefficients are needed, we need not only angular but also null derivatives of the connection coefficients in estimating the fluid variables via Euler equations. But as shown in Lemma \ref{DintermsofDbnablas}, $D$ derivative and $\Db$ derivative of the fluid variables are in principle equivalent, and discussed in Remark \ref{DbintermsofDlossOmega}, we will try to commute only $\Db$ derivatives to the Euler equations. So the above assumptions for the connection coefficients are mostly about the $\Db$ and $\nablas$ derivatives. 

But to get the top $(N+1)$-order energy estimates of the fluid variables, we will commute the top $(N+1)$-order differential operators with the Euler equation. we cannot commute $\Db^{N+1}$  to the Euler equations, since we do not have estimates for $\Db^{N+1}\eta$, which equals to $\nablas\Db^N\omegab$ up to  lower order terms. So we will instead commute $\Db^ND$    to the Euler equations, and express $\Db^{N+1}$ derivative to the fluid variables in terms of $\Db^ND$ and $\Db^N\nablas$ (or $\nablas\Db^N$, equivalently) derivatives. But as in Remark \ref{DbintermsofDlossOmega} we know that there will be an $\Omega^2$ loss, so we should carefully treat these terms.  Therefore we have at most one $D$ derivative in the above bootstrap assumptions. 

For a similar reason we cannot commute $\nablas\Db^N$ to the Euler equations, since we do not have estimates for $\nablas\Db^N\omegab$. At this stage we have two choices, commuting $\nablas\Db^{N-1}D$, or writing the Euler equations in the form that does not contain $\omegab$. We choose the latter way. 

The $\nablas\Db^{N-1}D$ derivative  applying to one of $\eta,\etab,\chibh$ and $\tr\chib$ come from commuting $\Db^{N}D$ to the Euler equations. We do not need $\nablas\Db^{N-1}D$ derivative to the other connection coefficients.
\end{remark}

Under the bootstrap assumptions \ref{bootstrapA} and $\varepsilon$ small enough, we will have the following  basic geometric lemmas needed in the proof in the region $\ub|u|^{-1}\le\varepsilon$.  Their proofs are standard (see \cite{Chr08} for example).
\begin{itemize}
\item {\bf Estimates for the lapse and area:}
\begin{equation}\label{lapse}\frac{1}{2}\Omega_0\le\Omega\le 2\Omega_0, \end{equation}
\begin{equation}\label{area}\frac{1}{2}|u|^2\le\frac{1}{4\pi}\mathrm{Area}(S_{\ub,u})\le 2|u|^2.\end{equation}
\item {\bf Sobolev inequality: }
Given a tangential tensorfield $\theta$, we have,
\begin{align}\label{Sobolev}
\|\theta\|_{L^\infty(S_{\ub,u})}&\lesssim\sum_{i=0}^2|u|^{-1}\|(|u|\nablas)^i\theta\|_{L^2(S_{\ub,u})},
\end{align}
\item {\bf H\"older inequality:} If we denote a scale-invariant form of the Sobolev norm, 
$$\|\theta\|_{H^k(S_{\ub,u})}=\sum_{i=0}^k \|(|u|\nablas)^i\theta\|_{L^2(S_{\ub,u})}, $$
we have, for tangential tensorfields $\theta_1,\cdots,\theta_n$, and $k\ge2$, 
\begin{equation}\label{Holder}
|u|^{-1}\|\theta_1\cdots\theta_n\|_{H^k(S_{\ub,u})}\lesssim|u|^{-n}\|\theta_1\|_{H^k(S_{\ub,u})}
\cdots\|\theta_n\|_{H^k(S_{\ub,u})}
\end{equation}
\item {\bf Gronwall type inequalities: } Given a $s$-tangential tensorfield $\theta$, we have 
\begin{align}
\label{Gronwallub}|u|^{-1}\|\theta\|_{L^2(S_{\ub,u})}
\le &C_s |u|^{-1}\left(\|\theta\|_{L^2(S_{0,u})}+\int_0^{\ub}\|D\theta\|_{L^2(S_{\ub',u})}\D\ub'\right).\\
\label{Gronwallu}\begin{split}\||u|^{s+\nu-1}\theta\|_{L^2(S_{\ub,u})}\lesssim& C_{s,\nu}\left(\||u|^{s+\nu-1}\theta\|_{L^2(S_{\ub,u_0})}\right.\\&\left.+\int_{u_0}^u\||u'|^{s+\nu-1}(\Db\theta+\frac{\nu}{2}\Omega\tr\chib\theta)\|_{L^2(S_{\ub,u'})}\D u'\right).\end{split}
\end{align}
The above estimates also hold for angular derivatives up to order $N$ (replacing $L^2$ by $H^N$), or $N+1$ when $\theta$ is a function. If $\theta$ is a trace-free $2$-tangential tensorfield, then $D\theta$ and $\Db\theta$ can be replaced by their trace-free parts $\Dh\theta$ and $\Dbh\theta$.
\end{itemize}
\begin{remark}
Together with the H\"older inequality, the bootstrap assumption for $\Omega$ guarantees that we can move $\Omega$ in or out a norm without changing its estimate.  For instance, for any $s\in\mathbb{R}$, 
\begin{equation*}|u|^{-1}\|\Omega^s\|_{H^{N}(S_{\ub,u})}\le C_s\Omega^s.\end{equation*}
This estimate tells us that for any tangential tensorfield $\theta$, $2\le k\le N$,
\begin{equation*}\|\Omega^s\theta\|_{H^k(S_{\ub,u})}\le C_{s,k}\Omega^s\|\theta\|_{H^k(S_{\ub,u})}.\end{equation*}
So $\Omega$ can be moved freely  outside and inside the $H^k$ norm. Similar inequalities hold for replacing $\nablas^k$ by $\nablas^i\Db^jD^k$ with $2\le i\le N$.  The top angular derivative $\nablas^{N+1}$ will cause a loss of $\Omega$, and then for any tangential tensorfield $\theta$, 
\begin{equation}\label{derivativelapsetop}\|\Omega^s\theta\|_{H^{N+1}(S_{\ub,u})}\le C_{s}(\Omega^s\|\theta\|_{H^{N+1}(S_{\ub,u})}+\Omega^{s-1}\|\theta\|_{H^{N}(S_{\ub,u})}),\end{equation}
which means that $\nablas^{N+1}(\Omega^s\theta)$ and $\Omega^s\nablas^{N+1}\theta$ behave similar, losing an $\Omega$ as compared to $\nablas^{\le N}(\Omega^s\theta)$ or $\Omega^s\nablas^{\le N}\theta$.  Moreover, if a quantity obeys better estimates after applying angular derivatives than itself, say
$$|u|^{-1}\|\theta\|_{L^2(S_{\ub,u})}\lesssim 1, \ |u|^{-1}\|\nablas\theta\|_{L^2(S_{\ub,u})}\lesssim\ub^{1+\gammat}|u|^{-1-\gammat},$$
then the bootstrap assumption for $\Omega$ still allow 
$$|u|^{-1}\|\Omega^s\theta\|_{L^2(S_{\ub,u})}\le C_s \Omega^s, \ |u|^{-1}\|\nablas(\Omega^s\theta)\|_{L^2(S_{\ub,u})}\le C_s \Omega^s\ub^{1+\gammat}|u|^{-1-\gammat}$$
whenever $\gamma>\gammat$ and $\varepsilon$ sufficiently small (depending on $\gamma-\gammat$). Similar relations hold for higher order estimates.
\end{remark}

 \renewcommand{\theboot}{(B)} 
\begin{boot}\label{bootstrapB} Top order derivatives: ($i+j=N+1$) 
   
  Mixed derivatives:
\begin{align*}
|u|^{2\gammat}\int_{\Cb_{\ub}}|u'|^{-1-2\gammat+2(N+1)}|  \Db^{N}D(\Omega\chih, \Omega\tr\chi,\Omega\chibh,\Omega\tr\chib, \Omega\eta,\Omega\etab, \omega)|^2\le \varepsilon^{-2\delta},\\
\int_{\Cb_{\ub}}\ub^{-2\gamma}|u|^{-1+2\gamma+2(N+1)}|\nablas^i \Db^j  (\Omega\chih, \Omega\tr\chi,\Omega\chibh,\Omega\tr\chib, \Omega\eta,\Omega\etab, \omega)|^2\le \varepsilon^{-2\delta},&\  1\le i\le N,\\
\int_{\Cb_{\ub}}\ub^{-2\gamma}|u|^{-1+2\gamma+2(N+1)}|\nablas \Db^{N-1}D  ( \Omega\chibh,\Omega\tr\chib, \Omega\eta,\Omega\etab)|^2\le \varepsilon^{-2\delta},&\ 
\end{align*}
  
 Top angular derivatives: 
\begin{align*}
\int_{C_u}\ub^{-1-2\gamma}|u|^{2\gamma+2(N+1)}|\nablas^{N+1}(\Omega^2\chih,\Omega\omega)|^2\le \varepsilon^{-2\delta},\\
\int_{\Cb_{\ub}}\ub^{-2\gamma}|u|^{-1+2\gamma+2(N+1)}|\nablas^{N+1}(\Omega^2\eta,\Omega^2\etab,\Omega^2\chibh)|^2\le \varepsilon^{-2\delta},\\
|u|^{-1}\|\nablas^{N+1}(\Omega^2\tr\chi,\Omega^2\tr\chib)\|_{L^2(S_{\ub,u})}\le \ub^{\gamma}|u|^{-1-\gamma-(N+1)}\varepsilon^{-\delta}, 
\end{align*}

For $\Ks$: ($i+j=N$)
\begin{align*}
\int_{\Cb_{\ub}}\ub^{-2\gamma}|u|^{-1+2\gamma+2N}|\nablas^i\Db^j(\Omega\Ks)|^2\le\varepsilon^{-2\delta}, &\ 1\le i,j\le N-1,\\
\int_{\Cb_{\ub}}\ub^{-2\gamma}|u|^{-1+2\gamma+2N}|\nablas^N(\Omega^2\Ks)|\le\varepsilon^{-2\delta}.
\end{align*}

\end{boot} 

\begin{remark}
Here the integral $\int_{\Cb_{\ub}}$ means
$$\int_{\Cb_{\ub}} f(\ub, u',\vartheta):=\int_{u_0}^u \int_{S_{\ub',u'}}f(\ub, u',\vartheta)\sqrt{\det\gs}\D\vartheta\D u'$$
for a positive function $f$ and $\ub|u|^{-1}=\varepsilon$. The integral is along the whole $\Cb_{\ub}$ in the region $ \ub'|u'|^{-1}\le\varepsilon$. Similarly, 
$$\int_{C_{u}} f(\ub', u, \vartheta):=\int_{0}^{\ub} \int_{S_{\ub',u'}}f(\ub', u,\vartheta)\sqrt{\det\gs}\D\vartheta\D \ub'$$
for a positive function $f$ and $\ub|u|^{-1}=\varepsilon$. 
\end{remark}
\begin{remark}
The top order angular derivatives will lose an $\Omega$. 
\end{remark}

We start to analyze the Euler equations and the obtain estimates for the fluid variables. Recall that for a vector field $X$, we denote 
$$U_X=g(U,X),$$
then we have the expression of fluid velocity
$$U=-\frac{1}{2}U_{\Lb}L'-\frac{1}{2}U_{L'}\Lb+\Us$$
where the condition $g(U,U)=-1$ reads
$$U_{\Lb}U_{L'}-|\Us|^2=1,$$
which will be used throughout the section. In the double null frame, the Euler equations \eqref{Euler}
  can be expressed in the following form in terms of $(\varrho, U_{\Lb}, U_{L'}, \Us)$, in which case $\omega$ does not appear: 
\begin{align*}
&U\varrho+(1+\kappa)\varrho(-\frac{1}{2}L'(U_{\Lb})-\frac{1}{2}\Lb(U_{L'})+\divs\Us)\\
&=(1+\kappa)\varrho(\omegab U_{L'}-\etab_A\Us^A-\eta_A\Us^A+\frac{1}{2}\tr\chi' U_{\Lb}+\frac{1}{2}\Omega\tr\chib U_{L'})
\end{align*}
\begin{align*}
& \varrho U(U_{\Lb})+\frac{\kappa}{(1+\kappa)} (\nabla_{\Lb}\varrho+U_{\Lb}U\varrho)\\
&=-\varrho U_{\Lb}(U_{L'}\omegab-\Us^A\eta_A)+\varrho\Us^A(-U_{\Lb}\etab_A+\Us^B\Omega\chib_{AB}), 
\end{align*}
\begin{align*}
 &\varrho U(U_{L'})+\frac{\kappa}{(1+\kappa)}(\nabla_{L'}\varrho+U_{L'}U\varrho)\\
 &=-\varrho U_{L'}(-U_{L'}\omegab+\Us^A\eta_A)+\varrho\Us^A(-U_{L'}\eta_A+\Us^B\Omega^{-1}\chi_{AB}),
 \end{align*}
 \begin{align*}
&2\varrho(\nablas_U(\Us))^A+\frac{2\kappa}{(1+\kappa)}(\nablas^A\varrho+\Us^AU\varrho)\\
&=\varrho U_{\Lb}(-U_{L'}\eta^A+\Us^B\Omega^{-1}\chi_B^A)+\varrho U_{L'}(-U_{\Lb}\etab^A+\Us^B\Omega\chib^A_B).
\end{align*}
The advantage of this expression is that the variables $U_{\Lb}, U_{L'}\lesssim1$ is the self-similar variables. However, as discussed in Remark \ref{boostrapAdiscussion}, we should use the form that $\omegab$ does not appear. It can be achieved by rewriting the Euler equations in terms of $(\varrho, U_{\Lb'}, U_{L}, \Us)$ where
$$U_{\Lb'}=g(U,\Lb')=\Omega^{-2}U_{\Lb},$$
$$ U_L=g(U,L)=\Omega^2U_{L'}.$$
The corresponding equations by taking conjugation:
\begin{equation}\label{EulerE0}
\begin{split}
&U\varrho+(1+\kappa)\varrho(-\frac{1}{2}\Lb'( U_{L})-\frac{1}{2}L( U_{\Lb'})+\divs\Us)\\
&=(1+\kappa)\varrho(\omega U_{\Lb'}-\etab_A\Us^A-\eta_A\Us^A+\frac{1}{2}\tr\chib' U_{L}+\frac{1}{2}\Omega\tr\chi U_{\Lb'})=:\mathcal{E}_0
\end{split}
\end{equation}
\begin{equation}\label{EulerEL}
\begin{split}
& \varrho U(U_{L})+\frac{\kappa}{(1+\kappa)} (\nabla_{L}\varrho+U_{L}U\varrho)\\
&=-\varrho U_{L}(U_{\Lb'}\omega-\Us^A\etab_A)+\varrho\Us^A(-U_{L}\eta_A+\Us^B\Omega\chi_{AB})=:\mathcal{E}_{L}, 
\end{split}
\end{equation}
\begin{equation}\label{EulerELb'}
\begin{split}
 &\varrho U(U_{\Lb'})+\frac{\kappa}{(1+\kappa)}(\nabla_{\Lb'}\varrho+U_{\Lb'}U\varrho)\\
 &=-\varrho U_{\Lb'}(-U_{\Lb'}\omega+\Us^A\etab_A)+\varrho\Us^A(-U_{\Lb'}\etab_A+\Us^B\Omega^{-1}\chib_{AB})=:\mathcal{E}_{\Lb'},
\end{split}
\end{equation}
\begin{equation}\label{EulerEs}
\begin{split}
&2\varrho(\nablas_U(\Us))^A+\frac{2\kappa}{(1+\kappa)}(\nablas^A\varrho+\Us^AU\varrho)\\
&=\varrho U_{L}(-U_{\Lb'}\etab^A+\Us^B\Omega^{-1}\chib_B^A)+\varrho U_{\Lb'}(-U_{L}\eta^A+\Us^B\Omega\chi^A_B)=:\Es.
\end{split}
\end{equation}

These equations to express $D$ derivative in terms of $\Db$ and $\nablas$ derivatives, or express $\Db$ derivatives in terms of $D$ and $\nablas$ derivatives. 
\begin{lemma}\label{DintermsofDbnablas}
If $\varrho (U_{\Lb'})^2>0$, then $D\varrho, DU_{\Lb'}, DU_{L}, D\Us$ can respectively be expressed as a sum which is a linear combination of $\Db,\nablas$ derivatives of $\varrho, U_{\Lb'}, U_L, \Us$, with the coefficients being the fluid variables themselves and the connection coefficients. Similar conclusions hold when reversing $D$ and $\Db$ derivatives and assuming  $\varrho (U_{L})^2>0$.
\end{lemma}
\begin{proof}
Rewriting \eqref{EulerE0} and \eqref{EulerELb'} as
\begin{align*}
&-\frac{1}{2}U_{\Lb'}L\varrho-\frac{1}{2}(1+\kappa)\varrho L( U_{\Lb'})\\
=&\mathcal{E}_0+\frac{1}{2}U_{L}\Lb'\varrho-\Us\varrho-(1+\kappa)\varrho\left(-\frac{1}{2}\Lb'( U_{L})+\divs\Us\right),\\
 &-\frac{1}{2}\varrho  U_{\Lb'}LU_{\Lb'}-\frac{1}{2}\frac{\kappa}{(1+\kappa)}(U_{\Lb'})^2L\varrho\\
 =&\mathcal{E}_{\Lb'}+\frac{1}{2}\varrho  U_{L}\Lb'U_{\Lb'}-\varrho \Us U_{\Lb'}-\frac{\kappa}{(1+\kappa)}\left(\frac{1}{2}\Lb'\varrho-\frac{1}{2}|\Us|^2\Lb'\varrho+U_{\Lb'}\Us\varrho)\right).\end{align*}
 The determinant of the coefficient matrix is $\frac{1}{4}(1-\kappa)\varrho(U_{\Lb'})^2>0$. By solving these equations directly, we have the schematic  form
 \begin{align*}
D \varrho=&(U_{\Lb'})^{-1}\mathcal{E}_0+(U_{\Lb'})^{-2}\mathcal{E}_{\Lb'}+\varrho  U_{L}(U_{\Lb'})^{-2}\Lb'U_{\Lb'}+\varrho  (U_{\Lb'})^{-2}\Us \nablas U_{\Lb'}\\
&+(U_{\Lb'})^{-2}(1+|\Us|^2)\Lb'\varrho+(U_{\Lb'})^{-1}\Us\nablas\varrho+\varrho(U_{\Lb'})^{-1}(\Lb'( U_{L})+\divs\Us),
 \end{align*}
  \begin{align*}
D U_{\Lb'}=&\varrho^{-1}\mathcal{E}_0+\varrho^{-1}(U_{\Lb'})^{-1}\mathcal{E}_{\Lb'}+  U_{L}(U_{\Lb'})^{-1}\Lb'U_{\Lb'}+(U_{\Lb'})^{-1}\Us \nablas U_{\Lb'}\\
 &+\varrho^{-1}(U_{\Lb'})^{-1}(1+|\Us|^2)\Lb'\varrho+\varrho^{-1}\Us\nablas\varrho+\Lb'( U_{L})+\divs\Us.
 \end{align*} 
Equation \eqref{EulerEL} then shows
 \begin{align*}
DU_L=-2(U_{\Lb'})^{-1}\left(\varrho^{-1}(\mathcal{E}_L-\frac{\kappa}{1+\kappa}(D\varrho+U_LU\varrho))-U_L\Lb' U_{L}+2\Us \nablas U_L\right).
\end{align*}
Equation \eqref{EulerEs} shows
 \begin{align*}
\nablas_L\Us=-(U_{\Lb'})^{-1}\left(\varrho^{-1}(\Es-\frac{2\kappa}{1+\kappa}(\nablas\varrho+\Us U\varrho))-U_L\nablas_{\Lb'}\Us+2\Us \nablas \Us\right).
\end{align*}
  In similar manner, $\Db$ derivative can be expressed in terms of $D$ and $\nablas$ derivatives. 
\end{proof}
 
 \begin{remark}
Through these formulas, readers can check directly the discussions in and before Remark \ref{DbintermsofDlossOmega}.
 \end{remark}
 
 In the following we present schematic formulas commuting derivatives to \eqref{EulerE0}--\eqref{EulerEs}, the proof of which is by successively using Lemma \ref{commutationformulas}, and will be omitted.   We denote $\Dbf^{\vec{i}}\theta$ to be the sum of mixed derivatives of $\nablas, \Db$ and $D$ with total orders $\is, i_{\Lb}$ and $i_{L}$ respectively, where $\vec{i}=(\is, i^{\Lb}, i^L)$. For example,
 $$\Dbf^{(1,1,1)}f=\nablas\Db D f+\nablas D\Db f+\Db\nablas Df+\Db D\nablas f+D\nablas\Db f+D\Db\nablas f.$$
  We also allow that $\is, i^{\Lb}, i^L$ take negative value, and denote $\Dbf^{\vec{i}}=0$ if one of $\is, i^{\Lb}, i^L$ is negative. We also denote
$$|\vec{i}|=\is+i^{\Lb}+i^L$$
and
$$\vec{i}+\vec{j}=(\is+\js, i^{\Lb}+j^{\Lb}, i^L+j^L).$$
We also denote, say $\Dbf^{(\is+1, i^{\Lb}, i^L)}$ by $\Dbf^{\vec{i}}\nablas$, and $\Dbf^{(\is, i^{\Lb}-1, i^L)}$ by $\Dbf^{\vec{i}}\Db^{-1}$. This is to say,  the notation, say $\Dbf^{\vec{i}}\nablas$, means taking an additional $\nablas$ without fixed order, rather than taking $\nablas$ first and then $\Dbf^{\vec{i}}$.

\begin{remark}
Note that in the bootstrap assumptions \ref{bootstrapA} and \ref{bootstrapB}, the derivatives $\nablas,\Db$ and $D$ appear in a fixed order. In fact, by Lemma \ref{commutationformulas}, the bounds also hold when we commute the derivatives, which will be omitted in the course of the proof. We introduce the notation $\Dbf^{\vec{i}}$ in estimating the fluid variables purely for the sake of concision. 
\end{remark}

  \begin{lemma} For function $f$ and tangential vectorfield $\theta$, we have  
\begin{align*}[\Dbf^{\vec{i}}, \nablas_L]f=\sum(\Dbf^{\vec{i}_1}\Db^{-1})(\Omega^2\zeta^\sharp)\Dbf^{\vec{i}_2}\nablas f+\sum(\Dbf^{\vec{i}_1}\nablas^{-2})\nablas(\Omega\chi)\Dbf^{\vec{i}_2}\nablas f,\end{align*}
\begin{align*}[\Dbf^{\vec{i}}, \nablas_L]\theta=&\sum(\Dbf^{\vec{i}_1}\Db^{-1})(\Omega^2\zeta^\sharp)\Dbf^{\vec{i}_2}\nablas \theta+\sum(\Dbf^{\vec{i}_1}\Db^{-1})\nablas (\Omega^2\zeta^\sharp)\Dbf^{\vec{i}_2}\theta\\
&+\sum  \Dbf^{\vec{i}_1}  (\Omega\chi)\Dbf^{\vec{i}_2}\theta,\end{align*}
\begin{align*}[\Dbf^{\vec{i}}, \nablas_{\Lb}]f=\sum(\Dbf^{\vec{i}_1}D^{-1})(\Omega^2\zeta^\sharp)\Dbf^{\vec{i}_2}\nablas f+\sum(\Dbf^{\vec{i}_1}\nablas^{-2})\nablas(\Omega\chib)\Dbf^{\vec{i}_2}\nablas f,\end{align*}
\begin{align*}[\Dbf^{\vec{i}}, \nablas_{\Lb}]\theta=&\sum(\Dbf^{\vec{i}_1}D^{-1})(\Omega^2\zeta^\sharp)\Dbf^{\vec{i}_2}\nablas \theta+\sum(\Dbf^{\vec{i}_1}D^{-1})\nablas (\Omega^2\zeta^\sharp)\Dbf^{\vec{i}_2}\theta,\\
&+\sum \Dbf^{\vec{i}_1}  (\Omega\chib)\Dbf^{\vec{i}_2}\theta\end{align*}
\begin{align*}[\Dbf^{\vec{i}}, \nablas]f=&\sum(\Dbf^{\vec{i}_1}\nablas^{-2})\Ks\Dbf^{\vec{i}_2}\nablas f\\
&+\sum(\Dbf^{\vec{i}_1}D^{-1}\nablas^{-1})\nablas(\Omega\chi)\Dbf^{\vec{i}_2}\nablas f+\sum(\Dbf^{\vec{i}_1}\Db^{-1}\nablas^{-1})\nablas(\Omega\chib)\Dbf^{\vec{i}_2}\nablas f\end{align*}
\begin{align*}[\Dbf^{\vec{i}}, \nablas]\theta=&\sum(\Dbf^{\vec{i}_1}\nablas^{-1})\Ks\Dbf^{\vec{i}_2}\theta,\\
&+\sum(\Dbf^{\vec{i}_1}D^{-1})\nablas(\Omega\chi)\Dbf^{\vec{i}_2}\theta+\sum(\Dbf^{\vec{i}_1}\Db^{-1})\nablas(\Omega\chib)\Dbf^{\vec{i}_2}\theta,\end{align*}
where the sums are over $\sum_j{\vec{i}_j}=\vec{i}$. 
\end{lemma}
\begin{remark}
In applying Lemma \ref{commutationformulas}, we also need the relations
$$\Db f=\nablas_{\Lb}f,$$
$$\Db\theta=\nablas_{\Lb}\theta+\chib\cdot\theta.$$
\end{remark}
  
 For any function $g$, we have
 \begin{align*} [\Dbf^{\vec{i}}, g\nablas_{\Lb}]=&g[\Dbf^{\vec{i}}, \nablas_{\Lb}]+\sum_{|\vec{i}_1|\ge1}\Dbf^{\vec{i}_1}g \Dbf^{\vec{i}_2}\Db+\sum_{|\vec{i}_1|\ge1, \text{applying to vector}}\Dbf^{\vec{i}_1}g \Dbf^{\vec{i}_2}(\Omega\chib)\Dbf^{\vec{i}_3},\\
 [\Dbf^{\vec{i}}, g\nablas_{L}]=&g[\Dbf^{\vec{i}}, \nablas_{L}]+\sum_{|\vec{i}_1|\ge1}\Dbf^{\vec{i}_1}g \Dbf^{\vec{i}_2}D+\sum_{|\vec{i}_1|\ge1, \text{applying to vector}}\Dbf^{\vec{i}_1}g \Dbf^{\vec{i}_2}(\Omega\chi)\Dbf^{\vec{i}_3},\\
 [\Dbf^{\vec{i}}, g\nablas]=&g[\Dbf^{\vec{i}},\nablas]+\sum_{|\vec{i}_1|\ge1}\Dbf^{\vec{i}_1}g \Dbf^{\vec{i}_2}\nablas.
 \end{align*}
The lemma below follows by writing $U=-\frac{1}{2}U_{\Lb'}L-\frac{1}{2}U_{L'}\Lb+\Us$ and commuting $\Dbf^{\vec{i}}$ with $\nablas_U$.
\begin{lemma}\label{EulerDbfi}
 \begin{align*}
\nablas_U(\Dbf^{\vec{i}}\varrho)+(1+\kappa)\varrho(-\frac{1}{2}\nablas_{\Lb'}(\Dbf^{\vec{i}}U_{L})-\frac{1}{2}\nablas_{L}(\Dbf^{\vec{i}}U_{\Lb'})+\divs(\Dbf^{\vec{i}}\Us))=\mathcal{E}_{0,\vec{i}},
\end{align*}
\begin{align*}
\varrho \nablas_U(\Dbf^{\vec{i}}U_{\Lb'})+\frac{\kappa}{(1+\kappa)} (\nablas_{\Lb'}(\Dbf^{\vec{i}}\varrho)+U_{\Lb'}\nablas_U(\Dbf^{\vec{i}}\varrho))=\mathcal{E}_{\Lb',\vec{i}},
\end{align*}
 \begin{align*}
2\varrho(\nablas_U(\Dbf^{\vec{i}}\Us))^A+\frac{2\kappa}{(1+\kappa)}(\nablas^A(\Dbf^{\vec{i}}\varrho)+\Us^A\nablas_U(\Dbf^{\vec{i}}\varrho))=\Es_{\vec{i}},
\end{align*}
where
\begin{align*}
\mathcal{E}_{0,\vec{i}}=\Dbf^{\vec{i}}\mathcal{E}_0+\sum&\Dbf^{\vec{i}_0}(\Db^{-1}(\Omega^2\zeta^\sharp), \nablas^{-2}\nablas(\Omega\chih,\Omega\tr\chi))(U_{\Lb'}\Dbf^{\vec{i}_1}\nablas\varrho+\varrho\Dbf^{\vec{i}_1}\nablas U_{\Lb'})\\
+&\Dbf^{\vec{i}_0} (D^{-1}(\Omega^2\zeta^\sharp), \nablas^{-2}\nablas(\Omega\chibh, \Omega\tr\chib))(U_{L'}\Dbf^{\vec{i}_1}\nablas\varrho+\varrho\Dbf^{\vec{i}_2}\nablas U_{L'})\\
+&\Dbf^{\vec{i}_0}\left(D^{-1}\nablas (\Omega\chih,\Omega\tr\chi),\Db^{-1}\nablas (\Omega\chibh,\Omega\tr\chib), \nablas^{-1}\Ks\right)\Dbf^{\vec{i}_1}\varrho\Dbf^{\vec{i}_2}\Us\\
+\sum_{|\vec{i}_1|\ge1}&\Dbf^{\vec{i}_1} U_{\Lb'}\Dbf^{\vec{i}_2}D\varrho+\Dbf^{\vec{i}_1}\varrho\Dbf^{\vec{i}_2} DU_{\Lb'}\\
+&\Dbf^{\vec{i}_1}U_{L'}\Dbf^{\vec{i}_2}\Db\varrho+\Dbf^{\vec{i}_1}\varrho\Dbf^{\vec{i}_2}\Db U_{L'}\\
+&\Dbf^{\vec{i}_1}\Us\Dbf^{\vec{i}_2}\nablas\varrho+\Dbf^{\vec{i}_1}\varrho\Dbf^{\vec{i}_2}\nablas\Us,
\end{align*}

\begin{align*}
\mathcal{E}_{\Lb',\vec{i}}=\Dbf^{\vec{i}}\mathcal{E}_{\Lb'}+\sum&\Dbf^{\vec{i}_0}( D^{-1}(\Omega^2\zeta^\sharp), \nablas^{-2}\nablas(\Omega\chibh, \Omega\tr\chib))\\&\times(\varrho U_{L'}\Dbf^{\vec{i}_1}\nablas U_{\Lb'}+ (U_{\Lb'}U_{L'}+\Omega^{-2})\Dbf^{\vec{i}_1}\nablas\varrho )\\
+&\Dbf^{\vec{i}_0}( \Db^{-1}(\Omega^2\zeta^\sharp),  \nablas^{-2}\nablas(\Omega\chih, \Omega\tr\chi))\\
&\times U_{\Lb' }(\varrho\Dbf^{\vec{i}_1}\nablas U_{\Lb'}+U_{\Lb'}\Dbf^{\vec{i}_1}\nablas\varrho)\\
+& \Dbf^{\vec{i}_0}\left( D^{-1}\nablas(\Omega\chih,\Omega\tr\chib), \Db^{-1}\nablas(\Omega\chibh,\Omega\tr\chib), \nablas^{-1}\Ks\right)\Us\Dbf^{\vec{i}_1}\varrho\Dbf^{\vec{i}_2}U_{\Lb'}\\
+\sum_{|\vec{i}_1|\ge1}&\Dbf^{\vec{i}_1}(\varrho U_{L'})\Dbf^{\vec{i}_2}\Db U_{\Lb'}+ \Dbf^{\vec{i}_1}(U_{\Lb'}U_{L'}+\Omega^{-2})\Dbf^{\vec{i}_2}\Db\varrho \\
+&\Dbf^{\vec{i}_1}(\varrho U_{\Lb'})\Dbf^{\vec{i}_2}D U_{\Lb'}+\Dbf^{\vec{i_1}}(U_{\Lb'})^2\Dbf^{\vec{i}_2}D\varrho\\
+&\Dbf^{\vec{i}_1}(\varrho\Us)\Dbf^{\vec{i}_2}\nablas U_{\Lb'}+\Dbf^{\vec{i_1}}(U_{\Lb'}\Us)\Dbf^{\vec{i}_2}\nablas\varrho.
\end{align*}
 Note that if $\Omega\chih$ or $\Omega\chibh$ appears in the expressions, then the corresponding term will contain at least one $\nablas$ acting on fluid variable. And also 
\begin{align*}
\Es_{\vec{i}}=\Dbf^{\vec{i}}\Es+\sum&\Dbf^{\vec{i}_0}( D^{-1}\nablas (\Omega^2\zeta^\sharp), \Omega\chibh, \Omega\tr\chib)\varrho U_{L'}\Dbf^{\vec{i}_1}\Us\\
+&\Dbf^{\vec{i}_0}D^{-1} (\Omega^2\zeta^\sharp)(\Us\Dbf^{\vec{i}_1}\nablas\varrho + \varrho\Dbf^{\vec{i}_1}\nablas\Us)U_{L'}\\
+&\Dbf^{\vec{i}_0}(  \Db^{-1}\nablas(\Omega^2\zeta^\sharp),  \Omega\chih, \Omega\tr\chi)\varrho U_{\Lb'}\Dbf^{\vec{i}_1}\Us\\
+&\Dbf^{\vec{i}_0}\Db^{-1} (\Omega^2\zeta^\sharp)(\Us\Dbf^{\vec{i}_1}\nablas\varrho + \varrho\Dbf^{\vec{i}_1}\nablas\Us)U_{\Lb'}\\
+&\Dbf^{\vec{i}_0}\left( D^{-1}\nablas(\Omega\chih,\Omega\tr\chi), \Db^{-1}\nablas(\Omega\chibh,\Omega\tr\chib), \nablas^{-1}\Ks\right)\Us\Dbf^{\vec{i}_1}\varrho\Dbf^{\vec{i}_2}\Us\\
+&\Dbf^{\vec{i}_0} \nablas^{-2}\Ks\Dbf^{\vec{i}_1}\nablas\varrho\\
+\sum_{|\vec{i}_1|\ge1}&\Dbf^{\vec{i}_1}(\varrho U_{L'})\Dbf^{\vec{i}_2}(\Db \Us+\Omega\chib \Us)+ \Dbf^{\vec{i}_1}(U_{L'}\Us)\Dbf^{\vec{i}_2}\Db\varrho \\
+&\Dbf^{\vec{i}_1}(\varrho U_{\Lb'})\Dbf^{\vec{i}_2}(D \Us+\Omega\chi \Us)+ \Dbf^{\vec{i}_1}(U_{\Lb'}\Us)\Dbf^{\vec{i}_2}D\varrho \\
+&\Dbf^{\vec{i}_1}(\varrho\Us)\Dbf^{\vec{i}_2}\nablas \Us+\Dbf^{\vec{i_1}}(\Us\Us)\Dbf^{\vec{i}_2}\nablas\varrho\\
+&\Dbf^{\vec{i}_1} \gs^{-1}\Dbf^{\vec{i}_2}\nablas\varrho.
\end{align*}
Here the sums are over $\sum_j{\vec{i}_j}=\vec{i}$. 

\end{lemma}

We start to do energy estimate. For each $\vec{i}$, let us define $\widehat{\Dbf^{\vec{i}}U_{L}}$ to be 
\begin{equation}\label{Ulddef} U_{\Lb'}\widehat{\Dbf^{\vec{i}}U_{L}}+U_{L}\Dbf^{\vec{i}}U_{\Lb'}-2\Us\cdot\Dbf^{\vec{i}}\Us=0.\end{equation}
The crucial fact is that the difference $\Dbf^{\vec{i}}U_{L}-\widehat{\Dbf^{\vec{i}}U_{L}}$ only contains lower order terms, which can be seen by the relation  $U_{\Lb'}U_{L}-|\Us|^2=1$. We write

\begin{equation}\label{EulerE0i}
\begin{split}
&\nablas_U\Dbf^{\vec{i}}\varrho+(1+\kappa)\varrho(-\frac{1}{2}\nablas_{\Lb'}(\widehat{\Dbf^{\vec{i}}U_{L}})-\frac{1}{2}\nablas_{L}(\Dbf^{\vec{i}}U_{\Lb'})+\divs\Dbf^{\vec{i}}\Us)\\=&\mathcal{E}_{0,\vec{i}}+\frac{1}{2}(1+\kappa)\varrho\nablas_{\Lb'}(\Dbf^{\vec{i}}U_{L}-\widehat{\Dbf^{\vec{i}}U_{L}}),
\end{split}
\end{equation}
\begin{equation}\label{EulerELb'i}
\begin{split}
& \varrho \nablas_U(\Dbf^{\vec{i}}U_{\Lb'})+\frac{\kappa}{(1+\kappa)} (\nablas_{\Lb'}\Dbf^{\vec{i}}\varrho+U_{\Lb'}\nablas_U\Dbf^{\vec{i}}\varrho)=\mathcal{E}_{\Lb',\vec{i}}, 
\end{split}
\end{equation}
\begin{equation}\label{EulerEsi}
\begin{split}
&2\varrho(\nablas_U(\Dbf^{\vec{i}}\Us))^A+\frac{2\kappa}{(1+\kappa)}(\nablas^A\Dbf^{\vec{i}}\varrho+\Us^A\nablas_U\Dbf^{\vec{i}}\varrho)=\Es_{\vec{i}},
\end{split}
\end{equation}
The equation of $\widehat{\Dbf^{\vec{i}}U_{L}}$ is obtained by taking derivatives to the relation \eqref{Ulddef}:
\begin{equation}\label{EulerELi}
\begin{split}
 &\varrho \nablas_U(\widehat{\Dbf^{\vec{i}}U_{L}})+\frac{\kappa}{(1+\kappa)}(\nablas_{L}\Dbf^{\vec{i}}\varrho+U_{L}\nablas_U\Dbf^{\vec{i}}\varrho)\\
 =&\varrho\nablas_U(-\frac{U_{L}}{U_{\Lb'}})\Dbf^{\vec{i}}U_{\Lb'}+\varrho\nablas_U(\frac{2\Us}{U_{\Lb'}})\cdot\Dbf^{\vec{i}}\Us-\frac{U_{L}}{U_{\Lb'}}\mathcal{E}_{\Lb',\vec{i}}+\frac{\Us}{U_{\Lb'}}\Es_{\vec{i}}=:\mathcal{E}_{L,\vec{i}}.
\end{split}
\end{equation}

We consider the following energy currents, introduced by Christodoulou \cite{Chr00,Chr07} (See also \cite{Song} in double null foliations). 
\begin{align*}J_{\Lb',\vec{i}}=&\frac{\kappa}{(1+\kappa)\varrho} U_{\Lb'}|\Dbf^{\vec{i}}\varrho|^2+2\kappa \Dbf^{\vec{i}}U_{\Lb'}\Dbf^{\vec{i}}\varrho+(1+\kappa)\varrho U_{\Lb'}\left(-\Dbf^{\vec{i}}U_{\Lb'}\widehat{\Dbf^{\vec{i}}U_{L}}+|\Dbf^{\vec{i}}\Us|^2\right).\end{align*}
\begin{align*}J_{L,\vec{i}}=&\frac{\kappa}{(1+\kappa)\varrho} U_{L}|\Dbf^{\vec{i}}\varrho|^2+2\kappa \widehat{\Dbf^{\vec{i}}U_{L}}\Dbf^{\vec{i}}\varrho+(1+\kappa)\varrho U_{L}\left(-\Dbf^{\vec{i}}U_{\Lb'}\widehat{\Dbf^{\vec{i}}U_{L}}+|\Dbf^{\vec{i}}\Us|^2\right).\end{align*}
\begin{align*}\Js_{\vec{i}}=&\frac{\kappa}{(1+\kappa)\varrho} \Us|\Dbf^{\vec{i}}\varrho|^2+2\kappa \Dbf^{\vec{i}}\Us\Dbf^{\vec{i}}\varrho+(1+\kappa)\varrho \Us\left(-\Dbf^{\vec{i}}U_{\Lb'}\widehat{\Dbf^{\vec{i}}U_{L}}+|\Dbf^{\vec{i}}\Us|^2\right).\end{align*}

The coerciveness of these currents are guaranteed by the following lemma.
\begin{lemma}\label{coerciveness} Assume
\begin{equation}\label{geometricacoustical1} U_{\Lb},  U_{L'}, |u|^2\varrho, |u|^{-2}\varrho^{-1}\lesssim 1.\end{equation}
Then it holds
$$|\Omega^2 J_{\Lb',\vec{i}}|, |\Omega^{-2} J_{L,\vec{i}}| \gtrsim \kappa^2|u|^2|\Dbf^{\vec{i}}\varrho|^2+|u|^{-2}\left(\Omega^4|\Dbf^{\vec{i}} U_{\Lb'}|^2+|\Dbf^{\vec{i}}\Us|^2\right).$$
where the constant does not depend on $\kappa$.  We also have
$$|\Omega^2 J_{\Lb',\vec{i}}|, |\Omega^{-2} J_{L,\vec{i}}| \lesssim |u|^2|\Dbf^{\vec{i}}\varrho|^2+|u|^{-2}\left(\Omega^4|\Dbf^{\vec{i}} U_{\Lb'}|^2+|\Dbf^{\vec{i}}\Us|^2\right).$$
In particular, $|\Omega^2 J_{\Lb',\vec{i}}|\approx |\Omega^{-2} J_{L,\vec{i}}|$.
\end{lemma}
\begin{proof}
They are direct consequences by using Cauchy--Schwarz inequality and \eqref{Ulddef}. 
\end{proof}

We write
\begin{equation}\label{fluiddivergence}\nabla_{L}\left(\Omega^2J_{\Lb',\vec{i}}\right)+\nabla_{\Lb}\left( J_{L,\vec{i}}\right)-2\divs\left(\Omega^2\Js_{\vec{i}}\right)=\tau_{\vec{i}},\end{equation}
where (in a schematic form)
\begin{align*}\tau_{\vec{i}}=&\nablas_L(\varrho^{-1}U_{\Lb})\Dbf^{\vec{i}}\varrho\Dbf^{\vec{i}}\varrho+\nablas_L(\varrho U_{\Lb})\left(-\Dbf^{\vec{i}}U_{\Lb'}\widehat{\Dbf^{\vec{i}}U_{L}}+|\Dbf^{\vec{i}}\Us|^2\right)\\
&+\nablas_{\Lb}( \varrho^{-1}U_{L})\Dbf^{\vec{i}}\varrho\Dbf^{\vec{i}}\varrho+\nablas_{\Lb}( \varrho U_{L}\gs)\left(-\Dbf^{\vec{i}}U_{\Lb'}\widehat{\Dbf^{\vec{i}}U_{L}}+|\Dbf^{\vec{i}}\Us|^2\right)\\
&+\divs(\Omega^2\varrho^{-1}\Us)\Dbf^{\vec{i}}\varrho\Dbf^{\vec{i}}\varrho+\divs(\Omega^2\varrho \Us)\left(-\Dbf^{\vec{i}}U_{\Lb'}\widehat{\Dbf^{\vec{i}}U_{L}}+|\Dbf^{\vec{i}}\Us|^2\right)\\
&+\Omega^2\left(\varrho^{-1}\mathcal{E}_{0,\vec{i}}\Dbf^{\vec{i}}\varrho+\mathcal{E}_{L,\vec{i}}\Dbf^{\vec{i}}U_{\Lb'}+\mathcal{E}_{\Lb',\vec{i}}\widehat{\Dbf^{\vec{i}}U_{L}}+ \Es_{\vec{i}}\Dbf^{\vec{i}}\Us\right).
\end{align*}
The crucial point is that the right hand side does not have top order terms. Using
\begin{align*}|\widehat{\Dbf^{\vec{i}}U_{L}}|^2\lesssim\Omega^8|\Dbf^{\vec{i}}U_{\Lb'}|^2+\Omega^4|\Dbf^{\vec{i}}\Us|^2
\end{align*}
which follows by \eqref{Ulddef}, and assuming in addition to \eqref{geometricacoustical1},
\begin{equation}\label{geometricacoustical2} \left||u|^{-1}\Us, D(U_{\Lb}),  \Db(U_{L'}), \nablas\Us , |u|^2D\varrho, |u|^2\Db\varrho, |u|^2\nablas\varrho \right| \lesssim\varepsilon^{-\delta} |u|^{-1}, \end{equation}
we have 
\begin{align*}
|\tau_{\vec{i}}|\lesssim&\varepsilon^{-\delta}\left( |u||\Dbf^{\vec{i}}\varrho|^2+|u|^{-1}\left(\Omega^4|\Dbf^{\vec{i}}U_{\Lb'}|^2+|\Dbf^{\vec{i}}\Us|^2\right)\right)\\
&+\varepsilon^{-\delta}\left|\Omega^2\left(\varrho^{-1}\mathcal{E}_{0,\vec{i}}\Dbf^{\vec{i}}\varrho+\mathcal{E}_{L,\vec{i}}\Dbf^{\vec{i}}U_{\Lb'}+\mathcal{E}_{\Lb',\vec{i}}\widehat{\Dbf^{\vec{i}}U_{L}}+ \Es_{\vec{i}}\Dbf^{\vec{i}}\Us\right)\right|.
\end{align*}

We should plug in different weight functions for different cases. The first case is $\is=0$. There are two subcases: $i^{\Lb}\le N, i^L=0$ and $i^{\Lb}=N,  i^L= 1$.   We set $w=w_{\vec{i}}=|u|^{-1-2\gammat+2|\vec{i}|}$ where $\gammat>0$ is to be determined, and write according to \eqref{fluiddivergence},
\begin{align*}&D\left(w_{\vec{i}}\Omega^2J_{\Lb',\vec{i}}\D\mu_{\gs}\right)+\Db\left( w_{\vec{i}}J_{L,\vec{i}}\D\mu_{\gs}\right)-2w_{\vec{i}}\divs\left(\Omega^2\Js_{\vec{i}}\right)\D\mu_{\gs}\\
=&\left(((1+2\gammat-2|\vec{i}|)|u|^{-1}+\Omega\tr\chib)w_{\vec{i}} J_{L,\vec{i}}+\Omega\tr\chi w_{\vec{i}}\Omega^2J_{\Lb',\vec{i}}+w_{\vec{i}}\tau_{\vec{i}}\right)\D\mu_{\gs}.\end{align*}
Integrating over $(\mathcal{M}, \D\ub\D u\D\mu_{\gs})$, we have
\begin{align*}&\int_{\Cb_{\ub}}-\int_{\Cb_0}w_{\vec{i}}\Omega^2J_{\Lb',\vec{i}}+\int_{C_u}-\int_{C_{u_0}} w_{\vec{i}}J_{L,\vec{i}}\\
\lesssim&\int_{\mathcal{M}}\left|((1+2\gammat-2|\vec{i}|)|u|^{-1}+\Omega\tr\chib)w_{\vec{i}} J_{L,\vec{i}}+w\Omega^2\Omega\tr\chi J_{\Lb',\vec{i}}+w_{\vec{i}}\tau_{\vec{i}}\right|\end{align*}
The initial data assumptions of Theorem \ref{existencetheorem} imply that the integral on initial null cones are bounded. By Lemma \ref{coerciveness}, for $\varepsilon$ sufficiently small, we then have 
\begin{equation*}\begin{split}&\int_{\Cb_{\ub}}|u|^{-1-2\gammat+2|\vec{i}|}\left(|u|^2|\Dbf^{\vec{i}}\varrho|^2+|u|^{-2}\left(\Omega^4|\Dbf^{\vec{i}} U_{\Lb'}|^2+|\Dbf^{\vec{i}}\Us|^2\right)\right)\\
&+\int_{C_u}|u|^{-1-2\gammat+2|\vec{i}|}\Omega^2\left(|u|^2|\Dbf^{\vec{i}}\varrho|^2+|u|^{-2}\left(\Omega^4|\Dbf^{\vec{i}} U_{\Lb'}|^2+|\Dbf^{\vec{i}}\Us|^2\right)\right)\\
\lesssim&|u|^{-2\gammat}+\varepsilon^{-\delta}\int_{\mathcal{M}}|u|^{-1}\cdot |u|^{-1-2\gammat+2|\vec{i}|}\left(|u|^2|\Dbf^{\vec{i}}\varrho|^2+|u|^{-2}\left(\Omega^4|\Dbf^{\vec{i}} U_{\Lb'}|^2+|\Dbf^{\vec{i}}\Us|^2\right)\right)\\
&+\varepsilon^{-\delta}\int_{\mathcal{M}}\left||u|^{-1-2\gammat+2|\vec{i}|}\Omega^2\left(\varrho^{-1}\mathcal{E}_{0,\vec{i}}\Dbf^{\vec{i}}\varrho+\mathcal{E}_{L,\vec{i}}\Dbf^{\vec{i}}U_{\Lb'}+\mathcal{E}_{\Lb',\vec{i}}\widehat{\Dbf^{\vec{i}}U_{L}}+ \Es_{\vec{i}}\Dbf^{\vec{i}}\Us\right)\right|\\
\lesssim&|u|^{-2\gammat}+\varepsilon^{-\delta}\int_{\mathcal{M}}|u|^{-1}\cdot |u|^{-1-2\gammat+2|\vec{i}|}\left(|u|^2|\Dbf^{\vec{i}}\varrho|^2+|u|^{-2}\left(\Omega^4|\Dbf^{\vec{i}} U_{\Lb'}|^2+|\Dbf^{\vec{i}}\Us|^2\right)\right)\\
&+\varepsilon^{-\delta}\int_{\mathcal{M}}\left||u|^{-1-2\gammat+2|\vec{i}|}\cdot|u|^3\left(|\Omega^2\mathcal{E}_{0,\vec{i}}|^2+|\mathcal{E}_{L,\vec{i}}|^2+|\Omega^4\mathcal{E}_{\Lb',\vec{i}}|^2+|\Omega^2\Es_{\vec{i}}|^2\right)\right|
\end{split}
\end{equation*}
Writing $\int_{\mathcal{M}}=\int_0^{\ub}\int_{\Cb_{\ub'}}\D\ub'$, the first integral can be absorbed when $\ub|u|^{-1}\le\varepsilon$ is sufficiently small. We then have
\begin{equation}\label{fluidestimate1}\begin{split}&\int_{\Cb_{\ub}}|u|^{-1-2\gammat+2|\vec{i}|}\left(|u|^2|\Dbf^{\vec{i}}\varrho|^2+|u|^{-2}\left(\Omega^4|\Dbf^{\vec{i}} U_{\Lb'}|^2+|\Dbf^{\vec{i}}\Us|^2\right)\right)\\
&+\int_{C_u}|u|^{-1-2\gammat+2|\vec{i}|}\Omega^2\left(|u|^2|\Dbf^{\vec{i}}\varrho|^2+|u|^{-2}\left(\Omega^4|\Dbf^{\vec{i}} U_{\Lb'}|^2+|\Dbf^{\vec{i}}\Us|^2\right)\right)\\
\lesssim&|u|^{-2\gammat}+\varepsilon^{-\delta}\int_{\mathcal{M}}|u|^{-1-2\gammat+2|\vec{i}|}\cdot|u|^3\left(|\Omega^2\mathcal{E}_{0,\vec{i}}|^2+|\mathcal{E}_{L,\vec{i}}|^2+|\Omega^4\mathcal{E}_{\Lb',\vec{i}}|^2+|\Omega^2\Es_{\vec{i}}|^2\right)
\end{split}
\end{equation}

The second case is $\is\ne0$ and $i^L=0$, we choose $w=w_{\vec{i}}=\ub^{-2-2\gammat}|u|^{1+2\gammat+2|\vec{i}|}$, and also compute
\begin{align*}&D\left(w_{\vec{i}}\Omega^2J_{\Lb',\vec{i}}\D\mu_{\gs}\right)+\Db\left( w_{\vec{i}}J_{L,\vec{i}}\D\mu_{\gs}\right)-2w_{\vec{i}}\divs\left(\Omega^2\Js_{\vec{i}}\right)\D\mu_{\gs}\\
=&\left((-(1+2\gammat+2|\vec{i}|)|u|^{-1}+\Omega\tr\chib)w_{\vec{i}} J_{L,\vec{i}}+(\Omega\tr\chi-(2+2\gammat)\ub^{-1} )w_{\vec{i}}\Omega^2J_{\Lb',\vec{i}}+w_{\vec{i}}\tau_{\vec{i}}\right)\D\mu_{\gs}.\end{align*}
The term $-(2+2\gammat)\ub^{-1} w_{\vec{i}}\Omega^2J_{\Lb',\vec{i}}<0$, which is not integrable,  can however be dropped since it has a good sign. We can then integrate again over $(\mathcal{M}, \D\ub\D u\D\mu_{\gs})$, and obtain (note that in this case the data on $\Cb_0$ vanishes because of spherical symmetry and the integral on $C_{u_0}$ is bounded again from the assumptions of Theorem \ref{existencetheorem})
\begin{equation*}\begin{split}&\int_{\Cb_{\ub}}\ub^{-2-2\gammat}|u|^{1+2\gammat+2|\vec{i}|}\left(|u|^2|\Dbf^{\vec{i}}\varrho|^2+|u|^{-2}\left(\Omega^4|\Dbf^{\vec{i}} U_{\Lb'}|^2+|\Dbf^{\vec{i}}\Us|^2\right)\right)\\
&+\int_{C_u}\ub^{-2-2\gammat}|u|^{1+2\gammat+2|\vec{i}|}\Omega^2\left(|u|^2|\Dbf^{\vec{i}}\varrho|^2+|u|^{-2}\left(\Omega^4|\Dbf^{\vec{i}} U_{\Lb'}|^2+|\Dbf^{\vec{i}}\Us|^2\right)\right)\\
\lesssim&1+\varepsilon^{-\delta}\int_{\mathcal{M}}|u|^{-1}\cdot \ub^{-2-2\gammat}|u|^{1+2\gammat+2|\vec{i}|}\left(|u|^2|\Dbf^{\vec{i}}\varrho|^2+|u|^{-2}\left(\Omega^4|\Dbf^{\vec{i}} U_{\Lb'}|^2+|\Dbf^{\vec{i}}\Us|^2\right)\right)\\
&+\varepsilon^{-\delta}\int_{\mathcal{M}}\left|\ub^{-2-2\gammat}|u|^{1+2\gammat+2|\vec{i}|}\Omega^2\left(\varrho^{-1}\mathcal{E}_{0,\vec{i}}\Dbf^{\vec{i}}\varrho+\mathcal{E}_{L,\vec{i}}\Dbf^{\vec{i}}U_{\Lb'}+\mathcal{E}_{\Lb',\vec{i}}\widehat{\Dbf^{\vec{i}}U_{L}}+ \Es_{\vec{i}}\Dbf^{\vec{i}}\Us\right)\right|\\
\lesssim&1+\varepsilon^{-\delta}\int_{\mathcal{M}}|u|^{-1}\cdot (\ub|u|^{-1})^{2\gammat-1}\cdot\ub^{-2-2\gammat}|u|^{1+2\gammat+2|\vec{i}|}\left(|u|^2|\Dbf^{\vec{i}}\varrho|^2+|u|^{-2}\left(\Omega^4|\Dbf^{\vec{i}} U_{\Lb'}|^2+|\Dbf^{\vec{i}}\Us|^2\right)\right)\\
&+\varepsilon^{-\delta}\int_{\mathcal{M}}\left|\ub^{-1-4\gammat}|u|^{4\gammat+2|\vec{i}|}\cdot|u|^3\left(|\Omega^2\mathcal{E}_{0,\vec{i}}|^2+|\mathcal{E}_{L,\vec{i}}|^2+|\Omega^4\mathcal{E}_{\Lb',\vec{i}}|^2+|\Omega^2\Es_{\vec{i}}|^2\right)\right|
\end{split}
\end{equation*}
Using $\ub|u|^{-1}\le\varepsilon$ small enough and $\gammat>0$, we have
\begin{equation}\label{fluidestimate2}\begin{split}&\int_{\Cb_{\ub}}\ub^{-2-2\gammat}|u|^{1+2\gammat+2|\vec{i}|}\left(|u|^2|\Dbf^{\vec{i}}\varrho|^2+|u|^{-2}\left(\Omega^4|\Dbf^{\vec{i}} U_{\Lb'}|^2+|\Dbf^{\vec{i}}\Us|^2\right)\right)\\
&+\int_{C_u}\ub^{-2-2\gammat}|u|^{1+2\gammat+2|\vec{i}|}\Omega^2\left(|u|^2|\Dbf^{\vec{i}}\varrho|^2+|u|^{-2}\left(\Omega^4|\Dbf^{\vec{i}} U_{\Lb'}|^2+|\Dbf^{\vec{i}}\Us|^2\right)\right)\\
\lesssim&1+\varepsilon^{-\delta}\int_{\mathcal{M}}\ub^{-1-4\gammat}|u|^{4\gammat+2|\vec{i}|}\cdot|u|^3\left(|\Omega^2\mathcal{E}_{0,\vec{i}}|^2+|\mathcal{E}_{L,\vec{i}}|^2+|\Omega^4\mathcal{E}_{\Lb',\vec{i}}|^2+|\Omega^2\Es_{\vec{i}}|^2\right).
\end{split}\end{equation}

The other choices for $\is$ not mentioned above are not needed to close the bootstrap argument. We are ready to present and prove the estimates for the fluid variables.
\begin{proposition}\label{fluidestimates} Under the bootstrap assumptions \ref{bootstrapA} and \ref{bootstrapB}, we have, for $\is=0$ (including two subcases $i^{\Lb}\le N, i^L=0$ and $i^{\Lb}=N, i^L=1$)
\begin{align*}
|u|^{2\gammat}\int_{\Cb_{\ub}}|u'|^{-3-2\gammat+2|\vec{i}|}\left||u'|^2\Dbf^{\vec{i}} \varrho, \Omega^{2}\Dbf^{\vec{i}} U_{\Lb'}, \Omega^{-2}\Dbf^{\vec{i}}U_{L}, \Dbf^{\vec{i}}\Us\right|^2\lesssim 1,\\
\int_{C_u}|u|^{-3+2|\vec{i}|}\Omega^2\left||u|^2\Dbf^{\vec{i}}\varrho, \Omega^{2}\Dbf^{\vec{i}} U_{\Lb'}, \Omega^{-2}\Dbf^{\vec{i}}U_{L}, \Dbf^{\vec{i}}\Us\right|^2\lesssim 1,
\end{align*}
and for $\vec{i}$ with $1\le\is\le N, i^L=0$, 
\begin{align*}
\ub^{-2-2\gammat}\int_{\Cb_{\ub}}|u'|^{-1+2\gammat+2|\vec{i}|}\left||u'|^2\Dbf^{\vec{i}} \varrho, \Omega^{2}\Dbf^{\vec{i}} U_{\Lb'}, \Omega^{-2}\Dbf^{\vec{i}}U_{L}, \Dbf^{\vec{i}}\Us\right|^2\lesssim 1,\\
\int_{C_u}\ub'^{-2-2\gammat}|u|^{-1+2\gammat+2|\vec{i}|}\Omega^2\left||u|^2\Dbf^{\vec{i}} \varrho, \Omega^{2}\Dbf^{\vec{i}} U_{\Lb'}, \Omega^{-2}\Dbf^{\vec{i}}U_{L}, \Dbf^{\vec{i}}\Us\right|^2\lesssim 1,
\end{align*}
and 
\begin{align*}
\ub^{-2-2\gammat}\int_{\Cb_{\ub}}|u'|^{-1+2\gammat+2(N+1)}\Omega^2\left||u'|^2\nablas^{N+1} \varrho, \Omega^{2}\nablas^{N+1} U_{\Lb'}, \Omega^{-2}\nablas^{N+1}U_{L}, \nablas^{N+1}\Us\right|^2\lesssim 1,\\
\int_{C_u}\ub'^{-2-2\gammat}|u|^{-1+2\gammat+2(N+1)}\Omega^4\left||u|^2\nablas^{N+1}\varrho, \Omega^{2}\nablas^{N+1} U_{\Lb'}, \Omega^{-2}\nablas^{N+1}U_{L}, \nablas^{N+1}\Us\right|^2\lesssim 1,
\end{align*}
for $\ub|u|^{-1}\le\varepsilon$ with $\varepsilon$ is sufficiently small.
\end{proposition}
\begin{proof}
We only need to prove under the assumptions
\begin{align}\label{bootstrapfluidwithoutnablas}
\nonumber |u|^{2\gammat}\int_{\Cb_{\ub}}|u'|^{-3-2\gammat+2|\vec{i}|}\left||u'|^2\Dbf^{\vec{i}} \varrho, \Omega^{2}\Dbf^{\vec{i}} U_{\Lb'}, \Omega^{-2}\Dbf^{\vec{i}}U_{L}, \Dbf^{\vec{i}}\Us\right|^2\lesssim \varepsilon^{-2\delta},\\ \is=i^L=0, i^{\Lb}\le N\ \text{or}\ \is=0, i^{\Lb}=N, i^L=1,
\end{align}
\begin{align}\label{bootstrapfluidwithnablas}
\ub^{-2-2\gammat}\int_{\Cb_{\ub}}|u'|^{-1+2\gammat+2|\vec{i}|}\left||u'|^2\Dbf^{\vec{i}} \varrho, \Omega^{2}\Dbf^{\vec{i}} U_{\Lb'}, \Omega^{-2}\Dbf^{\vec{i}}U_{L}, \Dbf^{\vec{i}}\Us\right|^2\lesssim  \varepsilon^{-2\delta},\ 1\le\is\le N, i^L=0,
\end{align}
\begin{align}\label{bootstrapfluidtopnablas}
\ub^{-2-2\gammat}\int_{\Cb_{\ub}}|u'|^{-1+2\gammat+2(N+1)}\Omega^2\left||u'|^2\nablas^{N+1} \varrho, \Omega^{2}\nablas^{N+1} U_{\Lb'}, \Omega^{-2}\nablas^{N+1}U_{L}, \nablas^{N+1}\Us\right|^2\lesssim \varepsilon^{-2\delta}, 
\end{align}
Applying \eqref{Gronwallu} for $\Dbf^{\vec{i}}\varrho, \Omega^{-2}\Dbf^{\vec{i}} U_{\Lb}, \Omega^2\Dbf^{\vec{i}} U_{L'}$ and $\Dbf^{\vec{i}}\Us$, $\is=0$ and $i^{\Lb}\le N-1, i^L=0$ or $i^{\Lb}=N-1, i^L=1$, for $\nu=0$ and $s=0,0,0,-1$ respectively, we will have
\begin{equation}\label{fluidL2Sbootstrap1}|u|^{-1}\||u|^2\Dbf^{\vec{i}} \varrho, \Omega^{2}\Dbf^{\vec{i}} U_{\Lb'}, \Omega^{-2}\Dbf^{\vec{i}}U_{L}, \Dbf^{\vec{i}}\Us\|_{L^2(S_{\ub,u})}\le |u|^{-|\vec{i}|}\varepsilon^{-\delta}, 1\le j\le N-1,\end{equation}
where we have used the fact that, the above estimates hold without $\varepsilon^{-2\delta}$ on the initial null cone $C_{u_0}$, and 
 \begin{equation}\label{fluidL2Sbootstrap2}
 \begin{split}
 |u|^{-1}\||u|^2\Dbf^{\vec{i}} \varrho, \Omega^{2}\Dbf^{\vec{i}} U_{\Lb'}, \Omega^{-2}\Dbf^{\vec{i}}U_{L}, \Dbf^{\vec{i}}\Us\|_{L^2(S_{\ub,u})}\le \ub^{1+\gammat} |u|^{-1-\gammat-|\vec{i}|}\varepsilon^{-\delta}, \\
 1\le\is\le N, i^L=0\end{split}\end{equation}
 In particular, we will have, by Sobolev inequality, 
 $$|u|^2|\Db\varrho|, |\Db(U_{\Lb}, U_{L'},\Us)|, |u|^2|\nablas\varrho|, |\nablas(U_{\Lb}, U_{L'},\Us)| \lesssim|u|^{-1}\varepsilon^{-\delta}.$$
 By Lemma \ref{DintermsofDbnablas}, we will have
$$|u|^2|D\varrho|, |D(U_{\Lb}, U_{L'}, \Us)|\lesssim|u|^{-1}\varepsilon^{-\delta},$$
 integrating along $D$ directions, together with the (upper and lower) bounds for $\varrho, U_{\Lb}$ and $\Us$, we have for $\varepsilon$ small enough (and $\delta<\frac{1}{4}$),
 \begin{equation}\label{fluidLinfty}|u|^2\varrho, |u|^{-2}\varrho^{-1}, U_{\Lb}, U_{L'}, \ub^{-\frac{3}{4}}|u|^{\frac{3}{4}}|\Us|\lesssim1.\end{equation}
 In view of these estimates, we will be able to apply \eqref{fluidestimate1} and \eqref{fluidestimate2} (which require \eqref{geometricacoustical1}, \eqref{geometricacoustical2}).\\
 
 \paragraph{Case (I): $\is=0$}
 
We consider the spacetime integral in \eqref{fluidestimate1}, that is, we shall consider
$$ |u|^{-1-2\gammat+2|\vec{i}|}\cdot|u|^3\left(|\Omega^2\mathcal{E}_{0,\vec{i}}|^2+|\mathcal{E}_{L,\vec{i}}|^2+|\Omega^4\mathcal{E}_{\Lb',\vec{i}}|^2+|\Omega^2\Es_{\vec{i}}|^2\right)$$
 In this case, these terms have simpler form, in which the terms with $\nablas^{-1}$ disappear. For example, we have, 
 \begin{align*}
\mathcal{E}_{0,\vec{i}}=\Dbf^{\vec{i}}&\left(\varrho(\omega U_{\Lb'}-\etab_A\Us^A-\eta_A\Us^A+\frac{1}{2}\tr\chib' U_{L}+\frac{1}{2}\Omega\tr\chi U_{\Lb'})\right)\\+\sum&\Dbf^{\vec{i}_0}(\Db^{-1}(\Omega^2\zeta^\sharp))(U_{\Lb'}\Dbf^{\vec{i}_1}\nablas\varrho+\varrho\Dbf^{\vec{i}_1}\nablas U_{\Lb'})\\
+&\Dbf^{\vec{i}_0} (D^{-1}(\Omega^2\zeta^\sharp))(U_{L'}\Dbf^{\vec{i}_1}\nablas\varrho+\varrho\Dbf^{\vec{i}_2}\nablas U_{L'})\\
+&\Dbf^{\vec{i}_0}\left(D^{-1}\nablas (\Omega\chih,\Omega\tr\chi),\Db^{-1}\nablas (\Omega\chibh,\Omega\tr\chib)\right)\Dbf^{\vec{i}_1}\varrho\Dbf^{\vec{i}_2}\Us\\
+\sum_{|\vec{i}_1|\ge1}&\Dbf^{\vec{i}_1} U_{\Lb'}\Dbf^{\vec{i}_2}D\varrho+\Dbf^{\vec{i}_1}\varrho\Dbf^{\vec{i}_2} DU_{\Lb'}\\
+&\Dbf^{\vec{i}_1}U_{L'}\Dbf^{\vec{i}_2}\Db\varrho+\Dbf^{\vec{i}_1}\varrho\Dbf^{\vec{i}_2}\Db U_{L'}\\
+&\Dbf^{\vec{i}_1}\Us\Dbf^{\vec{i}_2}\nablas\varrho+\Dbf^{\vec{i}_1}\varrho\Dbf^{\vec{i}_2}\nablas\Us,
\end{align*} 
and $D$ derivatives of the fluid variables should be expressed in terms of $\Db$ and $\nablas$ derivatives using Lemma \ref{DintermsofDbnablas}, except that they have   $\Db^N$ and exactly one $D$ derivative on themselves. On the other hand, if some factor is $\Db^{N+1}$ applying on a fluid variable (possibly appearing in the sixth line),  then we express one $\Db$ in terms of $D$ and $\nablas$ derivatives.  The latter procedure will cause a loss of $\Omega^2$, but since we have an additional $\Omega^2$ before $\mathcal{E}_{0,\vec{i}}$, the estimate can still be closed. 

Then we divide our estimate for $\Omega^2\mathcal{E}_{0,\vec{i}}$ into two groups. The first group involves the terms without $N+1$ order derivatives or $\Db^N$ derivatives. For $N$ is large, each factor can be estimated in $L^\infty_{\ub,u}L^2(S_{\ub,u})$. So we will use bootstrap assumptions \ref{bootstrapA}  and \eqref{fluidL2Sbootstrap1}, \eqref{fluidL2Sbootstrap2} for these estimates.  By direct computation, we have
\begin{align*}
|u|^{-1}\left\|\text{lower order terms in}\ \Omega^2\mathcal{E}_{0,\vec{i}}\right\|_{L^2(S_{\ub,u})}\lesssim |u|^{-3-|\vec{i}|}\varepsilon^{-\max\{3,|\vec{i}|\}\delta}.
\end{align*}
In this estimate we should be careful about the terms (like the last term in the first line, second line, fifth line,  and the sixth line after expressing $\Db$ in terms of $D$ and $\nablas$) that can only be bounded with the additional factor $\Omega^2$. For the other terms, we simply bound $\Omega\lesssim1$. Therefore,
\begin{align*}
\int_{\mathcal{M}}|u|^{-1-2\gammat+2|\vec{i}|}\cdot|u|^3\left|\text{lower order terms in}\ \Omega^2\mathcal{E}_{0,\vec{i}}\right|^2\lesssim\int_{0}^{\ub}\int_{u_0}^u |u'|^{-2-2\gammat}\varepsilon^{-2N\delta}=|u|^{-2\gammat}\varepsilon^{1-2N\delta}
\end{align*}
This term can be bounded when $1-2N\delta>0$.

The second groups involving terms with top order derivatives or $\Db^N$ derivatives. In this case, all derivatives apply in only one factor, and the other factors can be simply bounded in $L^\infty_{\ub,u}L^\infty(S_{\ub,u})$ with a correct scale (maybe with an additional $\varepsilon^{-\delta}$). Taking top order case as an example (that is, $|\vec{i}|=N+1$), we have
\begin{align*}
&\left|\text{top order terms in}\ \Omega^2\mathcal{E}_{0,\vec{i}}\right|\\
\lesssim&\varepsilon^{-\delta}\cdot\left(|u|^{-3}\left(\Omega^2|\Db^{N}DU_{\Lb'}|+\Omega^{-2}|\Db^{N}DU_L|+|\Db^{N}D\Us|\right)+|u|^{-1}|\Db^{N}D\varrho|\right.\\
&+|u|^{-3}\left(\Omega^2|\Db^{N}\nablas U_{\Lb'}|+\Omega^{-2}|\Db^{N}\nablas U_L|+|\Db^{N}\nablas\Us|\right)+|u|^{-1}|\Db^{N}\nablas\varrho|\\
&+|u|^{-2}|\Db^{N}D(\omega,\Omega\eta,\Omega\etab,\Omega\tr\chi,\Omega\tr\chib)|\\
&+|u|^{-2}|\Db^{N}\nablas(\Omega\chih,\Omega\tr\chi)|\\
&\left.+|u|^{-2}|\Db^{N-1}D\nablas( \Omega\chibh,\Omega\tr\chib)|\right).
\end{align*}
Writing $\int_{\mathcal{M}}=\int_0^{\ub}\int_{\Cb_{\ub'}}$ or $\int_{\mathcal{M}}=\int_{u_0}^u\int_{C_{u'}}$ according to different norms the connection coefficients are bounded in (bootstrap assumptions (B) with $\gamma=0$ and  \eqref{bootstrapfluidwithoutnablas}, \eqref{bootstrapfluidwithnablas}),  we have
{\footnotesize
\begin{align*}
&\int_{\mathcal{M}}|u'|^{-1-2\gammat+2(N+1)}\cdot|u|^3\left|\text{top order terms in}\ \Omega^2\mathcal{E}_{0,\vec{i}}\right|^2\\
\lesssim&\varepsilon^{-2\delta}\int_0^{\ub}\int_{\Cb_{\ub'}}|u'|^{-4-2\gammat+2(N+1)}\left(\Omega^4|\Db^{N}DU_{\Lb'}|^2+\Omega^{-4}|\Db^{N}DU_L|^2+|\Db^{N}D\Us|^2+|u'|^4|\Db^{N}D\varrho|^2\right)\\
&+\varepsilon^{-2\delta}\int_0^{\ub}\int_{\Cb_{\ub'}}|u'|^{-4-2\gammat+2(N+1)}\left(\Omega^4|\nablas\Db^{N}U_{\Lb'}|^2+\Omega^{-4}|\nablas\Db^{N}U_L|^2+|\nablas\Db^{N}\Us|^2+|u'|^4|\nablas\Db^{N}\varrho|^2\right)\\
&+\varepsilon^{-2\delta}\int_0^{\ub}\int_{\Cb_{\ub'}}|u'|^{-2-2\gammat+2(N+1)}|\Db^{N}D(\Omega\eta,\Omega\etab, \Omega\tr\chi,\Omega\tr\chib,\omega), \Db^N\nablas(\Omega\chibh,\Omega\chih,\Omega\tr\chi, \Omega\tr\chib), \Db^{N-1}D\nablas(\Omega\chibh,\Omega\tr\chib)|^2\\
\lesssim&\varepsilon^{-2\delta}\int_0^{\ub}|u|^{-1}\underbrace{\int_{\Cb_{\ub'}}|u'|^{-3-2\gammat+2(N+1)}\left(\Omega^4|\Db^{N}DU_{\Lb'}|^2+\Omega^{-4}|\Db^{N}DU_L|^2+|\Db^{N}D\Us|^2+|u'|^4|\Db^{N}D\varrho|^2\right)}_{\lesssim\varepsilon^{-2\delta}|u|^{-2\gammat}}\\
&+\varepsilon^{-2\delta}\int_0^{\ub}|u|^{-1}\underbrace{\int_{\Cb_{\ub'}}|u'|^{-3-2\gammat+2(N+1)}\left(\Omega^4|\nablas\Db^{N}U_{\Lb'}|^2+\Omega^{-4}|\nablas\Db^{N}U_L|^2+|\nablas\Db^{N}\Us|^2+|u'|^4|\nablas\Db^{N}\varrho|^2\right)}_{\lesssim\varepsilon^{-2\delta}\ub^{2+2\gammat}|u|^{-2-4\gammat}}\\
&+\varepsilon^{-2\delta}\int_0^{\ub}|u|^{-1}\underbrace{\int_{\Cb_{\ub'}} |u'|^{-1-2\gammat+2(N+1)}|\Db^{N}D(\Omega\eta,\Omega\etab, \Omega\tr\chi,\Omega\tr\chib, \omega), \Db^N\nablas( \Omega\chih,\Omega\tr\chi ), \Db^{N-1}D\nablas(\Omega\chibh,\Omega\tr\chib)|^2}_{\lesssim\varepsilon^{-2\delta}|u|^{-2\gammat}}\\
\lesssim&|u|^{-2\gammat}\varepsilon^{1-4\delta}.
\end{align*}}
The estimate for $\Omega^2\mathcal{E}_{0,\vec{i}}$ is done. 

The estimates for $\Omega^4\mathcal{E}_{\Lb',\vec{i}}, \Omega^2\Es_{\vec{i}}$, and therefore $\mathcal{E}_{L,\vec{i}}$  with the same $\vec{i}$ are the essentially the same (maybe more factors for some terms). But one should also note that in $\Omega^2\Es_{\vec{i}}$ we may encounter $D(\Omega\chih)$ without $\Db$ applying on it, which is the worse term, behaving like $\ub^{-1+\gamma}|u|^{-1-\gamma}$ (by bootstrap assumpitons \ref{bootstrapA}). After integration, the estimate can be closed if $\gamma>0$ and beside this, $D(\Omega\chih)$ appears together with $\Us$ so it should have more room for the choice of $\gammat$. Combining all estimates above for lower order and top order, for $\varepsilon$ sufficiently small, we have (by \eqref{fluidestimate1})
\begin{equation}\label{fluidestimateresult1}\begin{split}&\int_{\Cb_{\ub}}|u|^{-1-2\gammat+2|\vec{i}|}\left(|u|^2|\Dbf^{\vec{i}}\varrho|^2+|u|^{-2}\left(\Omega^4|\Dbf^{\vec{i}} U_{\Lb'}|^2+|\Dbf^{\vec{i}}\Us|^2\right)\right)\\
&+\int_{C_u}|u|^{-1-2\gammat+2|\vec{i}|}\Omega^2\left(|u|^2|\Dbf^{\vec{i}}\varrho|^2+|u|^{-2}\left(\Omega^4|\Dbf^{\vec{i}} U_{\Lb'}|^2+|\Dbf^{\vec{i}}\Us|^2\right)\right)\lesssim|u|^{-2\gammat}.
\end{split}
\end{equation}
Note that the estimates \eqref{fluidL2Sbootstrap1} will then hold with $\varepsilon^{-\delta}$ replaced by $1$. \\

 \paragraph{Case (II): $i^L=0, 1\le\is\le N$}
We shall consider the spacetime integral in \eqref{fluidestimate2}. We still take $\Omega^2\mathcal{E}_{0,\vec{i}}$ as an example and the others terms are similar. We write
\begin{align*}
\mathcal{E}_{0,\vec{i}}=&\Dbf^{\vec{i}}\left(\varrho(\omega U_{\Lb'}-\etab_A\Us^A-\eta_A\Us^A+\frac{1}{2}\tr\chib' U_{L}+\frac{1}{2}\Omega\tr\chi U_{\Lb'})\right)\\
+\sum&\Dbf^{\vec{i}_0}(\Db^{-1}(\Omega^2\zeta^\sharp), \nablas^{-2}\nablas(\Omega\chih,\Omega\tr\chi))(U_{\Lb'}\Dbf^{\vec{i}_1}\nablas\varrho+\varrho\Dbf^{\vec{i}_1}\nablas U_{\Lb'})\\
+&\Dbf^{\vec{i}_0} ( \nablas^{-2}\nablas(\Omega\chibh, \Omega\tr\chib))(U_{L'}\Dbf^{\vec{i}_1}\nablas\varrho+\varrho\Dbf^{\vec{i}_2}\nablas U_{L'})\\
+&\Dbf^{\vec{i}_0}\left(\Db^{-1}\nablas (\Omega\chibh,\Omega\tr\chib), \nablas^{-1}\Ks\right)\Dbf^{\vec{i}_1}\varrho\Dbf^{\vec{i}_2}\Us\\
+\sum_{|\vec{i}_1|\ge1}&\Dbf^{\vec{i}_1} U_{\Lb'}\Dbf^{\vec{i}_2}D\varrho+\Dbf^{\vec{i}_1}\varrho\Dbf^{\vec{i}_2} DU_{\Lb'}\\
+&\Dbf^{\vec{i}_1}U_{L'}\Dbf^{\vec{i}_2}\Db\varrho+\Dbf^{\vec{i}_1}\varrho\Dbf^{\vec{i}_2}\Db U_{L'}\\
+&\Dbf^{\vec{i}_1}\Us\Dbf^{\vec{i}_2}\nablas\varrho+\Dbf^{\vec{i}_1}\varrho\Dbf^{\vec{i}_2}\nablas\Us
\end{align*}

For lower order terms, we use \eqref{fluidL2Sbootstrap1}, \eqref{fluidL2Sbootstrap2} together with bootstrap assumptions \ref{bootstrapA}. The important thing is to check that, for every term in the expression, there is at least one factor (with a correct weight by scaling consideration) that is estimated by $(\ub|u|^{-1})^{\gamma}$. Then the proof can be done for $\gammat$ small enough.

To see this, note that except the term like $\Db^{j_0}\Ks\Db^{j_1}\varrho\Db^{j_2}\Us, j_0+j_1+j_2\le N, j_0\le N-1$, which appears when $\is=1$, all terms contain at least one factor to which at least one $\nablas$ applies.  By direct computation, we have
\begin{align*}
|u|^{-1}\left\|\text{lower order terms in}\ \Omega^2\mathcal{E}_{0,\vec{i}}\ \text{except $\Omega^2\Db^{j_0}\Ks\Db^{j_1}\varrho\Db^{j_2}\Us$}\right\|_{L^2(S_{\ub,u})}\lesssim (\ub|u|^{-1})^{\gamma}|u|^{-3-|\vec{i}|}\varepsilon^{-N\delta}.
\end{align*}
In this estimate we should also be careful whether the additional factor $\Omega^2$ is sufficient. To estimate $\Omega^2\Db^{j_0}\Ks\Db^{j_1}\varrho\Db^{j_2}\Us$, we will need an improved estimate for $\Db^j\Us$ for $1\le j\le N$. A similar computation to Lemma \ref{DintermsofDbnablas} gives an expression of $D\Db^j\Us$ in terms of $\Db^{j+1}\Us$ and $\nablas\Db^j\Us$ together with lower order terms. By the estimate \eqref{fluidestimateresult1}, and \eqref{fluidL2Sbootstrap1},  \eqref{fluidL2Sbootstrap2}, we have
$$\int_{C_u}|u|^{-3+2(j+1)}\Omega^2|D\Db^{j}\Us|^2\lesssim 1, j\le N.$$
Applying \eqref{Gronwallub} for $\Db^j\Us$ (note that we have zero initial data) we will have
\begin{equation}\label{Usimproved}
|u|^{-1}\|\Db^j\Us\|_{L^2(S_{\ub,u})}\lesssim \Omega^{-1}\ub^{\frac{1}{2}}|u|^{-\frac{1}{2}-j}, j\le N
\end{equation}
which has an $\Omega^{-1}$. For $\gamma<\frac{1}{2}$, we will have
\begin{align*}
|u|^{-1}\left\|\text{lower order terms in}\ \Omega^2\mathcal{E}_{0,\vec{i}}\right\|_{L^2(S_{\ub,u})}\lesssim (\ub|u|^{-1})^{\gamma}|u|^{-3-|\vec{i}|}\varepsilon^{-N\delta}.
\end{align*}
 and therefore
 \begin{align*}
&\int_{\mathcal{M}}\ub^{-1-4\gammat}|u|^{4\gammat+2|\vec{i}|}\cdot|u|^3|\text{lower order terms in}\ \Omega^2\mathcal{E}_{0,\vec{i}}|^2\\
\lesssim&\int_{0}^{\ub}\int_{u_0}^u \ub'^{-1-4\gammat+2\gamma}|u|^{-1+4\gammat-2\gamma}\varepsilon^{-2N\delta}=(\ub|u|^{-1})^{2\gamma-4\gammat}\varepsilon^{-2N\delta}\lesssim1
\end{align*}
if $2\gamma-4\gammat-2N\delta>0$.

For top order terms, all derivatives apply in only one factor, and the other factor can be simply bounded in $L^\infty_{\ub,u}L^\infty(S_{\ub,u})$ with a correct scale (maybe with an additional $\varepsilon^{-\delta}$). In this case, we simply express $D$ derivative to the fluid variables in terms of $\Db$ and $\nablas$ derivatives, except that when we have $\Db^{N+1}$ applying to a fluid variable, we express one $\Db$ in terms of $D$ and $\nablas$. There are two special cases we need to concern. When $\is=1, i^{\Lb}=N$, there are top order terms that are $\Db^{N+1}$ applying to some fluid variable (and we will express it in terms of $\Db^{N}D$ and $\Db^N\nablas$), and then there will be exactly one $\nablas$ applying to fluid variables (estimated by $(\ub|u|^{-1})^{1+\gammat}$ multiple its scale). When $\is=N, i^{\Lb}=1$, there are top order terms containing a factor that is $\nablas^{N+1}$ applying to the fluid variables or connection coefficients, then we can check that there will be an additional $\Omega$ in each such term (in this case, we will have an $\Us$ as a factor so we will have extra $\Omega$). Also, we will pay special attention to the terms involving Gauss curvature $\Ks$. In summarize, we will have (for $|\vec{i}|=N+1$)
\begin{align*}
&\left|\text{top order terms in}\ \Omega^2\mathcal{E}_{0,\vec{i}}\right|\\
\lesssim&\varepsilon^{-\delta}\cdot\left((\ub|u|^{-1})^{1+\gammat}|u|^{-3}\left(\Omega^2|\Db^{N}DU_{\Lb'}|+\Omega^{-2}|\Db^{N}DU_L|+|\Db^{N}D\Us|+|u|^{2}|\Db^{N}D\varrho|\right)\right.\\
&+\sum_{|\vec{i}|=N+1, 1\le \is\le N, i^L=0}|u|^{-3}\left(\Omega^2|\Dbf^{\vec{i}}U_{\Lb'}|+\Omega^{-2}|\Dbf^{\vec{i}}U_L|+|\Dbf^{\vec{i}}\Us|+|u|^{2}|\Dbf^{\vec{i}}\varrho|\right)\\
&+|u|^{-3}\Omega\left(\Omega^2|\nablas^{N+1}U_{\Lb'}|+\Omega^{-2}|\nablas^{N+1}U_L|+|\nablas^{N+1}\Us|+|u|^{2}|\nablas^{N+1}\varrho|\right)\\
&+\sum_{|\vec{i}|=N+1, 1\le \is\le N, i^L=0}|u|^{-2}|\Dbf^{\vec{i}}(\omega,\Omega\eta,\Omega\etab,\Omega\tr\chi,\Omega\tr\chib)|\\
&+|u|^{-2}\Omega|\nablas^{N+1}(\Omega\chibh,\Omega\tr\chib)|\\
&\left.+(\ub|u|^{-1})^{\frac{1}{2}}|u|^{-2}|\Db^{N}\Ks, \sum_{|\vec{i}|=N, 1\le\is\le N-1, i^L=0}\Dbf^{\vec{i}}(\Omega\Ks)|\right).
\end{align*}
We then have
{\tiny
\begin{align*}
&\int_{\mathcal{M}}\ub^{-1-4\gammat}|u|^{4\gammat+2(N+1)}\cdot|u|^3|\text{top order terms in}\ \Omega^2\mathcal{E}_{0,\vec{i}}|^2\\
\lesssim&\varepsilon^{-2\delta}\int_0^{\ub}\int_{\Cb_{\ub'}}\ub'^{1-2\gammat}|u'|^{-5+2\gammat+2(N+1)}\left(\Omega^4|\Db^{N}DU_{\Lb'}|^2+\Omega^{-4}|\Db^{N}DU_L|^2+|\Db^{N}D\Us|^2+|u'|^4|\Db^{N}D\varrho|^2\right)\\
&+\varepsilon^{-2\delta}\int_0^{\ub}\int_{\Cb_{\ub'}}\ub'^{-1-4\gammat}|u'|^{-3+4\gammat+2(N+1)}\sum_{|\vec{i}|=N+1, 1\le \is\le N, i^L=0}\left(\Omega^4|\Dbf^{\vec{i}}U_{\Lb'}|^2+\Omega^{-4}|\Dbf^{\vec{i}}U_L|^2+|\Dbf^{\vec{i}}\Us|^2+|u'|^4|\Dbf^{\vec{i}}\varrho|^2\right)\\
&+\varepsilon^{-2\delta}\int_0^{\ub}\int_{\Cb_{\ub'}}\ub'^{-1-4\gammat}|u'|^{-3+4\gammat+2(N+1)}\Omega^2\left(\Omega^4|\nablas^{N+1}U_{\Lb'}|^2+\Omega^{-4}|\nablas^{N+1}U_L|^2+|\nablas^{N+1}\Us|^2+|u'|^4|\nablas^{N+1}\varrho|^2\right)\\
&+\varepsilon^{-2\delta}\int_0^{\ub}\int_{\Cb_{\ub'}}\ub'^{-1-4\gammat}|u'|^{-1+4\gammat+2(N+1)}\sum_{|\vec{i}|=N+1, 1\le \is\le N, i^L=0}|\Dbf^{\vec{i}}(\Omega\eta,\Omega\etab,\Omega\tr\chi,\Omega\tr\chib,\omega), \Omega\nablas^{N+1}(\Omega\chibh,\Omega\tr\chib)|^2\\
&+\varepsilon^{-2\delta}\int_{0}^{\ub}\int_{\Cb_{\ub'}}\ub'^{-4\gammat}|u'|^{-2+4\gammat+2(N+1)}|\Db^{N}\Ks, \sum_{|\vec{i}|=N, 1\le\is\le N-1, i^L=0}\Dbf^{\vec{i}}(\Omega\Ks)|^2\\
\lesssim&\varepsilon^{-2\delta}\int_0^{\ub}\ub'^{1-2\gammat}|u|^{-2+4\gammat}\underbrace{\int_{\Cb_{\ub'}}|u'|^{-3-2\gammat+2(N+1)}\left(\Omega^4|\Db^{N}DU_{\Lb'}|^2+\Omega^{-4}|\Db^{N}DU_L|^2+|\Db^{N}D\Us|^2+|u'|^4|\Db^{N}D\varrho|^2\right)}_{\lesssim|u|^{-2\gammat}}\\
&+\varepsilon^{-2\delta}\int_0^{\ub}\ub'^{1-2\gammat}|u|^{-2+2\gammat}\underbrace{\int_{\Cb_{\ub'}}\ub'^{-2-2\gammat}|u'|^{-1+2\gammat+2(N+1)}\sum_{|\vec{i}|=N+1, 1\le \is\le N, i^L=0}\left(\Omega^4|\Dbf^{\vec{i}}U_{\Lb'}|^2+\Omega^{-4}|\Dbf^{\vec{i}}U_L|^2+|\Dbf^{\vec{i}}\Us|^2+|u'|^4|\Dbf^{\vec{i}}\varrho|^2\right)}_{\lesssim\varepsilon^{-2\delta}}\\
&+\varepsilon^{-2\delta}\int_0^{\ub}\ub'^{1-2\gammat}|u|^{-2+2\gammat}\underbrace{\int_{\Cb_{\ub'}}\ub'^{-2-2\gammat}|u'|^{-1+2\gammat+2(N+1)}\Omega^2\left(\Omega^4|\nablas^{N+1}U_{\Lb'}|^2+\Omega^{-4}|\nablas^{N+1}U_L|^2+|\nablas^{N+1}\Us|^2+|u'|^4|\nablas^{N+1}\varrho|^2\right)}_{\lesssim\varepsilon^{-2\delta}}\\
&+\varepsilon^{-2\delta}\int_0^{\ub}\ub'^{-1}\underbrace{\int_{\Cb_{\ub'}}\ub'^{-4\gammat}|u'|^{-1+4\gammat+2(N+1)}\sum_{|\vec{i}|=N+1, 1\le \is\le N, i^L=0}|\Dbf^{\vec{i}}(\Omega\eta,\Omega\etab,\Omega\tr\chi,\Omega\tr\chib,\omega), \Omega\nablas^{N+1}(\Omega\chibh,\Omega\tr\chib)|^2}_{(\ub'|u|^{-1})^{2(\gamma-2\gammat)}\varepsilon^{-2\delta}}\\
&+\varepsilon^{-2\delta}\int_{0}^{\ub}\ub'^{-4\gammat}|u|^{-1+6\gammat}\underbrace{\int_{\Cb_{\ub'}}|u'|^{-1-2\gammat+2(N+1)}|\Db^{N}\Ks, \sum_{|\vec{i}|=N, 1\le\is\le N-1, i^L=0}\Dbf^{\vec{i}}(\Omega\Ks)|^2}_{\lesssim |u|^{-2\gammat}\varepsilon^{-2\delta}} \\
\lesssim&\varepsilon^{2(\gamma-2\gammat)-4\delta}.
\end{align*}}
Here we need $6\gammat<1$. Here the term involving $\Dbf^{\vec{i}}\omega$ is the most delicate term. Moreover, $\nablas^{N+1}\omega$ does not appear in Case (II) is crucial because we will have no additional $\Omega$ for terms involving $\nablas^{N+1}\omega$. 

The estimates for $\Omega^4\mathcal{E}_{\Lb',\vec{i}}, \Omega^2\Es_{\vec{i}}$, and therefore $\mathcal{E}_{L,\vec{i}}$ for $1\le\is\le N, i^L=0$ are the essentially the same (maybe more factors for some terms). Combining all estimates above for lower order and top order, for $\varepsilon$ sufficiently small, we have (by \eqref{fluidestimate2})
\begin{equation}\label{fluidestimateresult2}\begin{split}&\int_{\Cb_{\ub}}\ub^{-2-2\gammat}|u|^{1+2\gammat+2|\vec{i}|}\left(|u|^2|\Dbf^{\vec{i}}\varrho|^2+|u|^{-2}\left(\Omega^4|\Dbf^{\vec{i}} U_{\Lb'}|^2+|\Dbf^{\vec{i}}\Us|^2\right)\right)\\
&+\int_{C_u}\ub'^{-2-2\gammat}|u|^{1+2\gammat+2|\vec{i}|}\Omega^2\left(|u|^2|\Dbf^{\vec{i}}\varrho|^2+|u|^{-2}\left(\Omega^4|\Dbf^{\vec{i}} U_{\Lb'}|^2+|\Dbf^{\vec{i}}\Us|^2\right)\right)\lesssim 1.
\end{split}
\end{equation}
Note that the estimates \eqref{fluidL2Sbootstrap2} will then hold with $\varepsilon^{-\delta}$ replaced by $1$. \\

 \paragraph{Case (III): $\is=N+1, i^L=i^{\Lb}=0$}
The last case $\is=N+1, i^L=i^{\Lb}=0$ is similar to Case (II). One difference is the weight function $w$ is chosen to be
$$w=\Omega^2\ub^{-2-2\gammat}|u|^{1+2\gammat+2(N+1)},$$
which has an additional $\Omega^2$ as compared to $w_{\vec{i}}$ with $1\le \is\le N$, for we will have an additional $\Omega^2$ in the spacetime integral term. Another difference is $\nablas^{N+1}(\Omega\chih,\omega)$ should be placed in $L^2(C_{u})$ instead of $L^2(\Cb_{\ub})$ and the estimates of spacetime integral should be modified slightly. 

\end{proof}

\section{The estimates for curvature components}

We are going to prove the estimates for curvature components. We would like to introduce additional bootstrap assumptions for the lower order connection coefficients which improves bootstrap assumption \ref{bootstrapA} for certain quantities:
 \renewcommand{\theboot}{(A')} 
\begin{boot}\label{bootstrapA'}
For $i\le N$:
\begin{align*}|u|^{-1}\|\nablas^i   ( \Omega\eta,\Omega\etab,\Omega\chibh,\widetilde{\Omega\tr\chib},\widetilde{\omegab})\|_{L^2(S_{\ub,u})}\le  \ub |u|^{-2-i}\varepsilon^{-\delta},\end{align*}
For $i+j\le N$:
\begin{align*}|u|^{-1}\|\nablas^i  \Db^j  (  \Omega\chibh)\|_{L^2(S_{\ub,u})}\le  \ub |u|^{-2-(i+j)}\varepsilon^{-\delta},\end{align*}
\end{boot}
Here a tilde of some quantity means the difference from its value on $\Cb_0$. 

\begin{proposition}\label{curvature}
Under  the bootstrap assumptions \ref{bootstrapA} and \ref{bootstrapA'}, we have
\begin{equation}\label{curvaturebound}\underline{\mathcal{R}},\mathcal{R},\mathcal{R}(\mathcal{M})\lesssim1,\end{equation}
where
\begin{align*}\underline{\mathcal{R}}^2=&\sup_{0\le i\le N}\{\ub^{-1-2\gamma}\int_{\Cb_{\ub}} |u|^{2+2\gamma+2i}|\nablas^i(\Omega^2\widetilde{\rho},\Omega^2 \sigma)|^2\\
&+\ub^{-2-2\gamma}\int_{\Cb_{\ub}} |u|^{3+2\gamma+2i}|\nablas^i(\Omega^2 \betab,\Omega^2\alphab)|^2
\end{align*}
\begin{align*}\mathcal{R}^2=&\sup_{0\le i\le N}\{\int_{C_u}\ub^{1-2\gamma}|u|^{2\gamma+2i}|\nablas^i(\Omega^2\alpha)|^2+\int_{C_u}\ub^{-1-2\gamma}|u|^{2+2\gamma+2i}|\nablas^i(\Omega^2\beta)|^2\\
&+\int_{C_u} \ub^{-2-2\gamma}|u|^{3+2\gamma+2i} |\nablas^i(\Omega^2 \widetilde{\rho},\Omega^2 \sigma)|^2\end{align*}
and
\begin{align*}
\mathcal{R}(\mathcal{M})=&\sup_{0\le i\le N}\{\int_{\mathcal{M}}\ub^{1-2\gamma}|u|^{-1+2\gamma+2i}|\nablas^i(\Omega^2\alpha)|^2+\int_{\mathcal{M}}\ub^{-1-2\gamma}|u|^{1+2\gamma+2i}|\nablas^i(\Omega^2\beta)|^2\\
&+\int_{\mathcal{M}}\ub^{-2-2\gamma}|u|^{2+2\gamma+2i}|\nablas^i(\Omega^2\widetilde{\rho},\Omega^2\sigma)|^2\\
&+\int_{\mathcal{M}}\ub^{-3-2\gamma}|u|^{3+2\gamma+2i}|\nablas^i(\Omega^2\betab,\Omega^2\alphab)|^2\}.
\end{align*}
\end{proposition}
\begin{remark}
Note that the above energy integrals on the initial cone $C_{u_0}$ or $\Cb_0$ are finite from the assumptions of Theorem \eqref{existencetheorem}. One can derive this as in Chapter 2 of \cite{Chr08}, which will be omitted. The same remark applies proof of Propositions \ref{improvebootA}, \ref{improvebootB} below. 
\end{remark}
\begin{proof}
We start  by assuming
\begin{equation}\label{bootstrapR}\mathcal{R}, \mathcal{R}(\mathcal{M})\le\varepsilon^{-\delta}.\end{equation}
It is easy to have
\begin{equation}\label{rhodaggeritself}|\Omega^2\rho|\lesssim|u|^{-2}\end{equation}
We also note that, by bootstrap assumption \ref{bootstrapA'} and Sobolev inequality \eqref{Sobolev}, when $\ub|u|^{-1}\le \varepsilon$ is sufficiently small, 
$$|\widetilde{\Omega\tr\chib}|\lesssim  \varepsilon^{-\delta}\ub|u|^{-2}\le\varepsilon^{1-\delta}|u|^{-1},$$
and hence (note also that $\Omega\tr\chib(0,u)=-2|u|^{-1}$) 
 \begin{equation}\label{trchib<0}\Omega\tr\chib<0\end{equation}
and
$$|\Omega\tr\chib|\lesssim|u|^{-1}.$$
Similarly, for $\varepsilon$ small enough, noting that  $\omegab(0,u)\approx -|u|^{-1}$,
 \begin{equation}\label{trchibomegab<0}\omegab<0.\end{equation}
 and also 
 $$|\omegab|\lesssim|u|^{-1}.$$

We will  need the estimates for the components of Weyl current, which can be expressed in terms of the fluid variables in a schematic form as follows:
\begin{align*}
\Xi_A=&\nabla_4(\varrho U_A U_4)+\nabla_A(\varrho U_4U_4),\\
\Xib_A=&\nabla_3(\varrho U_A U_3)+\nabla_A(\varrho U_3U_3),\\
\Lambda=&\nabla_3(\varrho U_4U_4)+\nabla_4(\varrho U_3U_4)+\nabla_4\varrho,\\
\Lambdab=&\nabla_4(\varrho U_3U_3)+\nabla_3(\varrho U_3U_4)+\nabla_3\varrho,\\
K=&\epsilons^{AB}\nabla_A(\varrho U_BU_4),\\
\Kb=&\epsilons^{AB}\nabla_A(\varrho U_BU_3),\\
I_A=&\nabla_4(\varrho U_AU_3)+\nabla_A(\varrho U_4U_3)+\nabla_A\varrho,\\
\Ib_A=&\nabla_3(\varrho U_AU_4)+\nabla_A(\varrho U_4U_3)+\nabla_A\varrho,\\
\Theta_{AB}=&\nabla_4(\varrho U_AU_B)+\nabla_A(\varrho U_4U_B),\\
\Thetab_{AB}=&\nabla_3(\varrho U_AU_B)+\nabla_A(\varrho U_3U_B).
\end{align*}
From Proposition \ref{fluidestimates} and the consequent lower order estimates \eqref{fluidL2Sbootstrap1}, \eqref{fluidL2Sbootstrap2} (without $\varepsilon^{-\delta}$) and \eqref{fluidLinfty}, and expressing $D$ derivative in terms of $\Db$ and $\nablas$ as in Lemma \ref{DintermsofDbnablas},  we will have
\begin{equation}\label{curvature-weylcurrentbound}
\begin{split}
\sup_{0\le i\le N}\int_{\mathcal{M}}\ub^{-2-2\gamma}|u|^{4+2\gamma+2i}|\nablas^i(\Omega^3(\Theta, I ,   \Ib ,\Xib,\Thetab,\widetilde{\Lambdab}, \Kb, \Xi, K))|^2\lesssim 1\\
\sup_{0\le i\le N}\int_{\mathcal{M}}\ub^{-\frac{1}{2}-2\gamma}|u|^{\frac{5}{2}+2\gamma+2i}|\nablas^i(\Omega^3\Lambda)|^2\lesssim 1\\
\end{split}
\end{equation}
as long as  $\gamma<\frac{1}{4}$. Note that  $\Lambda$ obeys a worse estimate (in fact just for $i=0$) simply because there are no angular components in its expression.

For $\alpha$-$\beta$ pair, we use the equations \eqref{Dbalpha} and \eqref{Dbeta}
\begin{align*}
&\Dbh\alpha-\frac{1}{2}\Omega\tr\chib \alpha+2\omegab\alpha+\Omega\{-\nablas\tensor\beta -(4\eta+\zeta)\tensor \beta+3\chih \rho+3{}^*\chih \sigma\}=-2\Omega\Theta,\\
&D\beta+\frac{3}{2}\Omega\tr\chi\beta-\Omega\chih\cdot\beta-\omega\beta-\Omega\{\divs\alpha+(\etab+2\zeta)\cdot\alpha\}=2\Omega\Xi,\end{align*}
we compute for $0\le i\le N$,
\begin{align*}
&\Db\left(\ub^{1-2\gamma}|u|^{2\gamma+2i}|\nablas^i(\Omega^2\alpha)|^2\D\mu_{\gs}\right)+D\left(2\ub^{1-2\gamma}|u|^{2\gamma+2i}|\nablas^i(\Omega^2\beta)|^2\D\mu_{\gs}\right)\\
&=\ub^{1-2\gamma}|u|^{2\gamma+2i}\divs(\nablas^i(\Omega^2\alpha)\cdot\nablas^i(\Omega^2\beta))\D\mu_{\gs}\\
&\underline{-2\gamma|u|^{-1}\left(\ub^{1-2\gamma}|u|^{2\gamma+2i}|\nablas^i(\Omega^2\alpha)|^2\D\mu_{\gs}\right)}\\&+\ub^{1-2\gamma}|u|^{2\gamma+2i}\tau_{0}\D\mu_{\gs},
\end{align*}
where $\tau_0$ consists of terms from lower order terms and the commutators of $D,\Db,\nablas$ with $\nablas^i$. The choices of the power of $|u|$ and $\Omega$ are made such that $\tau_0$ does not contain terms like $(\Omega\tr\chib,\omegab)|\nablas^i(\Omega^2\alpha)|^2$ and the underlined term above has a negative sign. Integrating over $(\mathcal{M}, \D\ub\D u\D\mu_{\gs})$, we have
\begin{align*}
&\int_{C_u}\ub^{1-2\gamma}|u|^{2\gamma+2i}|\nablas^i(\Omega^2\alpha)|^2+\int_{\Cb_{\ub}}2\ub^{1-2\gamma}|u|^{2\gamma+2i}|\nablas^i(\Omega^2\beta)|^2\\&+2\gamma\int_{\mathcal{M}}\ub^{1-2\gamma}|u|^{-1+2\gamma+2i}|\nablas^i(\Omega^2\alpha)|^2\\
\lesssim&1+\int_{\mathcal{M}}\ub^{1-2\gamma}|u|^{2\gamma+2i}|\tau_{0}|.
\end{align*}
The spacetime integrant $\ub^{1-2\gamma}|u|^{2\gamma+2i}\tau_0$ can be decomposed as sum of the following terms:
\begin{equation}\label{nonlinearterm1}\underbrace{\ub|u|^{-1}}_{\le\varepsilon}\sum_{i_1+i_2=i}\underbrace{\ub^{-1}|u|^{2+i_1}\nablas^{i_1}(\Omega\chibh,\widetilde{\Omega\tr\chib})}_{|u|^{-1}\|\cdot\|_{L^2(S_{\ub,u})}\lesssim\varepsilon^{-\delta}\ \text{by \ref{bootstrapA'}}}\cdot\underbrace{\ub^{\frac{1}{2}-\gamma}|u|^{-\frac{1}{2}+\gamma+i_2}\nablas^{i_2}(\Omega^2\alpha)}_{\|\cdot\|_{L^2(\mathcal{M})}\lesssim \mathcal{R}({\mathcal{M}})\le\varepsilon^{-\delta}}\cdot\underbrace{\ub^{\frac{1}{2}-\gamma} |u|^{-\frac{1}{2}+\gamma+i}\nablas^i(\Omega^2\alpha)}_{\|\cdot\|_{L^2(\mathcal{M})}\lesssim \mathcal{R}({\mathcal{M}})\le\varepsilon^{-\delta}},\end{equation}
$$(\ub|u|^{-1})^{\frac{3}{2}}\sum_{\substack{i_1+i_2=i}}|u|^{1+i_1}\nablas^{i_1}(\Omega\chih)\cdot\ub^{-1-\gamma}|u|^{1+\gamma+i_2}\nablas^{i_2}(\Omega^2\widetilde{\rho},\Omega^2\sigma)\cdot\ub^{\frac{1}{2}-\gamma} |u|^{-\frac{1}{2}+\gamma+i}\nablas^i(\Omega^2\alpha),$$
\begin{equation}\label{nonlineartermrhoitself}\underbrace{(\ub|u|^{-1})^{\frac{1}{2}-\gamma}|u|^{-2}}_{\|\cdot\|_{L^2(\mathcal{M})}\lesssim(\ub|u|^{-1})^{1-\gamma}}\cdot\underbrace{ |u|^{1+i}\nablas^{i}(\Omega\chih)}_{|u|^{-1}\|\cdot\|_{L^2(S_{\ub,u})}\le\varepsilon^{-\delta}\ \text{by \ref{bootstrapA}}}\cdot \underbrace{|u|^2\Omega^2\rho}_{|\cdot|\lesssim1}\cdot\underbrace{\ub^{\frac{1}{2}-\gamma} |u|^{-\frac{1}{2}+\gamma+i}\nablas^i(\Omega^2\alpha)}_{\|\cdot\|_{L^2(\mathcal{M})}},\end{equation}
$$\ub|u|^{-1}\sum_{i_1+i_2=i}\ub|u|^{i_1}\nablas^{i_1}(\Omega\chih,\Omega\tr\chi,\omega,\ub^{-1})\cdot\ub^{-\frac{1}{2}-\gamma}|u|^{\frac{1}{2}+\gamma+i_2}\nablas^{i_2}(\Omega^2\beta)\cdot\ub^{-\frac{1}{2}-\gamma}|u|^{\frac{1}{2}+\gamma+i}\nablas^i(\Omega^2\beta),$$
$$(\ub|u|^{-1})^2\sum_{\substack{i_1+i_2+i_3=2i\\ i_1,i_2,i_3\le i}}\ub^{-1}|u|^{2+i_1}\nablas^{i_1}(\Omega\eta,\Omega\etab)\cdot\ub^{-\frac{1}{2}-\gamma}|u|^{\frac{1}{2}+\gamma+i_2}\nablas^{i_2}(\Omega^2\beta)\cdot\ub^{\frac{1}{2}-\gamma} |u|^{-\frac{1}{2}+\gamma+i_3}\nablas^{i_3}(\Omega^2\alpha),$$
\begin{equation}\label{nonlineartermKs}\ub|u|^{-1}\sum_{\substack{i_1+i_2+i_3=2i-1\\ i_1+1,i_2,i_3\le i}}\underbrace{|u|^{2+i_1}\nablas^{i_1}(\Omega \Ks)}_{|u|^{-1}\|\cdot\|\lesssim 1\ \text{by \ref{bootstrapA}}}\cdot\ub^{-\frac{1}{2}-\gamma}|u|^{\frac{1}{2}+\gamma+i_2}\nablas^{i_2}(\Omega^2\beta)\cdot\ub^{\frac{1}{2}-\gamma} |u|^{-\frac{1}{2}+\gamma+i_3}\nablas^{i_3}(\Omega^2\alpha),\end{equation}
\begin{equation}\label{weylcurrentterm}\underbrace{(\ub|u|^{-1})^{\frac{3}{2}}}_{\le\varepsilon}\cdot\underbrace{\ub^{-1-\gamma}|u|^{2+\gamma+i}\nablas^i(\Omega^3\Theta)}_{\|\cdot\|_{L^2(\mathcal{M})}\lesssim1}\cdot\underbrace{\ub^{\frac{1}{2}-\gamma} |u|^{-\frac{1}{2}+\gamma+i}\nablas^i(\Omega^2\alpha)}_{\|\cdot\|_{L^2(\mathcal{M})}\le\varepsilon^{-\delta}},\end{equation}
$$(\ub|u|^{-1})^{\frac{5}{2}}\cdot\ub^{-1-\gamma}|u|^{2+\gamma+i}\nablas^i(\Omega^3\Xi)\cdot\ub^{-\frac{1}{2}-\gamma} |u|^{\frac{1}{2}+\gamma+i}\nablas^i(\Omega^2\beta),$$
These terms are estimated in the following way. The terms without matter field contain three factors in principle: connection coefficients or $\Ks$, and two curvature components. Unless otherwise specified, we estimate these terms as follows: since $N$ is large enough, we use H\"older inequality \eqref{Holder} to put every factor in $H^{i_j}(S_{\ub,u})$ for suitable $i_j\le N$. We also use H\"older over $(\ub,u)$ plane to put the curvature components factors in spacetime $L^2$. See the under-braces in \eqref{nonlinearterm1} as an example. For most terms, we will have an extra $\ub|u|^{-1}$ to some positive power, which can be made arbitrarily small.  There may be terms having more factors, coming from taking $\nablas$ on $\Omega$. This factor would be bounded by bootstrap assumption \ref{bootstrapA} for $\log\Omega$. We omit these terms in our list.

There could also be terms like \eqref{nonlineartermrhoitself}, in which $\rho$ but not its derivatives appears.  In this case, we put $\rho$ in $L^\infty$ by \eqref{rhodaggeritself}, and put the first factor in spacetime $L^2$. The last type worth mentioning is terms like \eqref{nonlineartermKs}, which come from commutators of $\nablas^i$ with the Hodge operators. In this case, we treat $\Ks$ as a connection coefficient. At last, the terms involving the matter field, like \eqref{weylcurrentterm}, can be simply estimated by putting every factor in spacetime $L^2$, using the bounds \eqref{curvature-weylcurrentbound}.

Combining all estimates above, when $\varepsilon$ is sufficiently small, we will have
$$\mathcal{R}[\alpha]+\mathcal{R}(\mathcal{M})[\alpha]\lesssim 1.$$

For $\beta$-$(\rho,\sigma)$ pair, we use the equations \eqref{Dbbeta}, \eqref{Drho} and \eqref{Dsigma}
\begin{align*}
&\Db\beta+\frac{1}{2}\Omega\tr\chib\beta-\Omega\chibh \cdot \beta+\omegab \beta-\Omega\{\ds \rho+{}^*\ds \sigma+3\eta\rho+3{}^*\eta\sigma+2\chih\cdot\betab\}=-2\Omega I\\
&D\rho+\frac{3}{2}\Omega\tr\chi \rho-\Omega\{\divs \beta+(2\etab+\zeta,\beta)-\frac{1}{2}(\chibh,\alpha)\}=-2\Omega\Lambda,\\
&D\sigma+\frac{3}{2}\Omega\tr\chi\sigma+\Omega\{\curls\beta+(2\etab+\zeta,{}^*\beta)-\frac{1}{2}\chibh\wedge\alpha\}=-2\Omega K,\end{align*}
we compute for $0\le i\le N$,
\begin{align*}
&\Db\left(\ub^{-1-2\gamma}|u|^{2+2\gamma+2i}|\nablas^i(\Omega^2\beta)|^2\D\mu_{\gs}\right)+D\left(\ub^{-1-2\gamma}|u|^{2+2\gamma+2i}\left(|\nablas^i(\Omega^2\widetilde{\rho})|^2+|\nablas^i(\Omega^2\sigma)|^2\right)\D\mu_{\gs}\right)\\
&=\ub^{-1-2\gamma}|u|^{2+2\gamma+2i}\divs(\nablas^i(\Omega^2\beta)\cdot (\nablas^i(\Omega^2\widetilde{\rho}),\nablas^i(\Omega^2\sigma)) )\D\mu_{\gs}\\
&\underline{-(2\gamma|u|^{-1}-\omegab)\left(\ub^{-1-2\gamma}|u|^{2+2\gamma+2i}|\nablas^i(\Omega^2\beta)|^2\D\mu_{\gs}\right)}\\
&+(-1-2\gamma)\ub^{-1}\left(\ub^{-1-2\gamma}|u|^{2+2\gamma+2i}\left(|\nablas^i(\Omega^2\widetilde{\rho})|^2+|\nablas^i(\Omega^2\sigma)|^2\right)\D\mu_{\gs}\right)\\&+\ub^{-1-2\gamma}|u|^{2+2\gamma+2i}\tau_{1}\D\mu_{\gs}.
\end{align*}
The choice of the power $|u|$ is made such that $\tau_1$ does not contain terms like $(\omegab,\Omega\tr\chib)|\nablas^i(\Omega^2\beta)|^2$ and the coefficient $-2\gamma$ of the underlined term has a negative sign. The choice of factor $\Omega^2$ is needed in estimating the nonlinear term $(\chibh,\alpha)$ in the equation for $D\rho$. As a result, the coefficient $\omegab$ is also negative by \eqref{trchibomegab<0}. Also, terms like $\ub^{-1}|\nablas^i(\Omega\widetilde{\rho},\Omega\sigma)|^2$, not included in $\tau_1$, has negative sign.  Integrating over $(\mathcal{M}, \D\ub\D u\D\mu_{\gs})$ and drop the term $\omegab|\nablas^i(\Omega^2\beta)|^2$ , we have
\begin{align*}
&\int_{C_u}\ub^{-1-2\gamma}|u|^{2+2\gamma+2i}|\nablas^i(\Omega^2\beta)|^2+\int_{\Cb_{\ub}}\ub^{-1-2\gamma}|u|^{2+2\gamma+2i}\left(|\nablas^i(\Omega^2\widetilde{\rho})|^2+|\nablas^i(\Omega^2\sigma)|^2\right)\\&+2\gamma\int_{\mathcal{M}}\ub^{-1-2\gamma}|u|^{1+2\gamma+2i}|\nablas^i(\Omega^2\beta)|^2\\
&+(1+2\gamma)\int_{\mathcal{M}}\ub^{-2-2\gamma}|u|^{2+2\gamma+2i}\left(|\nablas^i(\Omega^2\widetilde{\rho})|^2+|\nablas^i(\Omega^2\sigma)|^2\right)\\
\lesssim&1+\int_{\mathcal{M}}\ub^{-1-2\gamma}|u|^{2+2\gamma+2i}|\tau_{1}|.
\end{align*}

The spacetime integrant $\ub^{-1-2\gamma}|u|^{2+2\gamma+2i}\tau_1$ can be decomposed as sum of the following terms:
$$\ub|u|^{-1}\sum_{i_1+i_2=i}\ub^{-1}|u|^{2+i_1}\nablas^{i_1}(\Omega\chibh,\widetilde{\Omega\tr\chib})\cdot\ub^{-\frac{1}{2}-\gamma}|u|^{\frac{1}{2}+\gamma+i_2}\nablas^{i_2}(\Omega^2\beta)\cdot\ub^{-\frac{1}{2}-\gamma}|u|^{\frac{1}{2}+\gamma+i}\nablas^i(\Omega^2\beta),$$
$$\ub|u|^{-1} \sum_{i_1+i_2=i}|u|^{1+i_1}\nablas^{i_1}(\Omega\chih)\cdot\ub^{-\frac{3}{2}-\gamma}|u|^{\frac{3}{2}+\gamma+i_2}\nablas^{i_2}(\Omega^2\betab)\cdot\ub^{-\frac{1}{2}-\gamma}|u|^{\frac{1}{2}+\gamma+i}\nablas^i(\Omega^2\beta),$$
$$\ub|u|^{-1}\sum_{i_1+i_2=i}|u|^{1+i_1}\nablas^{i_1}(\Omega\chih,\Omega\tr\chi,\omega)\cdot\ub^{-1-\gamma}|u|^{1+\gamma+i_2}\nablas^{i_2}(\Omega^2\widetilde{\rho},\Omega^2\sigma)\cdot\ub^{-1-\gamma}|u|^{1+\gamma+i}\nablas^i(\Omega^2\widetilde{\rho},\Omega^2\sigma),$$
\begin{equation*}(\ub|u|^{-1})^{-\gamma}|u|^{-2}\cdot|u|^{1+i}\nablas^{i}(\Omega\chih,\Omega\tr\chi,\omega)\cdot|u|^2\Omega^2\rho\cdot\ub^{-1-\gamma}|u|^{1+\gamma+i}\nablas^i(\Omega^2\widetilde{\rho},\Omega^2\sigma),\end{equation*}

$$(\ub|u|^{-1})^{\frac{1}{2}} \sum_{i_1+i_2=i}\ub^{-1}|u|^{2+i_1}\nablas^{i_1}(\Omega\chibh)\cdot\ub^{\frac{1}{2}-\gamma}|u|^{-\frac{1}{2}+\gamma+i_2}\nablas^{i_2}(\Omega^2\alpha)\cdot\ub^{-1-\gamma}|u|^{1+\gamma+i}\nablas^i(\Omega^2\widetilde{\rho},\Omega^2\sigma),$$
$$(\ub|u|^{-1})^{\frac{3}{2}}\sum_{\substack{i_1+i_2+i_3=2i\\ i_1,i_2,i_3\le i}}\ub^{-1}|u|^{2+i_1}\nablas^{i_1}(\Omega\eta,\Omega\etab)\cdot\ub^{-1-\gamma}|u|^{1+\gamma+i_2}\nablas^{i_2}(\Omega^2\widetilde{\rho},\Omega^2\sigma)\cdot\ub^{-\frac{1}{2}-\gamma}|u|^{\frac{1}{2}+\gamma+i_3}\nablas^{i_3}(\Omega^2\beta),$$
\begin{equation*}(\ub|u|^{-1})^{\frac{1}{2}-\gamma}|u|^{-2}\cdot\ub^{-1}|u|^{2+i}\nablas^{i}(\Omega\eta,\Omega\etab)\cdot |u|^2\Omega^2\rho\cdot\ub^{-\frac{1}{2}-\gamma}|u|^{\frac{1}{2}+\gamma+i}\nablas^{i}(\Omega^2\beta),\end{equation*}
$$(\ub|u|^{-1})^{\frac{1}{2}}\sum_{\substack{i_1+i_2+i_3=2i-1\\ i_1+1,i_2,i_3\le i}}|u|^{2+i_1}\nablas^{i_1}(\Omega \Ks)\cdot\ub^{-1-\gamma}|u|^{1+\gamma+i_2}\nablas^{i_2}(\Omega^2\widetilde{\rho},\Omega^2\sigma)\cdot\ub^{-\frac{1}{2}-\gamma}|u|^{\frac{1}{2}+\gamma+i_3}\nablas^{i_3}(\Omega^2\beta),$$
\begin{equation*}(\ub|u|^{-1})^{\frac{1}{2}}\cdot\ub^{-1-\gamma}|u|^{2+\gamma+i}\nablas^{i}(\Omega^3I)\cdot\ub^{-\frac{1}{2}-\gamma}|u|^{\frac{1}{2}+\gamma+i}\nablas^{i}(\Omega^2\beta),\end{equation*}
\begin{equation*} (\ub|u|^{-1})^{\frac{1}{4}}\cdot\ub^{ -\frac{1}{4}-\gamma}|u|^{\frac{5}{4}+\gamma+i}\nablas^{i}(\Omega^3\Lambda,\Omega^3K)\cdot\ub^{-1-\gamma}|u|^{1+\gamma+i}\nablas^{i}(\Omega^2\widetilde{\rho},\Omega^2\sigma),\end{equation*}
The estimates of $\tau_1$ can then be done similar to $\tau_0$. 
We have
\begin{equation}\label{curvatureestimate1}\mathcal{R}[\beta]+\underline{\mathcal{R}}[\rho,\sigma]+\mathcal{R}(\mathcal{M})[\beta,\rho,\sigma]\lesssim 1.\end{equation}

For $(\rho,\sigma)$-$\betab$ pair, we use the equations \eqref{Dbrho}, \eqref{Dbsigma} and \eqref{Dbetab}
\begin{align*}
&\Db\widetilde{\rho}+\widetilde{\frac{3}{2}\Omega\tr\chib \rho}+\Omega\{\divs \betab+(2\eta-\zeta,\betab)+\frac{1}{2}(\chih,\alphab)\}=-2\widetilde{\Omega\Lambdab},\\
&\Db\sigma+\frac{3}{2}\Omega\tr\chib\sigma+\Omega\{\curls\betab+(2\eta-\zeta,{}^*\betab)+\frac{1}{2}\chih\wedge\alphab\}=2\Omega\Kb,\\
&D\betab+\frac{1}{2}\Omega\tr\chi\betab-\Omega\chih \cdot \betab+\omega \betab+\Omega\{\ds \rho-{}^*\ds \sigma+3\etab\rho-3{}^*\etab\sigma-2\chibh\cdot\beta\}=2\Omega\Ib
\end{align*}
we compute for $0\le i\le N$, 
\begin{align*}
&\Db\left(\ub^{-2-2\gamma}|u|^{3+2\gamma+2i}\left(|\nablas^i(\Omega^2\widetilde{\rho})|^2+|\nablas^i(\Omega^2\sigma)|^2\right)\D\mu_{\gs}\right)+D\left(\ub^{-2-2\gamma}|u|^{3+2\gamma+2i}|\nablas^i(\Omega^2\betab)|^2\D\mu_{\gs}\right)\\
=&\ub^{-2-2\gamma}|u|^{3+2\gamma+2i}\divs((\nablas^i(\Omega^2\widetilde{\rho}),\nablas^i(\Omega^2\sigma))\cdot\nablas^i(\Omega^2\betab))\D\mu_{\gs}\\
&+(-2-2\gamma)\ub^{-1}\left(\ub^{-2-2\gamma}|u|^{3+2\gamma+2i}|\nablas^i(\Omega^2\betab)|^2\D\mu_{\gs}\right)\\
&+\ub^{-2-2\gamma}|u|^{3+2\gamma+2i}\tau_{2}\D\mu_{\gs}.
\end{align*}
Different to the previous steps, $\tau_2$  is allowed to contain terms like $(\omegab,\Omega\tr\chib)|\nablas^i(\Omega^2\widetilde{\rho},\Omega^2\sigma)|^2$, which is controlled by $\mathcal{R}(\mathcal{M})(\rho,\sigma)$, which is already estimated in \eqref{curvatureestimate1}. The choice of the power of $|u|$ also relies on the estimate \eqref{curvatureestimate1}.

The spacetime integrant $\ub^{-2-2\gamma}|u|^{3+2\gamma+2i}\tau_2$ can be decomposed as sum of the following terms:
\begin{equation}\label{curvatureboarderline}\sum_{i_1+i_2=i}|u|^{1+i_1}\nablas^{i_1}(\Omega\chibh,\Omega\tr\chib,\omegab)\cdot\ub^{-1-\gamma}|u|^{1+\gamma+i_2}\nablas^{i_2}(\Omega^2\widetilde{\rho},\Omega^2\sigma)\cdot\ub^{-1-\gamma}|u|^{1+\gamma+i}\nablas^i(\Omega^2\widetilde{\rho},\Omega^2\sigma),\end{equation}
\begin{equation*}(\ub|u|^{-1})^{-\gamma}|u|^{-2}\cdot\ub^{-1}|u|^{2+i}\nablas^{i}(\widetilde{\Omega\tr\chib})\cdot|u|^2\Omega^2\rho\cdot\ub^{-1-\gamma}|u|^{1+\gamma+i}\nablas^i(\Omega^2\widetilde{\rho},\Omega^2\sigma),\end{equation*}
$$(\ub|u|^{-1})^{\frac{1}{2}}\sum_{i_1+i_2=i}|u|^{1+i_1}\nablas^{i_1}(\Omega\chih)\cdot\ub^{-\frac{3}{2}-\gamma}|u|^{\frac{3}{2}+\gamma+i_2}\nablas^{i_2}(\Omega^2\alphab)\cdot\ub^{-1-\gamma}|u|^{1+\gamma+i}\nablas^i(\Omega^2\widetilde{\rho},\Omega^2\sigma),$$
$$\ub|u|^{-1}\sum_{i_1+i_2=i}|u|^{1+i_1}\nablas^{i_1}(\Omega\chih,\Omega\tr\chi,\omega)\cdot\ub^{-\frac{3}{2}-\gamma}|u|^{\frac{3}{2}+\gamma+i_2}\nablas^{i_2}(\Omega^2\betab)\cdot\ub^{-\frac{3}{2}-\gamma}|u|^{\frac{3}{2}+\gamma+i}\nablas^i(\Omega^2\betab),$$
$$\ub|u|^{-1}\sum_{i_1+i_2=i}\ub^{-1}|u|^{2+i_1}\nablas^{i_1}(\Omega\chibh)\cdot\ub^{-\frac{1}{2}-\gamma}|u|^{\frac{1}{2}+\gamma+i_2}\nablas^{i_2}(\Omega^2\beta)\cdot\ub^{-\frac{3}{2}-\gamma}|u|^{\frac{3}{2}+\gamma+i}\nablas^i(\Omega^2\betab),$$
$$(\ub|u|^{-1})^{\frac{3}{2}}\sum_{\substack{i_1+i_2+i_3=2i\\ i_1,i_2,i_3\le i}}\ub^{-1}|u|^{2+i_1}\nablas^{i_1}(\Omega\eta,\Omega\etab)\cdot\ub^{-\frac{3}{2}-\gamma}|u|^{\frac{3}{2}+\gamma+i_2}\nablas^{i_2}(\Omega^2\betab)\cdot\ub^{-1-\gamma}|u|^{1+\gamma+i_3}\nablas^{i_3}(\Omega^2\widetilde{\rho},\Omega^2\sigma),$$
\begin{equation*}(\ub|u|^{-1})^{\frac{1}{2}-\gamma}|u|^{-2}\cdot\ub^{-1}|u|^{2+i}\nablas^{i}(\Omega\eta,\Omega\etab)\cdot\ub^{-\frac{3}{2}-\gamma}|u|^{\frac{3}{2}+\gamma+i}\nablas^{i}(\Omega^2\betab)\cdot|u|^2\Omega^2\rho,\end{equation*}
$$(\ub|u|^{-1})^{\frac{1}{2}}\sum_{\substack{i_1+i_2+i_3=2i-1\\ i_1+1,i_2,i_3\le i}}|u|^{2+i_1}\nablas^{i_1}(\Omega \Ks)\cdot\ub^{-\frac{3}{2}-\gamma}|u|^{\frac{3}{2}+\gamma+i_2}\nablas^{i_2}(\Omega^2\betab)\cdot\ub^{-1-\gamma}|u|^{1+\gamma+i_3}\nablas^{i_3}(\Omega^2\widetilde{\rho},\Omega^2\sigma),$$
\begin{equation}\label{curvatureLambdab}\ub^{-1-\gamma}|u|^{2+\gamma+i}\nablas^{i}(\Omega^3\widetilde{\Lambdab},\Omega^3\Kb)\cdot\ub^{-1-\gamma}|u|^{1+\gamma+i} \nablas^{i}(\Omega^2\widetilde{\rho},\Omega^2\sigma),\end{equation}
\begin{equation*}(\ub|u|^{-1})^{\frac{1}{2}}\cdot\ub^{-1-\gamma}|u|^{2+\gamma+i}\nablas^{i}(\Omega^3\Ib)\cdot\ub^{-\frac{3}{2}-\gamma}|u|^{\frac{3}{2}+\gamma+i} \nablas^{i}(\Omega^2\betab),\end{equation*}
Integrating over $(\mathcal{M}, \D\ub\D u\D\mu_{\gs})$ we then have
$$\mathcal{R}[\rho,\sigma]+\underline{\mathcal{R}}[\betab]+\mathcal{R}(\mathcal{M})[\betab]\lesssim1+\mathcal{R}(\mathcal{M})[\rho,\sigma]+\sqrt{\mathcal{R}(\mathcal{M})[\rho,\sigma]}.$$
The term $\mathcal{R}(\mathcal{M})[\rho,\sigma]$ on the right hand side comes from  the term \eqref{curvatureboarderline} and the term $\sqrt{\mathcal{R}(\mathcal{M})[\rho,\sigma]}$ comes from \eqref{curvatureLambdab}.  By \eqref{curvatureestimate1} and absorbing the square root term,  we have
$$\mathcal{R}[\rho,\sigma]+\underline{\mathcal{R}}[\betab]+\mathcal{R}(\mathcal{M})[\betab]\lesssim1.$$

From $\betab$-$\alphab$ pair, we use equations \eqref{Dbbetab} and \eqref{Dalphab}
\begin{align*}
&\Db\betab+\frac{3}{2}\Omega\tr\chib\betab-\Omega\chibh\cdot\betab-\omegab\betab+\Omega\{\divs\alphab+(\eta-2\zeta)\cdot\alphab\}=-2\Omega\Xib,\\
&\Dh\alphab-\frac{1}{2}\Omega\tr\chi \alphab+2\omega\alphab+\Omega\{\nablas\tensor\betab +(4\etab-\zeta)\tensor \betab+3\chibh \rho-3{}^*\chibh \sigma\}=-2\Omega\Thetab.\end{align*}
we compute for $0\le i\le N$,
\begin{align*}
&\Db\left(\ub^{-2-2\gamma}|u|^{3+2\gamma+2i}\left(|\nablas^i(\Omega^2\betab)|^2\right)\D\mu_{\gs}\right)+D\left(\ub^{-2-2\gamma}|u|^{3+2\gamma+2i}|\nablas^i(\Omega^2\alphab)|^2\D\mu_{\gs}\right)\\
=&\ub^{-2-2\gamma}|u|^{3+2\gamma+2i}\divs(\nablas^i(\Omega^2\betab)\cdot\nablas^i(\Omega^2\alphab))\D\mu_{\gs}\\
&+(-2-2\gamma)\ub^{-1}\left(\ub^{-2-2\gamma}|u|^{3+2\gamma+2i}|\nablas^i(\Omega^2\alphab)|^2\D\mu_{\gs}\right)\\
&+\ub^{-2-2\gamma}|u|^{3+2\gamma+2i}\tau_{3}\D\mu_{\gs},
\end{align*}
where $\ub^{-2-2\gamma}|u|^{3+2\gamma+2i}\tau_{3}$ consists the following terms:
$$\ub|u|^{-1}\sum_{i_1+i_2=i}|u|^{1+i_1}\nablas^{i_1}(\Omega\chibh,\Omega\tr\chib,\omegab)\cdot\ub^{-\frac{3}{2}-\gamma}|u|^{\frac{3}{2}+\gamma+i_2}\nablas^{i_2}(\Omega^2\betab)\cdot\ub^{-\frac{3}{2}-\gamma}|u|^{\frac{3}{2}+\gamma+i}\nablas^i(\Omega^2\betab),$$
$$\ub|u|^{-1} \sum_{i_1+i_2=i}|u|^{1+i_1}\nablas^{i_1}(\Omega\chih,\Omega\tr\chi,\omega)\cdot\ub^{-\frac{3}{2}-\gamma}|u|^{\frac{3}{2}+\gamma+i_2}\nablas^{i_2}(\Omega^2\alphab)\cdot\ub^{-\frac{3}{2}-\gamma}|u|^{\frac{3}{2}+\gamma+i}\nablas^i(\Omega^2\alphab),$$
$$(\ub|u|^{-1})^{\frac{3}{2}}\sum_{i_1+i_2=i}\ub^{-1}|u|^{2+i_1}\nablas^{i_1}(\Omega\chibh)\cdot\ub^{-1-\gamma}|u|^{1+\gamma+i_2}\nablas^{i_2}(\Omega^2\widetilde{\rho},\Omega^2\sigma)\cdot\ub^{-\frac{3}{2}-\gamma}|u|^{\frac{3}{2}+\gamma+i}\nablas^i(\Omega^2\alphab),$$
\begin{equation*}(\ub|u|^{-1})^{\frac{1}{2}-\gamma}|u|^{-2}\cdot\ub^{-1}|u|^{2+i}\nablas^{i}(\Omega\chibh)\cdot|u|^2\Omega^2\rho\cdot\ub^{-\frac{3}{2}-\gamma}|u|^{\frac{3}{2}+\gamma+i}\nablas^i(\Omega^2\alphab),\end{equation*}
$$(\ub|u|^{-1})^2 \sum_{\substack{i_1+i_2+i_3=2i\\ i_1,i_2,i_3\le i}}\ub^{-1}|u|^{2+i_1}\nablas^{i_1}(\Omega\eta,\Omega\etab)\cdot\ub^{-\frac{3}{2}-\gamma}|u|^{\frac{3}{2}+\gamma+i_2}\nablas^{i_2}(\Omega^2\alphab)\cdot\ub^{-\frac{3}{2}-\gamma}|u|^{\frac{3}{2}+\gamma+i_3}\nablas^{i_3}(\Omega^2\betab),$$
$$\ub|u|^{-1} \sum_{\substack{i_1+i_2+i_3=2i-1\\ i_1+1,i_2,i_3\le i}}|u|^{2+i_1}\nablas^{i_1}(\Omega \Ks)\cdot\ub^{-\frac{3}{2}-\gamma}|u|^{\frac{3}{2}+\gamma+i_2}\nablas^{i_2}(\Omega^2\alphab)\cdot\ub^{-\frac{3}{2}-\gamma}|u|^{\frac{3}{2}+\gamma+i_3}\nablas^{i_3}(\Omega^2\betab),$$
\begin{equation*}(\ub|u|^{-1})^{\frac{1}{2}}\cdot\ub^{-1-\gamma}|u|^{2+\gamma+i}\nablas^{i}(\Omega^3\Xib)\cdot\ub^{-\frac{3}{2}-\gamma}|u|^{\frac{3}{2}+\gamma+i} \nablas^{i}(\Omega^2\betab),\end{equation*}
\begin{equation*}(\ub|u|^{-1})^{\frac{1}{2}}\cdot\ub^{-1-\gamma}|u|^{2+\gamma+i}\nablas^{i}(\Omega^3\Thetab)\cdot\ub^{-\frac{3}{2}-\gamma}|u|^{\frac{3}{2}+\gamma+i} \nablas^{i}(\Omega^2\alphab).\end{equation*}
Integrating over $(\mathcal{M}, \D\ub\D u\D\mu_{\gs})$ we then have
$$\underline{\mathcal{R}}[\alphab]+\mathcal{R}(\mathcal{M})[\alphab]\lesssim1.$$
Here  the notation $\underline{\mathcal{R}}[\alphab], \mathcal{R}(\mathcal{M})[\alphab]$ excludes the $\Db$ derivative applying to $\alphab$.

\end{proof}

It remains to obtain estimates for $\alphab$ involving $\Db$ derivatives. We will commute $\nablas^i\Db^j$ to the $\betab$--$\alphab$ group of  the Bianchi equations. In the estimate, we will encounter the terms $\Db^j\rho, \Db^j\sigma$ and $\Db^j\betab$, and their angular derivatives  whose bounds are not included in our theorem explicitly. Nevertheless, these terms can be expressed in terms of angular derivatives of other curvature components and $\Db^j\alphab$, and the components of the Weyl current using the Bianchi equations. Recall that $\rho$ obeys the following equation
$$\Db\rho+ \frac{3}{2}\Omega\tr\chib \rho+\Omega\{\divs \betab+(2\eta-\zeta,\betab)+\frac{1}{2}(\chih,\alphab)\}=-2 \Omega\Lambdab,$$
and $\sigma$ obeys a similar one. Writing in $L^2(S_{\ub,u})$, using bootstrap assumptions \ref{bootstrapA} and \ref{bootstrapA'} for  all connection coefficients terms, we have
$$\|\Db\rho,\Db\sigma\|\lesssim\frac{1}{|u|}\|\rho,\sigma,\betab,\alphab\|+\|\Omega\nablas\betab\|+\|\Omega\Lambdab,\Omega\Kb\|$$
and similarly for $\Db\betab$
$$\|\Db\betab\|\lesssim\frac{1}{|u|}\|\betab,\alphab\|+\|\Omega\nablas\alphab\|+\|\Omega\Xib\|.$$
For $\Db^2\rho,\Db^2\sigma$, we apply one more $\Db$ to the $\Db\rho,\Db\sigma$ equation, commute $\Db$ with $\divs$, and express the one order $\Db$ terms in terms of angular derivatives, we have
$$\|\Db^2\rho,\Db^2\sigma\|\lesssim\frac{1}{|u|^2}\|\rho,\sigma,\betab,\alphab\|+\frac{1}{|u|}(\|\Omega\nablas\betab,\Omega\nablas\alphab, \Db\alphab,\Omega\Lambdab,\Omega\Kb,\Omega\Xib\|)+\|\Omega^2\nablas^2\alphab, \Omega\Db\Lambdab,\Omega\Db\Kb, \Omega^2\nablas\Xib\|,$$
and similarly, 
$$\|\Db^2\betab\|\lesssim\frac{1}{|u|^2}\|\betab,\alphab\|+\frac{1}{|u|}\|\Omega\nablas\alphab,\Db\alphab,\Omega\Xib\|+\|\nablas\Db\alphab, \Omega\Db\Xib\|.$$
Inductively, for $2\le j\le N$, we have
\begin{equation}\label{Dbjrhosigma}\begin{split}\|\Db^j(\rho, \sigma)\|\lesssim&\frac{1}{|u|^j}\| \rho,\sigma,\betab\|+\frac{1}{|u|^{j-1}}\|\Omega\nablas\betab\|+\sum_{0\le k+l\le j, k\le 2, l\le j-1}\frac{1}{|u|^{j-k-l}}\|(\Omega\nablas)^k\Db^l\alphab\|\\
&+\sum_{0\le l\le j-1}\frac{1}{|u|^{j-l-1}}\|\Omega\Db^l\Lambdab,\Omega\Db^l\Kb\|)+\sum_{0\le k+l\le j-1, k\le 1, l\le j-2}\frac{1}{|u|^{j-k-l-1}}\|\Omega(\Omega\nablas)^k\Db^l\Xib\|,\end{split}\end{equation}
and
\begin{equation}\label{Dbjbetab}\|\Db^j\betab\|\lesssim\frac{1}{|u|^j}\| \betab\|+\sum_{0\le k+l\le j, k\le 1, l\le j-1}\frac{1}{|u|^{j-k-l}}\|(\Omega\nablas)^k\Db^l\alphab\|+\sum_{0\le l\le j-1}\frac{1}{|u|^{j-l-1}}\|\Omega\Db^l\Xib\|,\end{equation}
and $\nablas^i\Db^j(\sigma,\rho,\betab)$ is obtained by commuting $\nablas^i$ to  the right hand side.

Now we compute for $0\le i+j\le N, i,j \ge 1$ (and hence $j\le N-1$),
\begin{align*}
&\Db\left(\ub^{-2-2\gamma}|u|^{3+2\gamma+2(i+j)}\left(|\nablas^i\Db^j(\Omega^2\betab)|^2\right)\D\mu_{\gs}\right)+D\left(\ub^{-2-2\gamma}|u|^{3+2\gamma+2(i+j)}|\nablas^i\Db^j(\Omega^2\alphab)|^2\D\mu_{\gs}\right)\\
=&\ub^{-2-2\gamma}|u|^{3+2\gamma+2(i+j)}\divs(\nablas^i\Db^j(\Omega^2\betab)\cdot\nablas^i\Db^j(\Omega^2\alphab))\D\mu_{\gs}\\
&+(-2-2\gamma)\ub^{-1}\left(\ub^{-2-2\gamma}|u|^{3+2\gamma+2(i+j)}|\nablas^i\Db^j(\Omega^2\alphab)|^2\D\mu_{\gs}\right)\\
&+\ub^{-2-2\gamma}|u|^{3+2\gamma+2(i+j)}\tau_{4}\D\mu_{\gs},
\end{align*}
where $\ub^{-2-2\gamma}|u|^{3+2\gamma}\tau_4$ consists of the following terms (we omit terms obtained by $\Db$ applying on the metric components, which will not cause any additional difficulties)
{\small $$\ub|u|^{-1}\sum_{i_1+i_2=i, j_1+j_2=j}|u|\nablas^{i_1}\Db^{j_1}(\Omega\chibh,\Omega\tr\chib,\omegab)\cdot\ub^{-\frac{3}{2}-\gamma}|u|^{\frac{3}{2}+\gamma}\nablas^{i_2}\Db^{j_2}(\Omega^2\betab)\cdot\ub^{-\frac{3}{2}-\gamma}|u|^{\frac{3}{2}+\gamma}\nablas^i\Db^j(\Omega^2\betab),$$
$$\ub|u|^{-1}\sum_{i_1+i_2=i, j_1+j_2=j}|u|\nablas^{i_1}\Db^{j_1}(\Omega\chih,\Omega\tr\chi,\omega,\Omega\eta,\Omega\etab)\cdot\ub^{-\frac{3}{2}-\gamma}|u|^{\frac{3}{2}+\gamma}\nablas^{i_2}\Db^{j_2}(\Omega^2\alphab)\cdot\ub^{-\frac{3}{2}-\gamma}|u|^{\frac{3}{2}+\gamma}\nablas^i\Db^j(\Omega^2\alphab),$$
$$(\ub|u|^{-1})^{\frac{3}{2}}\sum_{i_1+i_2=i, i_2\ge1,  j_1+j_2=j}\ub^{-1}|u|^2\nablas^{i_1}\Db^{j_1}(\Omega\chibh)\cdot \ub^{-1-\gamma}|u|^{1+\gamma}\nablas^{i_2}\Db^{j_2}(\Omega^2\rho,\Omega^2\sigma)\cdot \ub^{-\frac{3}{2}-\gamma}|u|^{\frac{3}{2}+\gamma}\nablas^i\Db^j(\Omega^2\alphab),$$
\begin{equation}\label{curvatureDbrho}(\ub|u|^{-1})^{\frac{1}{2}-\gamma}|u|^{-2}\sum_{ j_1+j_2=j} \ub^{-1}|u|^2\nablas^{i}\Db^{j_1}(\Omega\chibh)\cdot|u|^2\Db^{j_2}(\Omega^2\rho,\Omega^2\sigma)\cdot\ub^{-\frac{3}{2}-\gamma}|u|^{\frac{3}{2}+\gamma}\nablas^i\Db^j(\Omega^2\alphab),\end{equation}
$$ \ub|u|^{-1}\sum_{\substack{i_1+i_2+i_3=2i, i_1,i_2,i_3\le i\\j_1+j_2+j_3=2j,j_1,j_2,j_3\le j}}|u|\nablas^{i_1}\Db^{j_1}(\Omega\eta,\Omega\etab)\cdot\ub^{-\frac{3}{2}-\gamma}|u|^{\frac{3}{2}+\gamma}\nablas^{i_2}\Db^{j_2}(\Omega^2\alphab)\cdot \ub^{-\frac{3}{2}-\gamma}|u|^{\frac{3}{2}+\gamma}\nablas^{i_3}\Db^{j_3}(\Omega^2\betab),$$
$$\ub|u|^{-1} \sum_{\substack{i_1+i_2+i_3=2i, i_1 ,i_2,i_3\le i\\j_1+j_2+j_3=2j, j_1, j_2, j_3\le j, j_1\ge1}} |u|\nablas^{i_1} \Db^{j_1-1}\nablas(\Omega\chib) \cdot \ub^{-\frac{3}{2}-\gamma}|u|^{\frac{3}{2}+\gamma}\nablas^{i_2}\Db^{j_2}(\Omega^2\alphab)\cdot \ub^{-\frac{3}{2}-\gamma}|u|^{\frac{3}{2}+\gamma}\nablas^{i_3}\Db^{j_3}(\Omega^2\betab),$$
\begin{equation*}(\ub|u|^{-1})^{\frac{1}{2}}\cdot\ub^{-1-\gamma}|u|^{2+\gamma}\nablas^i\Db^j(\Omega^3\Xib)\cdot\ub^{-\frac{3}{2}-\gamma}|u|^{\frac{3}{2}+\gamma}\nablas^i\Db^j(\Omega^2\betab),\end{equation*}
\begin{equation*}(\ub|u|^{-1})^{\frac{1}{2}}\cdot\ub^{-1-\gamma}|u|^{2+\gamma}\nablas^i\Db^j(\Omega^3\Thetab)\cdot \ub^{-\frac{3}{2}-\gamma}|u|^{\frac{3}{2}+\gamma} \nablas^{i}\Db^j(\Omega^2\alphab).\end{equation*}}
Compared with $\tau_3$, we have additional terms coming from $\Db$ commuting with $D$ and the Hodge operators, but no terms coming from $\nablas$ commuting with Hodge operators. Similar to $\tau_3$, we would like to argue that the above weighted factors are suitably bounded. The only difference is that we should express  curvature terms involving $\Db(\rho,\sigma)$ using \eqref{Dbjrhosigma}. The $\rho,\sigma,\betab$ on the right hand side of \eqref{Dbjrhosigma} (with at least one additional $\nablas$) are bounded by $\mathcal{R}(\mathcal{M})[\rho,\sigma,\betab]$ which has been already estimated before, and the $\alphab$ terms (together with its $\nablas, \Db$ derivatives) are simply what we are going to estimate, which can be absorbed for $\varepsilon$ small enough. For fluid terms, we still have
$$\sup_{i+j\le N, i\ge1}\int_{\mathcal{M}}\ub^{-2-2\gamma}|u|^{3+2\gamma+2(i+j)}|\nablas^i\Db^j(\Omega^3(\Xib,\Thetab, \Lambdab , \Kb))|^2\lesssim 1$$
from Proposition \ref{fluidestimates}. We still have the terms involving  $\Db^j\rho$ (without $\nablas$) to estimate, that is, \eqref{curvatureDbrho} in the above list.  Note that the estimate for $\Db^j\sigma$ is in fact the same to $\nablas^i\Db^j\sigma$ since the matter field term in $\Db\sigma$ is $\Kb$ which is the angular derivatives of fluid variables and hence behaves better than $\Lambdab$. We only need to concern about the terms $\Db\Xib$ and $\Db\Lambdab$ on the right hand side of \eqref{Dbjrhosigma}, or more precisely, module the better behaved terms, we have, in $L^2(S_{\ub,u})$,  for $j\le N-1$, 
\begin{equation}\label{Dbjrhomodule}\|\Db^j(\Omega^2\rho)\|\lesssim\frac{1}{|u|^j}\|\Omega^2\rho\|+\sum_{0\le l\le j-1}\frac{1}{|u|^{j-l-1}}\|\Db^l(\Omega^3\Lambdab)\|+\sum_{0\le l\le j-2}\frac{1}{|u|^{j-l-1}}\|\Db^l(\Omega^3\Xi)\|+\cdots.\end{equation}
 Using  \eqref{fluidL2Sbootstrap1} (without $\varepsilon^{-\delta}$)  we have for $j\le N-1$,  module the better behaved terms. 
$$|u|^{-1}\|\Db^j(\Omega^2\rho)\|_{L^2(S_{\ub,u})}\lesssim|u|^{-2-j}+\cdots,$$
and we have the desired estimate
\begin{equation}\label{curvaturenablasDbalphab}
\begin{split}
\sup_{ i+j\le N, i\ge1}&\int_{C_u} \ub^{-2-2\gamma}|u|^{3+2\gamma+2(i+j)}|\nablas^i\Db^j(\Omega^2 \betab)|^2\\
+&\ub^{-2-2\gamma}\int_{\Cb_{\ub}} |u|^{3+2\gamma+2(i+j)}|\nablas^i\Db^j(\Omega^2 \alphab)|^2\\
+&\int_{\mathcal{M}} \ub^{-3-2\gamma}|u|^{3+2\gamma+2(i+j)}|\nablas^i\Db^j(\Omega^2 \alphab)|^2\lesssim1,
\end{split}
\end{equation}
which can be done as in estimating $\underline{\mathcal{R}}[\alphab],\mathcal{R}(\mathcal{M})[\alphab]$.

\section{The estimates for connection coefficients: lower orders}

In this section, we are going to recover all the bootstrap assumptions \ref{bootstrapA}, \ref{bootstrapB} and \ref{bootstrapA'}. But let us introduce another bootstrap assumption concerning top order angular  derivative of $\etab$:
 \renewcommand{\theboot}{(B')} 
\begin{boot}\label{bootstrapB'} 
$$\int_{C_u}\ub^{-1-2\gamma}|u|^{2\gamma+2(N+1)} |\nablas^{N+1}(\Omega^2\eta, \Omega^2\etab)|^2\le\varepsilon^{-2\delta}.$$
$$\ub^{-1-2\gamma}\int_{\Cb_{\ub}}|u|^{2\gamma+2(N+1)} |\nablas^{N+1}(\Omega^2\eta)|^2\le\varepsilon^{-2\delta}.$$
\end{boot}
We have
\begin{proposition}\label{improvebootA}
 Under the bootstrap assumptions \ref{bootstrapA}, \ref{bootstrapB} and \ref{bootstrapA'}, \ref{bootstrapB'}, if $\varepsilon$ is small enough, for $\ub|u|^{-1}\le\varepsilon$, we have the following estimates which improve bootstrap assumptions \ref{bootstrapA}  and \ref{bootstrapA'}:
 
 For $\Omega$:
\begin{align*}
|u|^{-1}\| \Db^j D^k\log\Omega\|_{L^2(S_{\ub,u})}\lesssim |u|^{-(j+k)},&\  j\le N, k\le1,\\
|u|^{-1}\|\nablas^i \Db^j \log\Omega\|_{L^2(S_{\ub,u})}\lesssim \ub^{1+\gamma}|u|^{-1-\gamma-(i+j)}, &\   i\ge1, i+j\le N+1,\ \\
|u|^{-1}\|\nablas^i\Db^j D \log\Omega\|_{L^2(S_{\ub,u})}\lesssim \ub^\gamma|u|^{-\gamma-1-(i+j)}, &\ i\ge1, i+j\le N,\\
|u|^{-1}\|\Omega\nablas^{N+1}\log\Omega\|_{L^2(S_{\ub,u})}\lesssim \ub^{1+\gamma}|u|^{-1-\gamma-(N+1)}.
\end{align*}
Lower order derivatives: ($i+j+k\le N$):

For $\chih,\tr\chi,\omega$:
\begin{align*}
|u|^{-1}\|  D  (\Omega\chih)\|_{L^2(S_{\ub,u})}\lesssim  \ub^{-1}|u|^{-1} ,&\\
|u|^{-1}\|  \Db^j D^k  (\Omega\chih,\Omega\tr\chi,\omega)\|_{L^2(S_{\ub,u})}\lesssim  |u|^{-1-(j+k)} ,&\  k\le1,\ \text{except}\ D(\Omega\chih),\\
|u|^{-1}\|  \nablas^i\Db^j (\Omega\chih,\Omega\tr\chi,\omega)\|_{L^2(S_{\ub,u})}\lesssim  \ub^{\gamma}|u|^{-1-\gamma-(j+k)},&\ i\ge1.
\end{align*}

For $\eta,\etab,\chibh,\tr\chib$:
\begin{align*}
|u|^{-1}\|\Db^j  D^k (\Omega\eta,\Omega\etab,\Omega\chibh, \Omega\tr\chib)\|_{L^2(S_{\ub,u})}\lesssim  |u|^{-1-(j+k)},&\ k\le 1,\\
|u|^{-1}\|\nablas \Db^j D  (\Omega\eta,\Omega\etab,\Omega\chibh,  \Omega\tr\chib)\|_{L^2(S_{\ub,u})}\lesssim  \ub^{\gamma}|u|^{-1-\gamma-(j+2)},\\
|u|^{-1}\|\nablas^i \Db^j (\Omega\eta,\Omega\etab,\Omega\chibh,  \widetilde{\Omega\tr\chib})\|_{L^2(S_{\ub,u})}\lesssim  \ub|u|^{-2-(i+j)}.
\end{align*}
For $\omegab$:
\begin{align*}
|u|^{-1}\|\nablas^i \widetilde{\omegab}\|_{L^2(S_{\ub,u})}\lesssim  \ub|u|^{-2-i}.
\end{align*}
For $\Ks$:
\begin{align*}
|u|^{-1}\| \Db^j \Ks\|_{L^2(S_{\ub,u})}\le |u|^{-2-j},&\ j\le N-1, \\
|u|^{-1}\|\nablas^i \Db^j  (\Omega\Ks)\|_{L^2(S_{\ub,u})}\le \ub |u|^{-3-(i+j)}, &\  i\ge1, i+j\le N-1.
\end{align*}

\end{proposition}
 
 \begin{proof}
We first consider the estimates only for  angular derivatives. The estimates can be done by applying \eqref{Gronwallub} and \eqref{Gronwallu} to suitable null structure equations. We will need the following bounds of the Ricci tensor, by Proposition \ref{fluidestimates} and consequence estimates \eqref{fluidL2Sbootstrap1}, \eqref{fluidL2Sbootstrap2} (without $\varepsilon^{-\delta}$) and \eqref{fluidLinfty} (note that we have improved bound $|\Us|\lesssim\ub|u|^{-1}$ by running the argument again without $\varepsilon^{-\delta}$) :

  \begin{equation}\label{connection-low-Riccibound}
\begin{split}\sup_{0\le i\le N}&\ub^{-2}\int_{S_{\ub,u}}|u|^{4+2i}|\nablas^i(\Omega^2(\Ricsef,\Ricset)|^2\\&+\int_{S_{\ub,u}}|u|^{2+2i}|\nablas^i(\Omega^2\Rics, \Omega^2\mathbf{Ric}_{34},\Omega^2\mathbf{Ric}_{44})|^2\lesssim 1.\end{split}\end{equation}

 For $\chih$, we write the equations for \eqref{Dchih} as
$$\Dh(\Omega\chih)=2\omega\cdot\Omega\chih-\Omega^2\alpha.$$
By applying  H\"older inequality \eqref{Holder} and Gronwall \eqref{Gronwallub}, using bootstrap assumptions \ref{bootstrapA} and \eqref{curvaturebound} from Proposition \ref{curvature}, we have
\begin{align}\nonumber|u|^{-1}\|\Omega\chih\|_{H^N(S_{\ub,u})}\lesssim&\int_0^{\ub} |u|^{-2}\|\omega\|_{H^N(S_{\ub',u})}\|\Omega\chih\|_{H^N(S_{\ub',u})}\D\ub'+\int_0^{\ub}|u|^{-1}\|\Omega^2\alpha\|_{H^N(S_{\ub',u})}\D\ub'
\\\nonumber\lesssim&\int_0^{\ub}\varepsilon^{-\delta}|u|^{-1}\cdot\varepsilon^{-\delta}|u|^{-1}\\\nonumber
&+\left(\int_0^{\ub}\ub'^{-1+2\gamma}|u|^{-2-2\gamma}\D\ub'\right)^{\frac{1}{2}}\left(\int_0^{\ub}\ub'^{1-2\gamma}|u|^{2\gamma}\|\Omega^2\alpha\|^2_{H^N(S_{\ub',u})}\D\ub'\right)^{\frac{1}{2}}\\\label{chihbound}
\lesssim&\varepsilon^{-2\delta}\ub|u|^{-2}+\ub^{\gamma}|u|^{-1-\gamma}\lesssim \ub^\gamma |u|^{-1-\gamma}\end{align}
 when $\ub|u|^{-1}\le\varepsilon$ is sufficiently small. This gives the desired estimate for $\chih$. 
 
 For $\tr\chi,\eta$ and $\omegab$, we rewrite the equation for \eqref{Dtrchi}, \eqref{Deta} and \eqref{Domegab} as
$$D(\widetilde{\Omega\tr\chi})=-\frac{1}{2}(\Omega\tr\chi)^2+2\omega\Omega\tr\chi-|\Omega\chih|^2-\Omega^2\mathbf{Ric}_{44},$$
$$D(\Omega\eta) = \omega\cdot\Omega\eta+(\Omega\chi)\cdot(\Omega\etab)-\Omega^2(\beta+\frac{1}{2}\Ricsef).$$
$$D  \widetilde{\omegab} =\Omega^2(2(\eta,\etab)-|\eta|^2-(\rho+\frac{1}{6}\mathbf{R}+\frac{1}{2}\mathbf{Ric}_{34})).$$
Similar arguments give the desired bounds. 

For $\chibh$ and $\tr\chib$, we use the following equations \eqref{Dchibh}, \eqref{Dtrchib} of $D$ direction, in which the right hand sides contain top order derivatives and we should use bootstrap assumption \ref{bootstrapB'}:
$$\Dh(\Omega\chibh)=\Omega^2(\nablas \tensor \etab + \etab \tensor \etab +\frac{1}{2}\tr\chi\chibh-\frac{1}{2}\tr\chib \chih+\frac{1}{2}\widehat{\Rics}),$$
$$D(\widetilde{\Omega\tr\chib})=\Omega^2(2\divs\etab+2|\etab|^2-\frac{1}{2}\tr\chi\tr\chib-(\chih,\chibh)+2(\rho+\frac{1}{6}\mathbf{R})).$$
The estimates are also in the same way.

 For $\etab$, we rewrite the equation  \eqref{Db-etab} in the form
$$\Db(\Omega\etab) = \omegab\cdot\Omega\etab+ (\Omega\chib) \cdot\Omega\eta+\Omega^2(\betab-\frac{1}{2}\Ricset).$$
Since $\omegab<0$, the first term on the right hand side can be dropped, and by H\"older inequality  \eqref{Holder}, by the bounds derived above (especially the bound for $\eta$, since the borderline term is $\Omega\tr\chib\eta$), for $\varepsilon$ small enough, we have
$$|u|^{-1}\|(\Omega\chibh,\Omega\tr\chib)\cdot(\Omega\eta)-\frac{1}{2}\Omega^2\Ricset\|_{H^{N}(S_{\ub,u})}\lesssim \ub|u|^{-3}.$$
By Gronwall \eqref{Gronwallu}, for $\etab$, $s=1,\nu=0$, so
\begin{align*}&\|\Omega\etab\|_{H^N(S_{\ub,u})}\\
\lesssim& \|\Omega\etab\|_{H^N(S_{\ub,u_0})}+\int_{u_0}^u \ub|u'|^{-2}\D u'+\int_{u_0}^u\|\Omega^2\betab\|_{H^N(S_{\ub,u'})}\D u'\\
\lesssim&\ub|u_0|^{-1}+\ub|u|^{-1}+\left(\int_{u_0}^u \ub^{2+2\gamma}|u'|^{-3-2\gamma}\D u'\right)^{\frac{1}{2}}\left(\int_{u_0}^u \ub^{-2-2\gamma}|u'|^{3+2\gamma}\|\Omega^2\betab\|^2_{H^N(S_{\ub,u'})}\D u'^2\right)^{\frac{1}{2}}\\
\lesssim&\ub|u|^{-1}
\end{align*}
when $\varepsilon$ is small enough. This gives the desired bound for $\etab$. The bounds for $\omegab$ gives the desired bounds for $\log\Omega$, except that the bound for $\Omega\nablas^{N+1}\log\Omega=\frac{1}{2}\Omega\nablas^N(\eta+\etab)$ is given by $\eta$ and $\etab$ (which will lose an $\Omega$). For $\omega$, we use the equation \eqref{Dbomega}
$$\Db  \omega =\Omega^2(2(\eta,\etab)-|\etab|^2-(\rho+\frac{1}{6}\mathbf{R}+\frac{1}{2}\mathbf{Ric}_{34})),$$
and apply \eqref{Gronwallu} to $\omega$ with $s=\nu=0$.

The estimates for $\Ks$  can be done by the following equation
$$D\widetilde{\Ks}+\Omega\tr\chi\Ks=\divs\divs(\Omega\chih)-\frac{1}{2}\Deltas(\Omega\tr\chi).$$
Rewrite this equation as $D(\Omega \widetilde{\Ks})=\omega\Omega\widetilde{\Ks}+\Omega D\widetilde{\Ks}$ and apply \eqref{Gronwallub}, using the bootstrap assumption \ref{bootstrapA} and \ref{bootstrapB} for $\chih$ and $\tr\chi$, we have
\begin{equation}\label{Gausscurvaturebound}|u|^{-1}\|\Omega\widetilde{\Ks}\|_{H^{N-1}(S_{\ub,u})}\lesssim\varepsilon^{-\delta}\ub^{1+\gamma}|u|^{-3-\gamma}\le\ub|u|^{-3},\end{equation}
for $\varepsilon$ small enough (recalling that $\gamma>\delta$), giving the desired bound for $\Ks$.

Now we turn to mixed derivatives involving $\Db$ and $\nablas$. We first deal with $\nablas^i\Db^j\log\Omega=\nablas^i\Db^{j-1}\omegab$ for $i+j\le N+1$ and $j\ge 2$. This can be estimated by commuting $\nablas^i\Db^{j-1}$ with the $D\omegab$ equation:
$$D( \nablas^i\Db^{j-1}\omegab)=\nablas^i\Db^{j-1}\left(\Omega^2(2(\eta,\etab)-|\eta|^2-(\rho+\frac{1}{6}\mathbf{R}+\frac{1}{2}\mathbf{Ric}_{34}))\right)+[D,\nablas^i\Db^{j-1}]\omegab.$$
Applying \eqref{Gronwallub} leads to the desired bounds for $\log\Omega$.

The bounds for $\nablas^i\Db^j$ derivative of the other connection coefficients can be directly seen from the corresponding $\Db$ null structure equations (after taking several $\nablas^i\Db^{j-1}$ derivatives). The bounds for $\nablas^i\Db^j\etab$ and $\nablas^i\Db^j\omega$ can be seen from the above equations for $\Db\etab$ and $\Db\omega$.  The bound for $\nablas^i\Db^j\eta$ is obtained by the relation
$$\Db\eta=2\nablas\omegab-\Db\etab.$$
For $\chih,\tr\chi,\chibh,\tr\chib$ and $\Ks$ we use \eqref{Dbchibh}, \eqref{Dbtrchib}, \eqref{Dbchih}, \eqref{Dbtrchi} and the equation for $\Db\Ks$:
$$\Dbh(\Omega\chibh)=2\omegab\cdot\Omega\chibh-\Omega^2\alphab.$$
$$\Db(\Omega\tr\chib)=-\frac{1}{2}(\Omega\tr\chib)^2+2\omegab\Omega\tr\chib-|\Omega\chibh|^2-\Omega^2\mathbf{Ric}_{33},$$
$$\Dbh(\Omega\chih)=\Omega^2(\nablas \tensor \eta + \eta \tensor \eta +\frac{1}{2}\tr\chib\chih-\frac{1}{2}\tr\chi \chibh+\frac{1}{2}\widehat{\mathbf{Ric}_{AB}}),$$
$$\Db(\Omega\tr\chi)=\Omega^2(2\divs\eta+2|\eta|^2-\tr\chi\tr\chib-2\Ks+\gs^{AB}\mathbf{Ric}_{AB}),$$
$$\Db\Ks+\Omega\tr\chib\Ks=\divs\divs(\Omega\chibh)-\frac{1}{2}\Deltas(\Omega\tr\chib).$$
One should note that in estimating $\nablas^i\Db^j(\etab, \omega, \chibh), i+j\le N, j\ge1$, we should use $L^2(S_{\ub,u})$ norm of $\nablas^i\Db^{j-1}(\rho,\betab,\alphab)$. These are not top order terms so that we can integrate using the corresponding null Bianchi equations of $D$ direction to get the $L^2(S_{\ub,u})$ norm (in other words, $L^\infty$ in $(\ub,u)$). We should also utilize the estimates  involving $\Db$ derivatives of the fluid variables from Proposition \ref{fluidestimates} and consequence lower order estimates. We omit the details.

Finally we turn to totally mixed derivatives involving $D, \Db$ and $\nablas$. Note first that $D\log\Omega=\omega$ so the estimates for $\Omega$ are exactly those for $\omega$.  For $D\omega$, we commute $D$ with the $\Db\omega$ equation and integrate along $\Db$ direction to get estimate for $D\omega$. Note that the $D$ derivative of the fluid variable can be expressed in terms of $\Db$ and $\nablas$.  Nevertheless, the estimate for $\Db^jD\omega$ can be directly seen of the commuted equation. For $D(\Omega\chih)$ which is essentially $\alpha$, we only need to integrate the equation for $\Dbh\alpha$. At last, the estimates for the mixed derivatives involving at least one $D$ derivative of the remaining connection coefficients can be directly seen from the $D$ equations which have been use above.  
\end{proof}

\section{The estimates for connection coefficients: top orders}

 Finally we are going to improve bootstrap assumptions \ref{bootstrapB} and \ref{bootstrapB'}.

\begin{proposition}\label{improvebootB} Under the bootstrap assumptions \ref{bootstrapA}, \ref{bootstrapB} and \ref{bootstrapA'}, \ref{bootstrapB'}, if $\varepsilon$ is small enough, for $\ub|u|^{-1}\le\varepsilon$, we have the following estimates which improve bootstrap assumptions \ref{bootstrapB}  and \ref{bootstrapB'}:

Top order derivatives: ($i+j=N+1$) 
   
  Mixed derivatives:
\begin{align*}
|u|^{2\gammat}\int_{\Cb_{\ub}}|u'|^{-1-2\gammat+2(N+1)}|  \Db^{N}D(\Omega\chih, \Omega\tr\chi,\Omega\chibh,\Omega\tr\chib, \Omega\eta,\Omega\etab, \omega)|^2\lesssim1  ,&\\
\int_{\Cb_{\ub}}\ub^{-2\gamma}|u|^{-1+2\gamma+2(N+1)}|\nablas^i \Db^j  (\Omega\chih, \Omega\tr\chi,\Omega\chibh,\Omega\tr\chib, \Omega\eta,\Omega\etab, \omega)|^2\lesssim1 ,&\  1\le i\le N,\\
\int_{\Cb_{\ub}}\ub^{-2\gamma}|u|^{-1+2\gamma+2(N+1)}|\nablas \Db^{N-1}D  ( \Omega\chibh,\Omega\tr\chib, \Omega\eta,\Omega\etab)|^2\lesssim1,&\ 
\end{align*}
  
 Top angular derivatives: 
\begin{align*}
\int_{C_u}\ub^{-1-2\gamma}|u|^{2\gamma+2(N+1)}|\nablas^{N+1}(\Omega^2\chih,\Omega\omega)|^2\lesssim1 ,\\
\int_{\Cb_{\ub}}\ub^{-2\gamma}|u|^{-1+2\gamma+2(N+1)}|\nablas^{N+1}(\Omega^2\etab,\Omega^2\chibh)|^2\lesssim 1,\\
\int_{\Cb_{\ub}}\ub^{-1-2\gamma}|u|^{2\gamma+2(N+1)}|\nablas^{N+1}(\Omega^2\eta)|^2\lesssim 1,\\
|u|^{-1}\|\nablas^{N+1}(\Omega^2\tr\chi,\Omega^2\tr\chib)\|_{L^2(S_{\ub,u})}\lesssim \ub^{\gamma}|u|^{-1-\gamma-(N+1)}, \\
\int_{C_u}\ub^{-1-2\gamma}|u|^{2\gamma+2(N+1)} |\nablas^{N+1}(\Omega^2\eta,\Omega^2\etab)|^2\lesssim1.
\end{align*}

For $\Ks$: ($i+j=N$)
\begin{align*}
\int_{\Cb_{\ub}}\ub^{-2\gamma}|u|^{-1+2\gamma+2N}|\nablas^i\Db^j(\Omega\Ks)|^2\lesssim1, &\ 1\le i,j\le N-1,\\
\int_{\Cb_{\ub}}\ub^{-2\gamma}|u|^{-1+2\gamma+2N}|\nablas^N(\Omega^2\Ks)|\lesssim1.
\end{align*}
\end{proposition}
\begin{proof}
Note first that by the Gauss equation, the estimates for $\Ks$ can be directly written down using estimates for $\rho$ and other lower order estimates. We then turn to other connection coefficients.  We still consider the top order angular derivatives first. The proof is by using the elliptic-transport systems. 
A general form of the elliptic estimate is the following
\begin{align*}
 \|\theta\|_{H^{N+1}(S_{\ub,u})}\lesssim \||u|(\divs\theta,\curls\theta)\|_{H^N(S_{\ub,u})}+(1+ |u|\|\Ks\|_{H^{N-1}(S_{\ub,u})}) \|\theta\|_{H^N(S_{\ub,u})}.
\end{align*}
provided that $|u|^{-1}\|\Ks\|_{L^2(S_{\ub,u})}\lesssim|u|^{-2}$, where $\theta$ is a tangential one-form or trace-free $(0,2)$ tensor. Its proof can also be found in \cite{Chr08}. Since we have \eqref{Gausscurvaturebound}, the above elliptic estimates hold in the following form
\begin{align}\label{elliptic1}
 \|\Omega\theta\|_{H^{N+1}(S_{\ub,u})}\lesssim \||u|(\Omega \divs\theta, \Omega \curls\theta)\|_{H^N(S_{\ub,u})}+\|\theta\|_{H^N(S_{\ub,u})},
\end{align}
where we have  used \eqref{derivativelapsetop}. 

For the spacetime Ricci tensor, from Proposition \ref{fluidestimates}, we have the following bounds:
\begin{equation}\label{connection-top-Riccibound}
\begin{split}\sup_{1\le i\le N+1}\int_{C_u}\ub^{-2-2\gammat}|u|^{3+2\gammat+2i}|\nablas^i(\Omega^2\mathbf{Ric}_{44})|^2\lesssim1,\\
\sup_{1\le i\le N+1}\ub^{-2-2\gammat}\int_{\Cb_{\ub}}|u|^{3+2\gammat+2i}|\nablas^i(\Omega^3\mathbf{Ric}_{33})|^2\lesssim1,\\
\sup_{1\le i\le N+1}\int_{C_u}\ub^{-2-2\gammat}|u|^{3+2\gammat+2i}|\nablas^i(\Omega^2\Ricsef,\Omega^2\mathbf{Ric}_{34})|^2\lesssim1,\\
\sup_{1\le i\le N+1}\ub^{-2-2\gammat} \int_{\Cb_{\ub}}|u'|^{3+2\gammat+2i}|\nablas^i(\Omega^2\Ricset,\Omega^2\mathbf{Ric}_{34})|^2\lesssim1.\end{split}
\end{equation}
\begin{remark}Because we focus on top order angular derivatives, we lose an $\Omega$ as compared to lower order terms. For simplicity we don't separate  top and lower order cases. Therefore in the second line we have $\Omega^3$ instead of $\Omega^2$ in the integral. Moreover, we will lose one more $\Omega$ for a top order bound written on $C_u$.
\end{remark}

For $\chih, \tr\chi$, we use \eqref{divchih} and \eqref{Dtrchi}
\begin{align*}
&\begin{dcases}\divs(\Omega\chih)&=\frac{1}{2}\Omega^2\ds\tr \chi'+\Omega\chih\cdot\etab+\frac{1}{2}\Omega\tr \chi\eta-\Omega(\beta-\frac{1}{2}\Ricsef),\\
D\tr\chi'&=-\frac{1}{2}(\Omega\tr\chi')^2-|\chih|^2-\mathbf{Ric}_{44}.
\end{dcases}
\end{align*}
Using the estimates in Proposition \ref{curvature} and \ref{improvebootA}, and \eqref{connection-low-Riccibound}
$$|u|^{-1}\|\Omega\chih\cdot\Omega\etab+\frac{1}{2}\Omega\tr \chi\cdot\Omega\eta+\frac{1}{2}\Omega^2\Ricsef\|_{H^N(S_{\ub,u})}\lesssim\ub|u|^{-3}.$$
The bootstrap assumptions \ref{bootstrapB} for $\tr\chi$ (and \eqref{derivativelapsetop}) gives
$$|u|^{-1}\|\Omega\cdot\Omega^2\ds\tr \chi'\|_{H^N(S_{\ub,u})}\lesssim|u|^{-1}(\|\ds(\Omega^2\tr\chi)\|_{H^N(S_{\ub,u})}+\|\ds(\Omega\tr\chi)\|_{H^{N-1}(S_{\ub,u})})\lesssim\varepsilon^{-\delta}\ub|u|^{-3}$$
together with the bound for $\beta$, when $\varepsilon$ is small enough,  we have 
$$\int_0^{\ub}\ub'^{-1-2\gamma}|u|^{2+2\gamma}\|\frac{1}{2}\Omega^3\ds\tr \chi'+\Omega\chih\cdot\Omega\etab+\frac{1}{2}\Omega\tr \chi\cdot\Omega\eta-\Omega^2(\beta-\frac{1}{2}\Ricsef)\|^2_{H^N(S_{\ub',u})}\D\ub'\lesssim1.$$
We will apply \eqref{elliptic1} for $\chih$, and then integrate over $\ub$ with suitable weight. But the lower order derivatives of $\chih$ in \eqref{chihbound} is not sufficient. For this, we appeal to the equation
$$\Dbh(\Omega\chih)-\frac{1}{2}\Omega\tr\chib\Omega\chih=\Omega^2(\nablas \tensor \eta + \eta \tensor \eta -\frac{1}{2}\tr\chi \chibh+\frac{1}{2}\widehat{\Rics}),$$
and apply \eqref{Gronwallu} to $\Omega\chih$ for $s=2, \nu=-1$. Using the second line of the bootstrap assumption \ref{bootstrapB'} for $\eta$, together with the lower order estimates, we have
$$\|\Omega\chih\|_{H^N(S_{\ub,u})}\lesssim \|\Omega\chih\|_{H^N(S_{\ub,u_0})}+\varepsilon^{-\delta}\ub^{\frac{1}{2}+\gamma}|u|^{-\frac{1}{2}-\gamma},$$
if $\varepsilon$ is small enough, we will have
\begin{equation*}\int_0^{\ub}\ub'^{-1-2\gamma}|u|^{2\gamma}\|\Omega\chih\|^2_{H^{N}(S_{\ub',u})}\D\ub'\lesssim1.\end{equation*}
By applying \eqref{elliptic1}, we have the desired bound
\begin{equation}\label{chihtopbound}\int_0^{\ub}\ub'^{-1-2\gamma}|u|^{2\gamma}\|\Omega^2\chih\|^2_{H^{N+1}(S_{\ub',u})}\D\ub'\lesssim1.\end{equation}
Commute $\nablas$ with the equation for $\tr\chi'$, we have
$$D\nablas(\tr\chi')=-\Omega\tr\chi'\nablas(\Omega\tr\chi')-2\chih\cdot\nablas\chih-\nablas\mathbf{Ric}_{44}.$$
The bootstrap assumption \ref{bootstrapB} for $\tr\chi$ (together with its lower order estimates and \eqref{derivativelapsetop}) gives
$$|u|^{-1}\|\Omega\tr\chi'\nablas(\Omega\tr\chi')\|_{H^N(S_{\ub,u})}\lesssim \Omega^{-3}\varepsilon^{-\delta}\ub^{\gamma}|u|^{-3-\gamma},$$
By \eqref{chihtopbound} we just derived, we have
\begin{align*}&\int_0^{\ub}|u|^{-1}\|\chih\cdot\nablas\chih\|_{H^N(S_{\ub,u})}\\\lesssim& \Omega^{-3}\int_0^{\ub}|u|^{-3}\|(|u|\nablas)(\Omega^2\chih)\|_{H^N(S_{\ub',u})}\D\ub'\\
\lesssim&\Omega^{-3}\left(\int_0^{\ub}\ub'^{1+2\gamma}|u|^{-6-2\gamma}\D\ub'\right)^{\frac{1}{2}}\left(\int_0^{\ub}\ub'^{-1-2\gamma}|u|^{2\gamma}\|\Omega^2\chih\|^2_{H^{N+1}(S_{\ub',u})}\D\ub'\right)^{\frac{1}{2}}\\
\lesssim&\Omega^{-3}\ub^{1+\gamma}|u|^{-3-\gamma}.\end{align*}
By \eqref{connection-top-Riccibound}, we have, in a similar way,
\begin{align*}&\int_0^{\ub}|u|^{-1}\|\nablas\mathbf{Ric}_{44}\|_{H^N(S_{\ub',u})}\D\ub'\lesssim\Omega^{-3}\ub^{1+\gamma}|u|^{-3-\gamma}.\end{align*}
Applying \eqref{Gronwallub} for $\nablas\tr\chi'$, for $\varepsilon$ sufficiently small, we have
$$|u|^{-1}\|\nablas\tr\chi'\|_{H^N(S_{\ub,u})}\lesssim\Omega^{-3}\ub^{1+\gamma}|u|^{-3-\gamma},$$
and hence
$$|u|^{-1}\|\nablas(\Omega^2\tr\chi)\|_{H^N(S_{\ub,u})}\lesssim\ub^{1+\gamma}|u|^{-3-\gamma},$$
which is the desired bound.

We turn to $\chibh$ and $\tr\chib$, which are based on the equations \eqref{divchibh} and \eqref{Dbtrchib}:
\begin{align*}
&\begin{dcases}\divs(\Omega\chibh)&=\frac{1}{2}\Omega^2\ds\tr \chib'+\Omega\chibh\cdot\eta+\frac{1}{2}\Omega\tr \chib\etab+\Omega(\betab+\frac{1}{2}\Ricset),\\
\Db\tr\chib'&=-\frac{1}{2}(\Omega\tr\chib')^2-|\chibh|^2-\mathbf{Ric}_{33},
\end{dcases}
\end{align*}
We introduce an auxiliary bootstrap assumption for $\nablas^{N+1}\tr\chib$ improving that in \ref{bootstrapB}:
\begin{equation}\label{auxiboottrchib}|u|^{-1}\|\nablas^{N+1}(\Omega^2\tr\chib)\|_{L^2(S_{\ub,u})}\le \ub^{1+\gammat}|u|^{-2-\gammat-(N+1)}\varepsilon^{-\delta}.\end{equation}
Applying \eqref{elliptic1} to the equation $\divs(\Omega\chibh)$, using \eqref{auxiboottrchib} (together with the lower order bounds), and then integrating along $u$, for $\varepsilon$ small enough, we have the desired estimate
\begin{equation}\label{chibhtopbound}\ub^{-2\gamma}\int_{u_0}^u|u'|^{-1+2\gamma}\|\Omega^2\chibh\|_{H^{N+1}(S_{\ub,u'})}^2\D u'\lesssim 1.\end{equation}

The estimate for top order derivatives of $\tr\chib$ is the most crucial estimate in the whole proof. Commute $\nablas$ with the second equation, we have 
$$\Db\nablas(\Omega\tr\chib)=-\Omega\tr\chib\nablas(\Omega\tr\chib)+2\nablas(\omegab\Omega\tr\chib)-2\Omega\chibh\cdot\nablas(\Omega\chibh)-\nablas(\Omega^2\mathbf{Ric}_{33}).$$
We will apply Gronwall \eqref{Gronwallu} to $\nablas(\Omega\tr\chib)$ with $s=1, \nu=2$. The most crucial term is  $\nablas(\omegab\Omega\tr\chib)$. After commuting $\nablas^N$ to the equation, the term $\omegab\nablas^{N+1}(\Omega\tr\chib)$ can be omitted because of the negativity of $\omegab$. So we only need to estimate
\begin{align*}&\int_{u_0}^u\||u'|^2\Omega\tr\chib\nablas\omegab\|_{H^N(S_{\ub,u'})}\D u'\\
\lesssim&\Omega^{-1}\left(\int_{u_0}^u\ub^{2+2\gamma}|u'|^{-1-2\gamma}\D u'\right)^{\frac{1}{2}}\left(\int_{u_0}^u\ub^{-2-2\gamma}|u'|^{5+2\gamma}\|\Omega\tr\chib\nablas(\Omega\omegab)\|^2_{H^N(S_{\ub,u'})}\D u'\right)^{\frac{1}{2}}\\
\lesssim&\Omega^{-1}\ub^{1+\gamma}|u|^{-\gamma}\underline{\mathcal{O}}[\omegab],
\end{align*}
where we have used the lower order estimates for $\omegab$ and $\underline{\mathcal{O}}[\omegab]$ is defined to be
$$\underline{\mathcal{O}}[\omegab]=\sup_{\ub}\int_{\Cb_{\ub}}\ub^{-2-2\gamma}|u|^{1+2\gamma+2(N+1)}|\nablas^{N+1}(\Omega\omegab)|^2.$$
The other terms can be estimated in similar way
\begin{align*}&\int_{u_0}^u\||u'|^2\Omega\chibh\cdot\nablas(\Omega\chibh)\|_{H^N(S_{\ub,u'})}\D u'\\
\lesssim&\Omega^{-1}\left(\int_{u_0}^u\ub^{2+2\gamma}|u'|^{-1-2\gamma}\D u'\right)^{\frac{1}{2}}\left(\int_{u_0}^u\ub^{-2}|u'|^{2}\|\Omega\chibh\|^2_{H^N(\ub,u')}\cdot\ub^{-2\gamma}|u'|^{-1+2\gamma}\|\Omega^2\chibh\|^2_{H^{N+1}(S_{\ub,u'})}\D u'\right)^{\frac{1}{2}}\\
\lesssim&\Omega^{-1}\ub^{1+\gamma}|u|^{-\gamma},
\end{align*}
where we have used \eqref{chibhtopbound}. By \eqref{connection-top-Riccibound},
\begin{align*}&\int_{u_0}^u\||u'|^2\nablas(\Omega^2\mathbf{Ric}_{33})\|_{H^N(S_{\ub,u'})}\D u'\\
\lesssim&\Omega^{-1}\left(\int_{u_0}^u\ub^{2+2\gammat}|u'|^{-1-2\gammat}\D u'\right)^{\frac{1}{2}}\left(\int_{u_0}^u\ub^{-2-2\gammat}|u'|^{3+2\gammat}\|(|u|\nablas)(\Omega^2\mathbf{Ric}_{33})\|^2_{H^N(S_{\ub,u'})}\D u'\right)^{\frac{1}{2}}\\
\lesssim&\Omega^{-1}\ub^{1+\gammat}|u|^{-\gammat}.
\end{align*}
 Here $\gammat>0$ in the power is crucial for the estimates to avoid logarithm loss. Combining all the estimates above, and if $\varepsilon$ is small enough, we have the  bound (recalling that $\gamma>\gammat$),
\begin{equation}\label{trchibtopbound}|u|^{-1}\|(|u|\nablas)(\Omega^2\tr\chib)\|_{H^N(S_{\ub,u})}\lesssim \ub^{1+\gammat}|u|^{-2-\gammat}\underline{\mathcal{O}}[\omegab].\end{equation}
We will show $\underline{\mathcal{O}}[\omegab]\lesssim 1$ below and hence this is the desired bound, improving the auxiliary bootstrap assumption \eqref{auxiboottrchib} and improve bootstrap assumption \ref{bootstrapB} as $1+\gammat>\gamma$.

We then turn to $\eta$. We introduce $\mu=\Ks-\frac{1}{|u|^2}-\divs\eta$ and consider the following system:
\begin{align*}
&\begin{dcases}\divs\eta=\Ks-\frac{1}{|u|^2}-\mu=-\widetilde{\rho+\frac{1}{6}\mathbf{R}}-\frac{1}{4}\widetilde{\tr\chi\tr\chib}+\frac{1}{2}(\chih,\chibh)+\frac{1}{2}\widetilde{\tr\Rics}-\mu,\\
\curls\eta=\sigma-\frac{1}{2}\chih\wedge\chibh,\\
D\mu+\Omega\tr\chi\mu=-\Omega\tr\chi\frac{1}{|u|^2}+\divs(2\Omega\chih\cdot\eta-\Omega\tr\chi\etab)\\\phantom{D\mu+\Omega\tr\chi\mu=}+\nablas(\Omega\Ricsef)\end{dcases}
\end{align*}
By applying \eqref{Gronwallub}, the bootstrap assumptions \ref{bootstrapB} for $\chih$, $\tr\chi$ and the first line of \ref{bootstrapB'} for $\eta$ and $\etab$, the lower order bounds, and \eqref{connection-top-Riccibound} we have, when $\varepsilon$ is small enough,
$$|u|^{-1}\|\Omega^2\mu\|_{H^N(S_{\ub,u})}\lesssim\ub|u|^{-3}, $$
Applying \eqref{elliptic1}, we have
\begin{equation*} \int_0^{\ub}\ub'^{-2-2\gamma}|u|^{1+2\gamma}\|\Omega^2\eta\|^2_{H^{N+1}(S_{\ub',u})}\D\ub'\lesssim1\end{equation*}
and 
\begin{equation}\label{etatopboundCu}\ub^{-1-2\gamma}\int_{u_0}^{u}|u'|^{2\gamma}\|\Omega^2\eta\|^2_{H^{N+1}(S_{\ub,u'})}\D u'\lesssim1\end{equation}
improving \ref{bootstrapB'}. 
For $\etab$,  we introduce $\mub=\Ks-\frac{1}{|u|^2}-\divs\etab$ and consider the following system:
\begin{align*}
&\begin{dcases}
\divs\etab=\Ks-\frac{1}{|u|^2}-\mub=-\widetilde{\rho+\frac{1}{6}\mathbf{R}}-\frac{1}{4}\widetilde{\tr\chi\tr\chib}+\frac{1}{2}(\chih,\chibh)+\frac{1}{2}\widetilde{\tr\Rics}-\mub,\\
\curls\etab=-\sigma+\frac{1}{2}\chih\wedge\chibh,\\
\Db\mub+\Omega\tr\chib\mub=-(\Omega\tr\chib+\frac{2}{|u|})\frac{1}{|u|^2}+\divs(2\Omega\chibh\cdot\etab-\Omega\tr\chib\eta)\\
\phantom{\Db\mub+\Omega\tr\chib\mub=}+\nablas^A(\Omega\Ricset)\end{dcases},
\end{align*}
We apply \eqref{Gronwallu} to $\mub$ with $s=0, \nu=2$. For top order derivatives of $\chibh$, we use \eqref{chibhtopbound}, and for $\tr\chib$, we use the bootstrap assumption \eqref{auxiboottrchib}, together with the lower order bounds, when $\varepsilon$ is small enough, we have
\begin{align*}\||u|\Omega^2\mub\|_{H^N(S_{\ub,u})}\lesssim&\||u_0|\Omega^2\mub\|_{H^N(S_{\ub,u_0})}+\ub|u|^{-1}+ \ub^{\gamma}|u|^{-\gamma}\cdot\ub|u|^{-1}\\
&+\int_{u_0}^u\|\nablas(\Omega^2\eta)\|_{H^{N}(\ub,u')}\D u'+\int_{u_0}^u\ub|u'|^{-1}\|\nablas(\Omega^2\etab)\|_{H^{N}(\ub,u')}\D u'\\&+\int_{u_0}^u\||u'|\nablas(\Omega^2\Ricset)\|_{H^{N}(\ub,u')}\D u'\\
\lesssim&\ub^{\frac{1}{2}+\gamma}|u|^{-\frac{1}{2}-\gamma}
\end{align*}
where we have used \eqref{etatopboundCu} for $\eta$,  the bootstrap assumption \ref{bootstrapB} for $\etab$, and \eqref{connection-top-Riccibound}.
Applying \eqref{elliptic1} to $\etab$, we obtain the desired estimate
\begin{equation*}\int_0^{\ub}\ub'^{-1-2\gamma}|u|^{2\gamma}\|\Omega^2\etab\|^2_{H^{N+1}(S_{\ub',u})}\D\ub'\lesssim1\end{equation*}
improving  \ref{bootstrapB'} and
\begin{equation*}\int_{u_0}^{u}\ub'^{-2\gamma}|u|^{-1+2\gamma}\|\Omega^2\etab\|^2_{H^{N+1}(S_{\ub',u})}\D\ub'\lesssim1\end{equation*}
improving \ref{bootstrapB}.

We turn to $\omegab$. We introduce $\omegabs=\Deltas\omegab-\divs(\Omega\betab)$ and consider the following system:
\begin{align*}
&\begin{dcases}\Deltas\omegab=\omegabs+\divs(\Omega\betab),\\
\phantom{=}D\omegabs+\Omega\tr\chi\omegabs=-2\Omega\chih\cdot\nablas\nablas\omegab-2\divs(\Omega\chih)\cdot\nablas\omegab\\\phantom{=} +\Deltas\left(\Omega^2(2(\eta,\etab)-|\etab|^2-\frac{1}{2}\mathbf{Ric}_{34})\right)-\Deltas(\Omega^2)\rho -\nablas(\Omega^2)\cdot\nablas\rho +\nablas(\Omega^2)\cdot^*\nablas\sigma\\
\phantom{=}-\divs \Omega\left(-\frac{1}{2}\Omega\tr\chi\betab+\Omega\chih \cdot \betab-\Omega\{3\etab\rho -3{}^*\etab\sigma-2\chibh\cdot\beta\}+2\Omega \Ib \right)\\
\phantom{=}-\Omega\tr\chi\divs(\Omega\betab)-2\divs(\Omega\chih\cdot(\Omega\betab)). \end{dcases}
\end{align*}
Observe that the right hand side of the second equation, is estimated by
$$\int_{[0,\ub]\times[u_0,u]}\ub'^{-1-2\gamma}|u'|^{5+2\gamma}\|\Omega(D\omegabs+\Omega\tr\chi\omegabs)\|_{H^{N-1}(\ub',u')}^2\D\ub' \D u'\lesssim1$$
for $\varepsilon$ small enough. We have used the bootstrap assumptions \ref{bootstrapB} for $\omegab$ and \ref{bootstrapB'}  for $\eta$ and $\etab$, the lower order estimates, the curvature estimates \eqref{curvaturebound},  \eqref{connection-top-Riccibound} for $\mathbf{Ric}_{34}$ and \eqref{curvature-weylcurrentbound} for $\Ib$. By applying \eqref{Gronwallub} for $\omegabs$, absorbing the term $\Omega\tr\chi\omegabs$ squaring and integrating over $[u_0,u]$, we have
\begin{align*}&\int_{u_0}^u\ub^{-2-2\gamma}|u'|^{5+2\gamma}\|\Omega\omegabs\|_{H^{N-1}(\ub,u')}^2\D u'\\
\lesssim&\int_{u_0}^u\ub^{-2-2\gamma}|u'|^{5+2\gamma}\left(\int_{0}^{\ub}\|\Omega(D\omegabs+\Omega\tr\chi\omegabs)\|_{H^{N-1}(\ub',u')}\D\ub'\right)^2\D u'\\
\lesssim&\int_{[0,\ub]\times[u_0,u]} \ub'^{-1-2\gamma}|u'|^{5+2\gamma}\|\Omega(D\omegabs+\Omega\tr\chi\omegabs)\|_{H^{N-1}(\ub',u')}^2\D\ub' \D u'\\
\lesssim&1.
\end{align*}
The desired estimate for $\omegab$ follows from applying \eqref{elliptic1}. Then $\underline{\mathcal{O}}[\omegab]\lesssim1$ and the estimate for $\Omega\tr\chib$ in \eqref{trchibtopbound} is done. Similar to the case of $\chih$, the lower order estimates of $\omegab$ is not sufficient. But one can more carefully check the equations for $D\omegab$ to see that $\nablas\omegab$ (in order to apply \eqref{elliptic1}, $\omegab$ itself is not needed) obeys a better estimate 
$$|u|^{-1}\|(|u|\nablas)\omegab\|_{H^{N-1}(S_{\ub,u})}\lesssim\ub^{\frac{3}{2}}|u|^{-\frac{5}{2}}+|u|^{-1}\int_0^{\ub}\|(|u|\nablas)(\Omega^2\mathbf{Ric}_{34})\|_{H^{N-1}(S_{\ub',u})}\D \ub'.$$
Then 
$$\ub^{-2-2\gamma}\int_{u_0}^u |u'|^{1+2\gamma} \|(|u'|\nablas)(\Omega\omegab)\|^2_{H^{N-1}(S_{\ub,u'})}\D u\lesssim1$$
which is sufficient. 

The last thing is $\omega$. We introduce $\omegas=\Deltas\omega+\divs(\Omega\beta)$ and consider the following system:
\begin{align*}
&\begin{dcases}\Deltas\omega=\omegas-\divs(\Omega\beta),\\
\phantom{=}\Db\omegas+\Omega\tr\chib\omegas=-2\Omega\chibh\cdot\nablas\nablas\omega-2\divs(\Omega\chibh)\cdot\nablas\omega\\\phantom{=}+\Deltas \left(\Omega^2(2(\eta,\etab)-|\eta|^2-\frac{1}{2}\mathbf{Ric}_{34})\right)-\Deltas(\Omega^2)\rho-\nablas(\Omega^2)\cdot\nablas\rho +\nablas(\Omega^2)\cdot^*\nablas\sigma\\
\phantom{=}+\divs \Omega\left(-\frac{1}{2}\Omega\tr\chib\beta+\Omega\chibh \cdot \beta+\Omega\{3\eta\rho +3{}^*\eta\sigma+2\chih\cdot\betab\}-2\Omega I \right)\\
\phantom{=}+\Omega\tr\chib\divs(\Omega\beta)+2\divs(\Omega\chibh\cdot(\Omega\beta)). \end{dcases}\\
\end{align*}
A direct computation shows similarly (using in addition \ref{bootstrapB} for $\omega$) that
$$\int_{[0,\ub]\times[u_0,u]}\ub'^{-1-2\gamma}|u'|^{5+2\gamma}\|\Omega(\Db\omegas+\Omega\tr\chib\omegas)\|_{H^{N-1}(\ub',u')}^2\D\ub' \D u'\lesssim1$$
when $\varepsilon$ is small enough. Applying \eqref{Gronwallu} to $\omegas$  with $s=0, \nu=2$, squaring and integrating over $[u_0,u]$, we have
\begin{align*}&\int_{0}^{\ub}\ub'^{-1-2\gamma}|u|^{2+2\gamma}\||u|\Omega\omegas\|_{H^{N-1}(\ub',u)}^2\D \ub'\\
\lesssim&1+\int_{0}^{\ub}\ub'^{-1-2\gamma}|u|^{2+2\gamma}\left(\int_{u_0}^{u}\||u'|\Omega(\Db\omegas+\Omega\tr\chib\omegas)\|_{H^{N-1}(\ub',u')}\D u'\right)^2\D \ub'\\
\lesssim &1+\int_{[0,\ub]\times[u_0,u]}\ub'^{-1-2\gamma}|u'|^{5+2\gamma}\|\Omega(\Db\omegas+\Omega\tr\chib\omegas)\|_{H^{N-1}(\ub',u')}^2\D\ub' \D u'\\
\lesssim&1.
\end{align*}
The desired estimate for $\omega$ can be obtained by applying \eqref{elliptic1} and improved lower order estimates for $\omega$, as done in the cases of $\chih$ and $\omegab$.

We turn to mixed derivatives $\nablas^i\Db^j$. We should first obtain estimates for $\nablas^i\Db^j\omegab$ for $i+j=N+1$ and $j\le N-1$ and hence $i\ge2$. We commute $\Db^j$ to the equations $\Deltas\omegab$ used before, 
\begin{align*}
&\begin{dcases}\Deltas\Db^j\omegab=\Db^j\omegabs+\divs\Db^j(\Omega\betab)+\text{commutators},\\
\phantom{=}D\Db^j\omegabs+\sum_{j_1+j_2=j}\Db^{j_1}\Omega\tr\chi\Db^{j_2}\omegabs=\Db^j\underline{m}+[D,\Db^j]\omegas \end{dcases}
\end{align*}
where
\begin{align*}
\underline{m}=&-2\Omega\chih\cdot\nablas\nablas\omegab-2\divs(\Omega\chih)\cdot\nablas\omegab \\\phantom{=} +&\Deltas\left(\Omega^2(2(\eta,\etab)-|\etab|^2-\frac{1}{2}\mathbf{Ric}_{34})\right)-\Deltas(\Omega^2)\rho -\nablas(\Omega^2)\cdot\nablas\rho +\nablas(\Omega^2)\cdot^*\nablas\sigma\\
\phantom{=}&-\divs \Omega\left(-\frac{1}{2}\Omega\tr\chi\betab+\Omega\chih \cdot \betab-\Omega\{3\etab\rho -3{}^*\etab\sigma-2\chibh\cdot\beta\}+2\Omega \Ib \right)\\
 \phantom{=}&-\Omega\tr\chi\divs(\Omega\betab)-2\divs(\Omega\chih\cdot(\Omega\betab)).
\end{align*}
As in deriving estimates for $\nablas^{N+1}\omegab$, to estimate $\nablas^i\Db^j\omegab$, we integrate the transport equation and apply elliptic estimate. The top order terms we encounter are the following:
$$\nablas^{i-1}\Db^j(\betab, \rho,\sigma,\beta,\Ib), \nablas^i\Db^j \mathbf{Ric}_{34},  \nablas^i\Db^j(\eta,\etab).$$
The first and second groups of terms are simply bounded using previously derived bounds, especially \eqref{curvaturenablasDbalphab}, which has not been used before.  Note that we express $\Db^j(\beta,\rho,\sigma)$ in terms of angular derivatives of the other curvature components and $\nablas$ and $\Db$ mixed derivatives of $\alphab$, and fluid terms (the Weyl current).  For $\nablas^i\Db^j(\eta,\etab)$, we instead use the bootstrap assumptions \ref{bootstrapB}. Finally, we will have 
\begin{equation}\label{nablasiDbjtopomegab}\sup_{\ub}\int_{\Cb_{\ub}}\ub^{-2-2\gamma}|u|^{1+2\gamma+2(N+1)}|\nablas^i\Db^j\omegab|^2\lesssim1,\ i+j=N+1, j\le N-1\end{equation}
and the details will be omitted.  We remark that it is not possible to derive estimate for $\nablas\Db^N\omegab$ and $\Db^{N+1}\omegab$ so we only work for $j\le N-1$. 

We turn to the other connection coefficients. Their estimates can be seen directly from the corresponding $\Db$ equations. Module the lower order terms and matter field terms, we have the following relations: 
$$\nablas^i\Db^j(\Omega\chih, \Omega\tr\chi)\simeq \Omega\nablas^{i+1}\Db^{j-1}(\Omega\eta)\simeq\Omega^2\nablas^{i+2}\Db^{j-2}\omegab +\Omega\nablas^{i+1}\Db^{j-2}(\Omega^2\betab),$$
$$\nablas^i\Db^j(\Omega\chibh)\simeq \nablas^i\Db^{j-1}(\Omega^2\alphab),$$
$$\nablas^i\Db^j(\Omega\eta)\simeq\Omega\nablas^{i+1}\Db^{j-1}\omegab +\Omega\nablas^{i}\Db^{j-1}(\Omega^2\betab),$$
$$\nablas^i\Db^j(\Omega\etab)\simeq \Omega\nablas^{i}\Db^{j-1}(\Omega^2\betab),$$
$$\nablas^i\Db^j\omega\simeq \Omega\nablas^{i}\Db^{j-1}(\Omega^2\rho),$$
and $\nablas^i\Db^j(\Omega\tr\chib)$ is of lower order. We can see that no estimates can be established for $ \Db^{N+1}\eta$.  For $\Db^ND$ derivatives, we use instead the following relations:
$$\Db^ND(\Omega\chih)\simeq \Db^N(\Omega^2\alpha),$$
$$\Db^ND(\Omega\chibh)\simeq \Omega\nablas\Db^N \etab\simeq \Omega\nablas\Db^{N-1}(\Omega^2\betab),$$
$$\Db^ND(\Omega\eta)\simeq\Db^N(\Omega^2\beta),$$
$$\Db^ND(\Omega\etab)\simeq \nablas\Db^N\omega+ \Db^N(\Omega^2\beta)\simeq\nablas\Db^{N-1}(\Omega^2\rho)+\Db^N(\Omega^2\beta),$$
$$\Db^ND\omega\simeq \Db^{N-1}D(\Omega^2\rho)\simeq \Omega\nablas\Db^{N-1}(\Omega^2\beta),$$
and $\Db^{N+1}D(\Omega\tr\chi, \Omega\tr\chib)$ is of lower order. $\Db^N(\Omega^2\alpha)$ can also be expressed in terms of angular derivatives of the other curvature components and $\nablas^i\Db^j$ derivatives of $\alphab$. The $\nablas\Db^{N-1}D$ estimates can be derived similarly.

\end{proof}

\section{The trapped surface formation}\label{sec:trapped}

Finally we are ready to prove Theorem \ref{main} based on the estimates established in Theorem \ref{existencetheorem}, Propositions \ref{fluidestimates}, \ref{curvature}, \ref{improvebootA}, \ref{improvebootB}. Recall that we have the estimate \eqref{Omegaupperlower} for $\Omega_0$ along $\Cb_0$. We only need the upper bound
\begin{equation}\label{Omegaupper}
\Omega_0^2(u)\le c_\beta|u|^\beta.
\end{equation}
Here the positive constant $\beta>0$ should not be confused with the component $\beta$ of the Weyl curvature. Rewrite the equation \ref{Dbchih} as, omitting the coefficients,
\begin{equation}\label{partialuchih}\frac{\partial}{\partial u}(|u|^2|\Omega\chih|^2)=|u|^2\left(\Omega^2\nablas \tensor \eta + \Omega\eta \tensor \Omega\eta +\widetilde{\Omega\tr\chib}\Omega\chih+\Omega\tr\chi \Omega\chibh+\Omega^2\widehat{\Rics}-\nablas_b(\Omega\chih)\right)\cdot\Omega\chih,\end{equation}
where $b$ is a tangental vectorfield in \eqref{doublenullmetric}, verifying the equation
$$Db=-4\Omega^2\zeta^\sharp.$$
We can assign that $b=0$ on $\Cb_0$  in constructing the double null coordinates, then integrating the above equation we have
$$|b|\lesssim \Omega\ub^2|u|^{-2}.$$
Integrating equation \eqref{partialuchih}, employing the estimates in Propositions \ref{fluidestimates}, \ref{improvebootA}, and estimate for $b$ above and Sobolev inequalities because we work in $L^\infty$, we have
$$\left||u|^2||\Omega\chih|^2(\ub,u,\vartheta)-|u_0|^2|\Omega\chih|^2(\ub, u_0,\vartheta)\right|\lesssim \ub|u|^{-1}.$$
and hence
\begin{equation}\label{energydifference}\left|\int_0^{\ub}|u|^2||\Omega\chih|^2(\ub',u,\vartheta)\D\ub'-\int_0^{\ub}|u_0|^2|\Omega\chih|^2(\ub', u_0,\vartheta)\D\ub'\right|\le  c \ub^2|u|^{-2}\cdot|u|\end{equation}
where $c$ is a constant depending on $A, \kappa, u_0,p$ in Theorem \ref{existencetheorem}. Assume that the perturbations $\chih_t(\ub,u_0,\vartheta)$ takes the form $t\ub^p$, or writing in terms of a condition on the  energy
\begin{equation}\label{lowerboundcondition}\int_0^{\ub}|u_0|^2|\Omega\chih_t|^2(\ub', u_0,\vartheta)\D\ub'\ge t^2\ub^{2p+1},\end{equation}
which admits smooth $\chih_t$ verifying the assumptions of Theorem \ref{existencetheorem}. By \eqref{energydifference} we have
$$\int_0^{\ub}|u|^2|\Omega\chih_t|^2(\ub', u,\vartheta)\D\ub'\ge t^2\ub^{2p+1}-c\varepsilon^2|u|,$$
On the other hand, in order that a trapped surface can form, we look at the equation \eqref{Dtrchi}. Integrating this equation along $C_u$ we have
$$\tr\chi'_t(\ub,u,\vartheta)\le \tr\chi'(0,u,\vartheta)-\int_0^{\ub}|\chih_t|^2(\ub',u,\vartheta)\D\ub'$$
where we have dropped other positive terms on the right hand side of \eqref{Dtrchi}. By \eqref{lapse} we have
\begin{align*}\int_0^{\ub}|\chih_t|^2(\ub',u,\vartheta)\D\ub'\ge& \frac{1}{4}\Omega_0^{-2}(u)|u|^{-2}\int_0^{\ub}|u|^2|\Omega\chih_t|^2(\ub',u,\vartheta)\D\ub'\\
\ge&\frac{1}{4}c_\beta^{-1}|u|^{-(\beta+2)}\left(t^2\ub^{2p+1}-c\ub^2|u|^{-2}\cdot|u|\right).\end{align*}
In order that $\tr\chi'_t(\ub,u,\vartheta)<0$, in view of initially $\tr\chi'(0,u,\vartheta)\le A|u|^{-1}$, we need
$$\frac{1}{4}c_\beta^{-1}|u|^{-\beta}\left(t^2\ub^{2p}\cdot\ub|u|^{-1}-c\ub^2|u|^{-2}\right)> A.$$
Now take $p,q$ such that $0<p<q< \frac{1}{2}\frac{\beta}{2+\beta}$. For each $t\ne0$ (and sufficiently small), we choose $\ub_{1,t}>0, u_{1,t}\in(u_0,0)$ so that
$$\ub_{1,t}|u_{1,t}|^{-1}=\ub_{1,t}^{2q},$$
then if $\ub_{1,t}$ is sufficiently small, 
\begin{align*}
&4^{-1}c_\beta^{-1}|u_{1,t}|^{-\beta}\left(t^2\ub_{1,t}^{2p}\cdot\ub_{1,t}|u_{1,t}|^{-1}-c\ub_{1,t}^2|u_{1,t}|^{-2}\right)\\
\ge& 4^{-1}c_\beta^{-1}\ub_{1,t}^{-\beta(1-2q)}\left(t^2\ub_{1,t}^{2p}\cdot\ub_{1,t}^{2q}-c\ub_{1,t}^{2q}\cdot\ub_{1,t}^{2q}\right)\\
\ge& 8^{-1}c_\beta^{-1}\ub_{1,t}^{-\beta(1-2q)}\cdot t^2\ub_{1,t}^{4q}\\
=&8^{-1}c_\beta^{-1}t^2\ub_{1,t}^{-\beta+2(2+\beta)q}
\end{align*}
The power $-\beta+2(2+\beta)q<0$ so we can choose $\ub_{1,t}$ small enough so that the above quantity is larger than $A$, and hence $\tr\chi'_t(\ub,u,\vartheta)<0$. $\Omega\tr\chib(\ub,u,\vartheta)<0$ follows directly from  \eqref{trchib<0}. We conclude that, $S_{\ub_{1,t}, u_{1,t}}$ is trapped.
 
 It is clear that we can find smooth $\chih_t(\ub,u_0,\vartheta)$ verifying \eqref{lowerboundcondition} so that $\chi_t\to0$ in $C^p_{\ub}$.
 
\appendix

\section{}\label{appendix}

\subsection{Double null foliation}
	We list all equations we need, which are written in double null frames.  The notation and setup are taken  from  \cite{Li-Li26, Li-Liu22}, following Christodoulou's monograph \cite{Chr08}, which may be slightly different from those in other literature.  We assume the solution $(\mathcal{M},g)$ is foliated by two optical functions, $u$ and $\ub$, that is
 $$g(\nabla u,\nabla u)=g(\nabla\ub,\nabla\ub)=0.$$
 In addition, we require that $u$ and $\ub$  increase towards the future. We use $C_u$ to denote the  null hypersurfaces that are the level sets of $u$ and use ${\Cb}_{\ub}$ to denote the incoming null hypersurfaces that are the level sets of $\ub$. We denote the intersection $S_{\ub,u}=\Cb_{\ub} \cap C_u$, which is a  space-like two-sphere. The space-time metric $g$ induces a Riemannian metric $\gs$ on $S_{\ub,u}$ and $\epsilons$ is the volume form of $\gs$ on $S_{\ub,u}$. We use   $\nablas$ to denote the  covariant derivative (with respect to $\gs$) on $S_{\ub,u}$ and $\Ks$ is the Gauss curvature of $S_{\ub,u}$.

The lapse function $\Omega$ is defined by the formula
$$ \Omega^{-2}=-2g(\nabla\ub,\nabla u).$$
  We then define the normalized null pair $(\Lbh, \Lh)$ by
  $$e_3=\Lbh=-2\Omega\nabla\ub,\ e_4=\Lh=-2\Omega\nabla u.$$
  We also define one another null pair
  $$\Lb=\Omega \Lbh,\ L=\Omega \Lh.$$
 The flows generated by $\Lb$ and $L$ preserve the double null foliation. Moreover, we define
 $$\Lb'=-2\nabla\ub=\Omega^{-1}\Lbh,\ \Lb'=-2\nabla u=\Omega^{-1}\Lh$$
 which is geodesic. 

We can introduce the double null coordinate system $(\ub,u,\vartheta^A)$ on $\mathcal{M}$, where $A,B,C,\cdots$ are denoted an index from $1$ to $2$. In such a coordinate system, the Lorentzian metric $g$ takes the following form
\begin{align}\label{doublenullmetric}
g=-2\Omega^2(\D\ub\otimes\D u+\D u\otimes\D \ub)+\gs_{AB}(\D\vartheta^A-b^A\D u)\otimes(\D\vartheta^B-b^B\D u),
\end{align}
where the vectorfield $b=b^A\partial_{\vartheta^A}$ is tangent to $S_{\ub,u}$. The null vectors $\Lb$ and $L$ can be computed as $\Lb=\partial_u+b^A\partial_{\vartheta^A}$ and $L=\partial_{\ub}$.

We call $\psi$ to be a tangential tensorfield if $\psi$ is \textit{a priori} a tensorfield defined on the space-time $\mathcal{M}$ and all the possible contractions of $\psi$ with either $\Lbh$ or $\Lh$ are zeros. By choosing a tangential frame $(e_1,e_2)$, which is tangent to $S_{\ub,u}$, each tangential tensorfield can be express as, for example $\psi_A, \psi_{AB}, \psi^A,\cdots$. We use $D\psi$ and $\Db\psi$ to denote the projection to $S_{\ub,u}$ of usual Lie derivatives $\mathcal{L}_L\psi$ and $\mathcal{L}_{\Lb}\psi$, which are again tangential tensorfields. The notation $\nablas\psi$ also makes sense when $\psi$ is tangential. We also denote $\nablas_X\psi$ to be the projection of the spacetime covariant derivative $\nabla_X\psi$ to $S_{\ub,u}$, where $X$ is not necessarily a tangential vectorfield.

Using the null frame $(e_1,e_2,\Lbh,\Lh)$, the connection coefficients can be decomposed as the following tangential tensorfields:
\begin{align*}
\chi_{AB}&=g(\nabla_A\Lh,e_B),\quad \eta_A=-\frac{1}{2}g(\nabla_{\Lbh}e_A,\Lh),\quad \omega=\frac{1}{2}\Omega g(\nabla_{\Lh}\Lbh,\Lh),\\
\chib_{AB}&=g(\nabla_A\Lbh,e_B), \quad\etab_A=-\frac{1}{2}g(\nabla_{\Lh}e_A,\Lbh), \quad\omegab=\frac{1}{2}\Omega g(\nabla_{\Lbh}\Lh,\Lbh).
\end{align*}
We also define the following normalized quantities:
$$\chi'=\Omega^{-1}\chi,\ \chib'=\Omega^{-1}\chi,\ \zeta=\frac{1}{2}(\eta-\etab).$$
 The ($2$-dimensional) trace of $\chi$ and $\chib$ are denoted by
 $$\tr\chi = \gs^{AB}\chi_{AB},\ \tr\chib = \gs^{AB}\chib_{AB},$$
 and the trace-free parts of $\chi$ and $\chib$ are denoted by
 $$\chih=\chi-\frac{1}{2}\tr\chi\gs,\ \chibh=\chib-\frac{1}{2}\tr\chib\gs.$$ 
  By definition, we can check directly the following useful identities :
  $$\ds\log\Omega=\frac{1}{2}(\eta+\etab),\ D\log\Omega=\omega,\ \Db\log\Omega=\omegab.$$

Recall that the Weyl curvature tensor is
$$\mathbf{W}_{\alpha\beta\gamma\delta}=\mathbf{R}_{\alpha\beta\gamma\delta}+\frac{\mathbf{R}}{6}(g_{\alpha\gamma}g_{\beta\delta}-g_{\alpha\delta}g_{\beta\gamma})+\frac{1}{2}(g_{\alpha\delta}\mathbf{Ric}_{\beta\gamma}+g_{\beta\gamma}\mathbf{Ric}_{\alpha\delta}-g_{\alpha\gamma}\mathbf{Ric}_{\beta\delta}-g_{\beta\delta}\mathbf{Ric}_{\alpha\gamma}).$$
It can be decomposed as
\begin{align*}
\alpha_{AB}&=\mathbf{W}(e_A,\Lh,e_B,\Lh),\quad\beta_A=\frac{1}{2}\mathbf{W}(e_A,\Lh,\Lbh,\Lh),\quad\rho=\frac{1}{4}\mathbf{W}(\Lbh,\Lh,\Lbh,\Lh),\\
\alphab_{AB}&=\mathbf{W}(e_A,\Lbh,e_B,\Lbh),\quad\betab_A=\frac{1}{2}\mathbf{W}(e_A,\Lbh,\Lbh,\Lh),\quad\sigma=\frac{1}{4}\mathbf{W}(\Lbh,\Lh,e_A,e_B)\epsilons^{AB}.
\end{align*}
By the algrebric property of Weyl tensor, $\alpha$ and $\alphab$ are trace-free.

\subsection{Equations}

We first define several kinds of contraction of the tangential tensorfields. Let $\theta$ be a symmetric $2$-tensorfield and $\xi$ be a tangential $1$-form. We denote
$$(\theta_1,\theta_2)=\gs^{AC}\gs^{BD}(\theta_1)_{AB}(\theta_2)_{CD},\ \ (\xi_1,\xi_2)=\gs^{AB}(\xi_1)_A(\xi_2)_B,$$ 
and
$$|\theta|^2=(\theta,\theta),\ |\xi|^2=(\xi,\xi).$$ 
We also denote the contraction
\begin{align*}(\theta\cdot\xi)_A=\theta_A{}^B\xi_B&,\ (\theta_1\cdot \theta_2)_{AB}=(\theta_1)_A{}^C(\theta_2)_{CB},\\
 \theta_1 \wedge\theta_2=\epsilons^{AC}\gs^{BD} (\theta_1)_{AB}(\theta_2)_{CD}&,\ \xi_1\tensor \xi_2=\xi_1\otimes\xi_2+\xi_2\otimes\xi_1-(\xi_1,\xi_2)\gs.
\end{align*} The Hodge dual for $\xi$ is defined by $\prescript{*}{}\xi_A=\epsilons_A{}^C\xi_C$. We also denote 
$$\divs\xi=\nablas^A\xi_A,\ \curls\xi_A=\epsilons^{AB}\nablas_A\xi_B,\ (\divs\theta)_A=\nablas^B\theta_{AB},$$ 
and
$$(\nablas\tensor\xi)_{AB}=(\nablas\xi)_{AB}+(\nablas\xi)_{BA}-\divs\xi \,\gs_{AB}.$$
If $\theta$ is trace-free, $\Dh\theta$ and $\Dbh\theta$ refer to the trace-free part of $D\theta$ and $\Db\theta$. 


Let us define the following tangential tensorfields related to the spacetime Ricci tensor:
\begin{align*}
(\Ricset)_{A}=\mathbf{Ric}(e_A, e_3),\  (\Ricsef)_{A}&=\mathbf{Ric}(e_A, e_4),\ \Rics_{\!\!AB}=\mathbf{Ric}(e_A, e_B),\\
 \mathbf{Ric}_{34}=\mathbf{Ric}(e_3,e_4), \  \mathbf{Ric}_{44}&=\mathbf{Ric}(e_4,e_4), \ \mathbf{Ric}_{33}=\mathbf{Ric}(e_3,e_3).
\end{align*}
Then we have the following \textit{null structure equations}:
\begin{align}
\label{Dbchibh}\Dbh\chibh'&=-\alphab,\\
\label{Dbtrchib}\Db(\Omega\tr\chib)&=-\frac{1}{2}(\Omega\tr\chib)^2+2\omegab\Omega\tr\chib-|\Omega\chibh|^2-\Omega^2\mathbf{Ric}_{33},\\
\label{Dchih}\Dh\chih'&=-\alpha,\\
\label{Dtrchi}D\tr\chi'&=-\frac{1}{2}(\Omega\tr\chi')^2-|\chih|^2-\mathbf{Ric}_{44},\\
\label{Dbchih}\Dbh(\Omega\chih)&=\Omega^2(\nablas \tensor \eta + \eta \tensor \eta +\frac{1}{2}\tr\chib\chih-\frac{1}{2}\tr\chi \chibh+\frac{1}{2}\widehat{\Rics}),\\
\label{Dchibh}\Dh(\Omega\chibh)&=\Omega^2(\nablas \tensor \etab + \etab \tensor \etab +\frac{1}{2}\tr\chi\chibh-\frac{1}{2}\tr\chib \chih+\frac{1}{2}\widehat{\Rics}),\\
\label{Dbtrchi}\Db(\Omega\tr\chi)&=\Omega^2(2\divs\eta+2|\eta|^2-\tr\chi\tr\chib-2\Ks+\tr\Rics),\\
\label{Dtrchib}D(\Omega\tr\chib)&=\Omega^2(2\divs\etab+2|\etab|^2-\tr\chi\tr\chib-2\Ks+\tr\Rics),\\
\label{Deta}D\eta &= (\Omega\chi)\cdot\etab-\Omega(\beta+\frac{1}{2}\Ricsef),\\
\label{Db-etab}\Db\etab &= (\Omega\chib) \cdot\eta+\Omega(\betab-\frac{1}{2}\Ricset),\\
\label{Domegab}D  \omegab &=\Omega^2(2(\eta,\etab)-|\eta|^2-(\rho+\frac{1}{6}\mathbf{R}+\frac{1}{2}\mathbf{Ric}_{34})),\\
\label{Dbomega}\Db  \omega &=\Omega^2(2(\eta,\etab)-|\etab|^2-(\rho+\frac{1}{6}\mathbf{R}+\frac{1}{2}\mathbf{Ric}_{34})),\\
\label{Gauss}\Ks+\frac{1}{4}\tr \chi\tr\chib-\frac{1}{2}(\chih,\chibh)&=-(\rho+\frac{1}{6}\mathbf{R})+\frac{1}{2}\tr\Rics,\\
\label{curleta}\curls\eta&=\sigma-\frac{1}{2}\chih\wedge\chibh,\\
\label{divchih}\divs(\Omega\chih)&=\frac{1}{2}\Omega^2\ds\tr \chi'+\Omega\chih\cdot\etab+\frac{1}{2}\Omega\tr \chi\eta-\Omega(\beta-\frac{1}{2}\Ricsef),\\
\label{divchibh}\divs(\Omega\chibh)&=\frac{1}{2}\Omega^2\ds\tr \chib'+\Omega\chibh\cdot\eta+\frac{1}{2}\Omega\tr \chib\etab+\Omega(\betab+\frac{1}{2}\Ricset).
\end{align}
Note that in Einstein--Euler system, we have
$$\mathbf{Ric}_{\alpha\beta}=2(1+\kappa)\varrho U_\alpha U_\beta+(1-\kappa){\varrho}g_{\alpha\beta},$$
and
$$\mathbf{R}=-2\mathbf{T}=2(1-3\kappa)\varrho.$$

The contracted second Bianchi identity will accordingly be the following inhomogeneous equation:
\begin{equation*}
\nabla^{\alpha}\mathbf{W}_{\alpha\beta\gamma\delta}=\nabla_{[\gamma}\mathbf{R}_{\delta]\beta}+\frac{1}{6}g_{\beta[\gamma}\nabla_{\delta]}\mathbf{R}\triangleq\mathbf{J}_{\beta\gamma\delta}.
\end{equation*}
We decompose the Weyl current $\mathbf{J}$ as
\begin{align*}\Xi_{A}=\frac{1}{2}\mathbf{J}_{44A},& \ \Xib_{A}=\frac{1}{2}\mathbf{J}_{33A}\\
\Lambda=\frac{1}{4}\mathbf{J}_{434},&\ \Lambdab=\frac{1}{4}\mathbf{J}_{343}\\
K=\frac{1}{4}\epsilons^{AB}\mathbf{J}_{4AB},&\ \Kb=\frac{1}{4}\epsilons^{AB}\mathbf{J}_{3AB}\\
I_A=\frac{1}{2}\mathbf{J}_{34A},&\ \Ib_A=\frac{1}{2}\mathbf{J}_{43A}\\
\Theta_{AB}=\frac{1}{2}(\mathbf{J}_{A4B}+\mathbf{J}_{B4A}-\gs_{AB}\gs^{CD}\mathbf{J}_{C4D}),&\ 
\Thetab_{AB}=\frac{1}{2}(\mathbf{J}_{A3B}+\mathbf{J}_{B3A}-\gs_{AB}\gs^{CD}\mathbf{J}_{C3D}).
\end{align*}

 This equation can be decomposed using the null frame into components, which we call \textit{null Bianchi equations}:
\begin{align}
\label{Dbalpha}&\Dbh\alpha-\frac{1}{2}\Omega\tr\chib \alpha+2\omegab\alpha+\Omega\{-\nablas\tensor\beta -(4\eta+\zeta)\tensor \beta+3\chih \rho+3{}^*\chih \sigma\}=-2\Omega\Theta,\\
\label{Dbeta}&D\beta+\frac{3}{2}\Omega\tr\chi\beta-\Omega\chih\cdot\beta-\omega\beta-\Omega\{\divs\alpha+(\etab+2\zeta)\cdot\alpha\}=2\Omega\Xi,\\
\label{Dbbeta}&\Db\beta+\frac{1}{2}\Omega\tr\chib\beta-\Omega\chibh \cdot \beta+\omegab \beta-\Omega\{\ds \rho+{}^*\ds \sigma+3\eta\rho+3{}^*\eta\sigma+2\chih\cdot\betab\}=-2\Omega I\\
\label{Drho}&D\rho+\frac{3}{2}\Omega\tr\chi \rho-\Omega\{\divs \beta+(2\etab+\zeta,\beta)-\frac{1}{2}(\chibh,\alpha)\}=-2\Omega\Lambda,\\
\label{Dsigma}&D\sigma+\frac{3}{2}\Omega\tr\chi\sigma+\Omega\{\curls\beta+(2\etab+\zeta,{}^*\beta)-\frac{1}{2}\chibh\wedge\alpha\}=-2\Omega K,\\
\label{Dbetab}&D\betab+\frac{1}{2}\Omega\tr\chi\betab-\Omega\chih \cdot \betab+\omega \betab+\Omega\{\ds \rho-{}^*\ds \sigma+3\etab\rho-3{}^*\etab\sigma-2\chibh\cdot\beta\}=2\Omega \Ib\\
\label{Dbrho}&\Db\rho+\frac{3}{2}\Omega\tr\chib \rho+\Omega\{\divs \betab+(2\eta-\zeta,\betab)+\frac{1}{2}(\chih,\alphab)\}=-2\Omega\Lambdab,\\
\label{Dbsigma}&\Db\sigma+\frac{3}{2}\Omega\tr\chib\sigma+\Omega\{\curls\betab+(2\eta-\zeta,{}^*\betab)+\frac{1}{2}\chih\wedge\alphab\}=2\Omega\Kb,\\
\label{Dbbetab}&\Db\betab+\frac{3}{2}\Omega\tr\chib\betab-\Omega\chibh\cdot\betab-\omegab\betab+\Omega\{\divs\alphab+(\eta-2\zeta)\cdot\alphab\}=-2\Omega\Xib,\\
\label{Dalphab}&\Dh\alphab-\frac{1}{2}\Omega\tr\chi \alphab+2\omega\alphab+\Omega\{\nablas\tensor\betab +(4\etab-\zeta)\tensor \betab+3\chibh \rho-3{}^*\chibh \sigma\}=-2\Omega\Thetab.
\end{align}

\subsection{Commutation formulas}
\begin{lemma}\label{commutationformulas}
Given integer $i$ and tangential tensorfield $\phi$. we have schematically 
\begin{align*}
[D,\nablas^i]\phi&=\sum_{j=1}^i\nablas^j(\Omega\chi)\cdot\nablas^{i-j}\phi,\\
[\Db,\nablas^i]\phi&=\sum_{j=1}^i\nablas^j(\Omega\chib)\cdot\nablas^{i-j}\phi,
\end{align*}
and
\begin{align*}
[\nablas,\nablas^i]\phi&=\sum_{j=1}^i\nablas^{j-1}\Ks\cdot\nablas^{i-j}\phi,
\end{align*}
where the first $\nablas$ represents all possible first order differential operators acting on $\phi$, like divergence or curl. Here we use ``$\cdot$'' to represent some contraction with the coefficients by $\gs$ or $\epsilons$. In addition, if $\phi$ is a function, then when $i=1$, all commutators above are zero; when $i\ge2$, all $i$'s are replaced by $i-1$'s in above formulas. Finally, we also have
\begin{align*}
[\Lb, L]&=4\Omega^2\zeta^\sharp,\\
[\Db, D]\phi&=\Lie_{4\Omega^2\zeta^\sharp}\phi.
\end{align*}
where $\Lie$ is the projection to $S_{\ub,u}$ of the usual Lie derivative. 
 \end{lemma}

\end{document}